\documentclass[12pt,reqno]{amsart}
\usepackage{lmodern}
\usepackage[T1]{fontenc}
\usepackage[numbers]{natbib}
\usepackage{amsmath,amssymb,amsfonts,amsthm,mathrsfs}
\usepackage{geometry}
\usepackage{microtype}
\usepackage{graphicx,multirow,booktabs}
\usepackage[flushleft]{threeparttable}
\usepackage{float}
\usepackage{algorithm}
\usepackage{algpseudocode}
\usepackage{xcolor}
\usepackage{tikz,pgfplots}
\usepgfplotslibrary{groupplots}
\usetikzlibrary{positioning}
\usepackage{caption}
\usepackage[colorlinks,linkcolor=blue,citecolor=magenta]{hyperref}
\newcommand{\mFA}{\text{F1}}
\newcommand{\mFB}{\text{F2}}
\newcommand{\mFC}{\text{F3}}
\let\cline\cmidrule
\newtheorem{theorem}{Theorem}

\newtheorem{corollary}[theorem]{Corollary}
\newtheorem{lemma}{Lemma}
\theoremstyle{definition}

\definecolor{LBbar0}{HTML}{E31A1C}
\definecolor{LBbar1}{HTML}{E69F00}
\definecolor{LBbar2}{HTML}{009E73}
\definecolor{LBbar3}{HTML}{CC79A7}
\definecolor{LBbar4}{HTML}{92C5DE}
\definecolor{LBbar5}{HTML}{2166AC}
\definecolor{LBbar6}{HTML}{F4A582}
\definecolor{LBbar7}{HTML}{B35806}
\definecolor{LBbar8}{HTML}{C2A5CF}
\definecolor{LBbar9}{HTML}{762A83}
\title[Identical Machine Scheduling with Conflicts to Minimize Total Completion Time]{Scheduling on identical machines with conflicts to minimize the total completion time}
\author{Nour Elhouda Tellache}
\address{N.E.H. Tellache, Decision Support \& Operations Research Group, Department of Informatics, University of Fribourg, Fribourg, Switzerland}
\email{nourelhouda.tellache@unifr.ch}
\author{Lydia Aoudia}
\address{L. Aoudia, RECITS Laboratory, USTHB university, BP 32 El-Alia, Bab-Ezzouar, Algiers, Algeria}
\email{aoudialydia96@gmail.com}
\author{Mourad Boudhar}
\address{M. Boudhar, RECITS Laboratory, USTHB university, BP 32 El-Alia, Bab-Ezzouar, Algiers, Algeria}
\email{mboudhar@yahoo.fr}
\date{}
\begin{document}
\begin{abstract}
This paper addresses the problem of scheduling jobs on identical machines with conflict constraints, where certain jobs cannot be scheduled simultaneously on different machines. We focus on the case where conflicts can be represented by a simple undirected graph, and the objective is to minimize the total completion time.  We establish that the problem is strongly NP-hard even on two machines with two distinct processing times. For unit-time jobs, we prove that the problem is polynomially solvable on two machines for any conflict graph, but becomes strongly NP-hard on three machines. We further identify polynomially solvable cases for specific classes of conflict graphs. For the general problem, we propose mathematical models, lower bounds, and a genetic algorithm. We evaluate their performance through computational experiments on a wide range of instances derived from well-known benchmark instances in the literature.
\end{abstract}
\maketitle

\section{Introduction}
We consider the following scheduling problem. Let \( J = \{J_1, \ldots, J_n\} \) be a set of \( n \) jobs, where each job \( J_j \in J \) has a processing time \( p_j \in \mathbb{N}_0 \). The jobs are to be scheduled on a set \( M = \{M_1, \ldots, M_m\} \) of \( m \) identical machines, where each machine can process at most one job at a time. Jobs are subject to conflict constraints given by a simple, undirected conflict graph \( G = (V, E) \), where each vertex \( j \in V \) corresponds to a job \( J_j \). An edge \( \{j, \ell\} \in E \) indicates that jobs \( J_j \) and \( J_\ell \) are in conflict and cannot be processed simultaneously on different machines at any point in time. The objective is to find a feasible schedule that minimizes the total completion time (TCT). Let \( C_j \) denote the completion time of job \( J_j \). In the three-field notation of \citet{GLLR79}, we denote our problem as \( Pm\,|\,\text{ConfG} = (V, E)\,|\,\sum C_j \), where \( \text{ConfG} = (V, E) \) indicates the presence of a conflict graph \( G = (V, E) \) over the jobs. The complement of the conflict graph \( G \) is called the agreement graph, denoted by \( \overline{G} = (V, \overline{E}) \). It is straightforward to see that scheduling with a conflict graph and scheduling with an agreement graph are polynomially equivalent.

Such conflicts arise naturally in various practical scenarios. A common example occurs when jobs require additional resources for processing, but those resources are available only in limited quantities. If the combined demands of certain jobs exceed the availability of at least one resource, the jobs cannot be processed simultaneously on different machines and are therefore in conflict~\cite{bendraouche2015scheduling}. Similar constraints may arise in workforce planning during resource assignment~\cite{gardi2009mutual}, and in parallel computing when balancing computational load~\cite{baker1996mutual}. Further applications can be found in~\cite{bendraouche2012scheduling}.

Scheduling with conflicts has been studied for parallel machines--specifically in the case of identical machines--as well as for dedicated machines, primarily with the objective of minimizing the maximum completion time ($C_{max}$). 

For identical machines, the problem was first introduced by~\citet{baker1996mutual} for unit-time jobs, under the name \emph{Mutual Exclusion Scheduling} (MES). This problem corresponds to finding a coloring of the conflict graph $G$, with the least number of colors, where each color appears at most $m$ times. The MES problem is NP-hard since the graph coloring problem is NP-hard~\cite{karp2009reducibility}. For arbitrary processing times, the problem remains NP-hard even on two machines without conflicts (i.e., when the conflict graph is edgeless), as it reduces to \( P2||C_{\max} \)~\cite{GJ79}. Subsequently, much research has focused on the complexity of the problem for specific classes of conflict graphs and restricted processing times. In particular, \citet{bendraouche2012scheduling} showed that for processing times \( p_j \in \{1,2,3\} \), the two-machine problem is strongly NP-hard when the conflict graph is a complement of a bipartite graph. Later, the same authors definitively closed the complexity status in~\cite{bendraouche2015scheduling} for the case of two machines with two fixed processing times. Specifically, they proved that for job sizes in \( \{a, 2a + b\} \) (where \( a \geq 1 \) and  (\( b \geq 1 \) or \( -a < b < 0 \))), the problem remains strongly NP-hard when restricted to complements of bipartite graphs. Moreover, they established that scheduling on identical machines with a simple, undirected conflict graph is polynomially equivalent to a special case of the resource-constrained scheduling problem, thereby deriving new complexity results for the latter. They further extended their NP-hardness results in~\cite{bendraouche2016scheduling} to more general agreement graphs and provided new complexity results for the cases of split graphs and complements of bipartite graphs. More recently, \citet{mohabeddine2019new} proved that the problem remains NP-hard for two machines when the conflict graph is a complement of a tree. They also proposed polynomial-time algorithms for the cases where the conflict graph is a complement of a caterpillar or a cycle.

Scheduling with conflict graphs has also been studied for dedicated machines, specifically in the contexts of the open shop~\cite{TellacheBoudharDAM2017, TELLACHE2019154, TELLACHE202185, tellache2023genetic}, flow shop~\cite{tellache2018ANOR, tellache2017two, cai2018approximation}, and job shop~\cite{TellacheAzerine2026} problems, with the objective of minimizing the maximum completion time.

In the absence of conflict constraints, the problem of minimizing the total completion time on identical parallel machines—denoted as \( Pm\,||\,\sum C_j \)—can be solved in polynomial time by scheduling jobs in non-decreasing order of their processing times, according to the Shortest Processing Time (SPT) rule~\cite{conway1967theory}. However, when job weights are introduced and the objective is to minimize the total weighted completion time, i.e., \( Pm\,||\,\sum w_j C_j \), the problem becomes strongly NP-hard~\cite{bruno1974scheduling} even for $m \geq 2$. In contrast, the single-machine case (\(1||\,\sum w_j C_j \)) can be solved in polynomial time using the well-known Smith's ratio rule~\cite{smith1956}. \citet{kowalczyk2018branch} investigated \( Pm\,||\,\sum w_j C_j \) and proposed a branch-and-price algorithm incorporating several enhancements, including a generic branching strategy, zero-suppressed binary decision diagrams, dual stabilization, and Farkas pricing. For the same problem, \citet{kramer2019enhanced} introduced an arc-flow formulation on a capacitated network, which proved effective in solving instances with up to 400 jobs to proven optimality. Their approach generally outperformed direct time-indexed formulations and earlier branch-and-price methods, including that of \citet{kowalczyk2018branch}. Beyond exact approaches, approximation algorithms have also been investigated for this problem. \citet{kawaguchi1986worst}
established a worst-case approximation ratio of $(1+\sqrt{2})/2$ for list scheduling based on Smith's rule. Improving upon this result, \citet{skutella2000ptas} introduced a polynomial-time approximation scheme (PTAS)
for $P||\sum_j w_jC_j$, providing a $(1+\epsilon)$-approximation for any fixed $\epsilon>0$. Their approach groups jobs according to their weight-to-processing-time ratios and combines rounding techniques with integer programming to construct schedules. A related extension was studied by \citet{KRAMER2021} who includes setup times and proposed multiple formulations based on one-commodity flow, arc-flow, and set-covering models. Their methods successfully handled instances with up to 80 jobs within practical computation times. Other related studies addressed mean flow time minimization on identical machines under different constraints. For instance, \citet{baptiste2007complexity} considered the preemptive scheduling of jobs with release times and showed that, when all processing times are equal, the mean flow time minimization problem can be solved in polynomial time using linear programming. \citet{wang2019complexity} focused on unit-time jobs with in-tree precedence constraints and demonstrated that the problem is strongly NP-hard by reduction from the 3-Partition problem.

As seen previously, minimizing the total completion time on any number of identical machines—i.e., \( Pm\,||\,\sum C_j \)—can be solved in polynomial time in the absence of conflict constraints. However, we show that introducing conflict constraints represented by a conflict graph makes the problem NP-hard even on two machines and when jobs have only two distinct processing times. For unit-time jobs, we establish an equivalence with a special case of batch scheduling, which allows us to show that the problem is polynomially solvable on two machines for any conflict graph. When the number of machines increases to three, the problem becomes NP-hard. Nevertheless, it remains polynomially solvable when the conflict graph is the complement of a bipartite graph. We also identify a tractable case for any number of machines when the conflict graph is the complement of a star. To address the general case $Pm\,|\,\text{ConfG} = (V, E)\,|\,\sum C_j$, we present three mixed-integer linear programming formulations and lower bounds on the total completion time. We also propose genetic algorithms for the general problem and evaluate their performance on a wide set of instances derived from the benchmark instances of~\citet{kowalczyk2018branch}.

The remainder of this paper is structured as follows. Section~\ref{sec:complexity} discusses the complexity of the problem. Section~\ref{sec:MILP} presents mathematical models, followed by lower bounds in Section~\ref{sec:lowerbounds}. Section~\ref{sec:GA} describes genetic algorithms. Section~\ref{sec:experiments} presents the computational experiments,
with additional results for different conflict graph structures
reported in Appendix~\ref{sec:CDLOG}.

\section{Computational complexity}\label{sec:complexity}
We observe two trivial cases that can be solved in polynomial time. First, if the conflict graph is empty (edgeless graph), the problem reduces to \( Pm\,||\,\sum C_j \), which can be solved optimally using the SPT rule~\cite{conway1967theory}. Second, if the conflict graph is a clique—meaning all jobs are in conflict—an optimal schedule can be obtained by processing the jobs in disjoint time intervals following the SPT order.

Another tractable case is when the conflict graph $G$ is the complement of a star. let \( J_1 \) denote the central job and \( J_2, \ldots, J_n \) the remaining jobs. Here, \( J_2, \ldots, J_n \) are all mutually in conflict, while each of them is in agreement with \( J_1 \). This implies that \( J_2, \ldots, J_n \) must be scheduled sequentially in non-overlapping time intervals, whereas \( J_1 \) can be scheduled concurrently with any of them. Since \( m \geq 2 \), an optimal schedule consists of assigning \( J_1 \) to start at time $0$ on one machine, while the remaining jobs are processed in SPT order on a second machine, also starting at time $0$.

In the next theorem, we show that the problem becomes NP-hard on two machines when job processing times are restricted to two distinct values, \( p_j \in \{a, b\} \), with \( a \ne b \) for all \( j = 1, \ldots, n \).

\begin{figure}[h!]
	\centering 
	\begin{tikzpicture}[scale=0.6]
		\tikzstyle{every node}=[draw,shape=circle, minimum size=10mm, fill=gray!5];
		\node (a) at (0,4) {$J_1$};
		\node (b) at (0,2.5) {$J_2$};
		\node (c) at (0,0) {$J_n$};
		
		\node (d) at (3,2) {$J_{n+1}$};
		
		\tikzstyle{fleche}=[-,>=latex,thick]
		\draw[fleche] (a) -- (d);
		\draw[fleche] (b) -- (d);
		\draw[fleche] (c) -- (d);
		
		\tikzstyle{every node}=[];
		\node at (0,1.35) {$\vdots$};
		
		\node[above] at (0,5.25) {$p_j=b$};
		\node[above] at (3,3.25) {$p_j=a$};
		
		\draw[dashed,thick] (-1,-1) rectangle (1,5);
		\draw[dashed,thick] (2,1) rectangle (4,3);
		
	\end{tikzpicture}
	\caption{Agreement graph $\overline{G} = (V, \overline{E})$ associated with the reduction of Theorem~\ref{th1}}\label{Fig1}
\end{figure}

\begin{theorem}\label{th1}
	Problem $ P2|\text{ConfG} = (V, E),p_{j}\in\{a,b\}|\sum C_j$ is NP-hard in the strong sense for $a\ne b$.
\end{theorem}

\begin{proof} We establish this result through a polynomial reduction from the Hamiltonian Path (HP) problem. Given an undirected graph \( G' = (V', E') \), the HP problem asks whether there exists a path that visits each vertex exactly once. It is well known that the HP problem is NP-complete~\cite{GJ79}.
	
	Given an arbitrary instance of HP, we construct an instance of \( P2 | \text{ConfG} = (V, E), p_i \in \{a, b\} | \sum C_j \) as follows (see Figure~\ref{Fig1}). Specifically, we derive the agreement graph $\overline{G} = (V, \overline{E})$ from \( G' \), where:
	\begin{itemize}
		\item $V=V'\cup \{J_{n+1}\} $ (adding a extra vertex),
		\item $\overline{E}=E'\cup \{(J_{i}, J_{n+1}), j \in 1,...,n\}$,
	\end{itemize}
	The processing times of the jobs are defined as follows, with $a < b$:
	\begin{itemize}
		\item $p_j=b$, $j=1,...,n$;
		\item $p_{n+1}=a$.
	\end{itemize}
	
	We show that the HP problem has a solution if and only if there exists a feasible schedule of the jobs such that the total completion time satisfies
	
	\begin{equation}
		\sum C_j \leq 
		\begin{cases}
			\dfrac{n+2}{2}\left( \dfrac{n}{2}b+a\right), & \text{if } n \text{ is even}, \\
			\\
			\dfrac{n+1}{2}\left( \dfrac{n+1}{2}b+a\right), & \text{if } n \text{ is odd}.
		\end{cases}
		\label{eqth1:completion_times}
	\end{equation}

	Suppose the HP problem has a solution, meaning there exists a Hamiltonian path in the graph \(G'\) that visits each vertex exactly once. From this solution, we construct a solution to our scheduling problem by assigning the jobs to two machines according to their order in the Hamiltonian path, alternating between the two machines and starting with \(J_{n+1}\), as shown in Figures~\ref{Fig2a} and~\ref{Fig2b}.
	
	The sum of the completion times when $n$ is even is given by:
	$$
	\sum C_j= \left( a+(a+b)+(a+2b)+\ldots+\left(a+ \frac{n}{2}b\right) \right) +\left( b+2b+\ldots+\frac{n}{2}b\right) 
	$$
	
	This simplifies to:
	$$
	\sum C_j= \dfrac{n+2}{2}\left( \dfrac{n}{2}b+a\right) 
	$$
	
	Similarly, for $n$ odd, the sum of the completion times is:
	$$
	\sum C_j= \left( a+(a+b)+(a+2b)+\ldots+\left(a+ \frac{n-1}{2}b\right) \right) +\left( b+2b+\ldots+\frac{n+1}{2}b\right) 
	$$
	
	This simplifies to:
	$$
	\sum C_j= \dfrac{n+1}{2}\left( \dfrac{n+1}{2}b+a\right) 
	$$
	
	Therefore, the desired result follows.
	
		\begin{figure}[H]
		\centering 
		\begin{tikzpicture}[scale=0.75]
			\draw[thick] (1,2) rectangle (2,3);
			\draw[thick] (2,2) rectangle (4,3);
			\draw[thick] (4,2) rectangle (6,3);
			\draw[thick] (9,2) rectangle (11,3);
			
			\draw[thick] (1,1) rectangle (3,2);
			\draw[thick] (3,1) rectangle (5,2);
			\draw[thick] (8,1) rectangle (10,2);
			
			\node (a) at (1.5,2.5) {$J_{n+1}$};
			\node (b) at (3,2.5) {$\bullet$};
			\node (c) at (5,2.5) {$\bullet$};
			\node (d) at (7,2.5) {};
			\node (f) at (8,2.5) {};
			\node (e) at (10,2.5) {$\bullet$};
			
			\node (g) at (2,1.5) {$\bullet$};
			\node (h) at (4,1.5) {$\bullet$};
			\node (i) at (6,1.5) {};
			\node (j) at (8,1.5) {};
			\node (k) at (9,1.5) {$\bullet$};

			\tikzstyle{fleche}=[-,>=latex,thick,dashed]
			\draw[fleche] (a) -- (g);
			\draw[fleche] (g) -- (b);
			\draw[fleche] (b) -- (h);
			\draw[fleche] (h) -- (c);
			\draw[fleche] (c) -- (i);
			\draw[fleche] (e) -- (k);
			\draw[fleche] (k) -- (f);
			
			\tikzstyle{every node}=[];
			\node at (7.5,2.5) {$\cdots$};
			\node at (6.5,1.5) {$\cdots$};
			\node[below=2pt] at (1,1) {$0$};
			\node[below=2pt] at (2,1) {$a$};
			\node[below=2pt] at (3,1) {$b$};
			\node[below=2pt] at (5,1) {$2b$};
			\node[below=2pt] at (10,1) {$\frac{n}{2}b$};
		\end{tikzpicture}
		\caption{Feasible schedule for even $n$ in the reduction instance of Theorem~\ref{th1}.}\label{Fig2a}
	\end{figure}
	
	\begin{figure}[H]
		\centering 
		\begin{tikzpicture}[scale=0.75]
			\draw[thick] (1,2) rectangle (2,3);
			\draw[thick] (2,2) rectangle (4,3);
			\draw[thick] (4,2) rectangle (6,3);
			\draw[thick] (8,2) rectangle (10,3);
			
			\draw[thick] (1,1) rectangle (3,2);
			\draw[thick] (3,1) rectangle (5,2);
			\draw[thick] (9,1) rectangle (11,2);
			
			\node (a) at (1.5,2.5) {$J_{n+1}$};
			\node (b) at (3,2.5) {$\bullet$};
			\node (c) at (5,2.5) {$\bullet$};
			\node (d) at (7,2.5) {};
			\node (f) at (8,2.5) {};
			\node (e) at (9,2.5) {$\bullet$};
			
			\node (g) at (2,1.5) {$\bullet$};
			\node (h) at (4,1.5) {$\bullet$};
			\node (i) at (6,1.5) {};
			\node (j) at (8,1.5) {};
			\node (k) at (10,1.5) {$\bullet$};

			\tikzstyle{fleche}=[-,>=latex,thick,dashed]
			\draw[fleche] (a) -- (g);
			\draw[fleche] (g) -- (b);
			\draw[fleche] (b) -- (h);
			\draw[fleche] (h) -- (c);
			\draw[fleche] (c) -- (i);
			\draw[fleche] (e) -- (k);
			\draw[fleche] (e) -- (j);
			
			\tikzstyle{every node}=[];
			\node at (7,2.5) {$\cdots$};
			\node at (7,1.5) {$\cdots$};
			\node[below=2pt] at (1,1) {$0$};
			\node[below=2pt] at (2,1) {$a$};
			\node[below=2pt] at (3,1) {$b$};
			\node[below=2pt] at (5,1) {$2b$};
			\node[below=2pt] at (11,1) {$\frac{n+1}{2}b$};
			
		\end{tikzpicture}
		\caption{Feasible schedule for odd $n$ in the reduction instance of Theorem~\ref{th1}.}\label{Fig2b}
	\end{figure}
	
	Conversely, suppose that the scheduling problem has a feasible solution that satisfies~\eqref{eqth1:completion_times}. We now demonstrate that the only possible schedule satisfying~\eqref{eqth1:completion_times} is the one shown in Figures~\ref{Fig2a} and~\ref{Fig2b}. Suppose, by contradiction, that there exists another schedule satisfying~\eqref{eqth1:completion_times} but differing from the one in Figures~\ref{Fig2a} and~\ref{Fig2b}. Such a schedule could arise in one of the following ways:
	
	\begin{itemize} 
		\item Scheduling $J_{n+1}$ at a time later than $0$. This would increase the sum of the completion times according to the SPT rule, thereby violating~\eqref{eqth1:completion_times}. 
		\item Scheduling $J_{n+1}$ at the beginning, but not in parallel with any other job. This would introduce idle time on the opposite machine, and it can be checked that~\eqref{eqth1:completion_times} would not hold.
		\item Scheduling at least one job with processing time $b$ not in parallel with two other jobs. This would cause idle time on one machine because $J_{n+1}$ is scheduled at the beginning, which increases the sum of the completion times and violates~\eqref{eqth1:completion_times}. \end{itemize}
	
	Thus, the only feasible schedule that satisfies~\eqref{eqth1:completion_times} is the one shown in Figures~\ref{Fig2a} and~\ref{Fig2b}. A Hamiltonian path is constructed by following the dashed paths in these figures, which connect jobs scheduled in parallel. Finally, it can be verified that the reduction is polynomial, completing the proof of Theorem~\ref{th1}.
\end{proof}

Since the problem is NP-hard even when processing times are restricted to two distinct values $\{a,b\}$, with $a \neq b$, we now consider unit-time jobs to investigate whether the problem admits polynomially solvable cases. The problem with unit-time jobs on two machines is related to the scheduling of unit-time jobs on a \(p\)-batching machine with capacity 2 and a compatibility graph~\cite{fowler2022survey}. The latter belongs to the class of batch scheduling problems, in which several jobs can be processed simultaneously on the same machine as a batch. The batch capacity specifies the maximum number of jobs that can be processed simultaneously. In batch scheduling with compatibility graphs, only mutually compatible jobs, i.e., jobs that are adjacent in the compatibility graph, can be assigned to the same batch. The processing time of a batch is equal to the maximum processing time among the jobs assigned to it.

For a single batch-processing machine \(B1\) with batch capacity \(b=2\), the problem is denoted by $B1 \mid G=(V,E),\, b=2,\, p_j=1 \mid C_{\max}.$ We establish the relationship between the two problems in the following theorem.

\begin{theorem}\label{thm2}
	The problems $P2 \mid \mathrm{ConfG}=(V,E),\, p_j=1 \mid \sum C_j$ and
	$B1 \mid G=(V,E),\, b=2,\, p_j=1 \mid C_{\max}$ are polynomially equivalent.
\end{theorem}

\begin{proof}
We first show the equivalence between 
	$P2 \mid \mathrm{ConfG}=(V,E),\, p_j=1 \mid \sum C_j$ and
	$B1 \mid G=(V,E),\, b=2,\, p_j=1 \mid \sum C_j$.
	
	Given an instance of the first problem, we construct an instance of the second problem as follows. The set of jobs is the same in both problems, and all jobs have unit processing times. The compatibility graph is defined as the complement of the conflict graph. We consider a single batch processing machine with capacity 2.
	
	Conversely, given an instance of the second problem, we construct an instance of the first problem as follows. The set of jobs is the same in both problems, and all jobs have unit processing times. The conflict graph is defined as the complement of the compatibility graph. We consider two parallel identical machines.
	
	Thus, two jobs can be processed simultaneously in 
	$P2 \mid \mathrm{ConfG}=(V,E),\, p_j=1 \mid \sum C_j$
	if and only if they can be assigned to the same batch in
	$B1 \mid G=(V,E),\, b=2,\, p_j=1 \mid \sum C_j$ and since all jobs have unit processing times, the equivalence between $P2 \mid \mathrm{ConfG}=(V,E),\, p_j=1 \mid \sum C_j$ and
	$B1 \mid G=(V,E),\, b=2,\, p_j=1 \mid \sum C_j$ follows.
	
	It remains to show the equivalence between
	$B1 \mid G=(V,E),\, b=2,\, p_j=1 \mid \sum C_j$
	and
	$B1 \mid G=(V,E),\, b=2,\, p_j=1 \mid C_{\max}$.
	Since all jobs have unit processing times, every batch has processing time equal to 1. Therefore, minimizing the total completion time is equivalent to minimizing the number of batches, which is equivalent to minimizing the makespan. Hence, the two objectives are equivalent.
	
	Combining the two equivalences, we conclude that
	$P2 \mid \mathrm{ConfG}=(V,E),\, p_j=1 \mid \sum C_j$
	and
	$B1 \mid G=(V,E),\, b=2,\, p_j=1 \mid C_{\max}$
	are polynomially equivalent.
\end{proof}

This correspondence allows us to derive the following polynomially solvable case.

\begin{corollary}\label{coro1}
	$P2 \mid \text{ConfG} = (V, E),\ p_j = 1 \mid \sum C_j $  can be solved in polynomial time by Algorithm~\ref{Algo1}, and the optimal total completion time is
	\[
	\sum C_j = |M|(|M|-n) + \frac{n(n+1)}{2},
	\]
	where \( M \) is a maximum matching in \( \overline{G} \).
\end{corollary}

\begin{figure}[h!]
	\centering 
	\scriptsize
	\begin{tikzpicture}[scale=0.75]
		\draw[thick] (1,2) rectangle (2,3);
		\draw[thick] (2,2) rectangle (3,3);
		\draw[thick] (5,2) rectangle (6,3);
		
		\draw[thick] (1,1) rectangle (2,2);
		\draw[thick] (2,1) rectangle (3,2);
		\draw[thick] (5,1) rectangle (6,2);
		\draw[thick] (6,1) rectangle (7,2);
		\draw[thick] (7,1) rectangle (8,2);
		\draw[thick] (10,1) rectangle (11,2);
		
		\draw[thick] (1,1) rectangle (11,1);
		
		\node (a) at (1.5,2.5) {$\bullet$};
		\node (b) at (2.5,2.5) {$\bullet$};
		\node (c) at (5.5,2.5) {$\bullet$};
		\node (d) at (1.5,1.5) {$\bullet$};
		\node (e) at (2.5,1.5) {$\bullet$};
		\node (f) at (5.5,1.5) {$\bullet$};
		\node (g) at (6.5,1.5) {$\bullet$};
		\node (h) at (7.5,1.5) {$\bullet$};
		\node (i) at (10.5,1.5) {$\bullet$};
		
		\tikzstyle{fleche}=[-,>=latex,thick,dashed]
		\draw[fleche] (a) -- (d);
		\draw[fleche] (b) -- (e);
		\draw[fleche] (c) -- (f);
		
		\tikzstyle{every node}=[];
		\node at (4,2.5) {$\cdots$};
		\node at (4,1.5) {$\cdots$};
		\node at (9,1.5) {$\cdots$};
		\node[above] at (1,0.3) {$0$};
		\node[above] at (2,0.3) {$1$};
		\node[above] at (3,0.3) {$2$};
		\node[above] at (6,0.3) {$|M|$};
		\node[above] at (7.25,0.3) {$|M|+1$};
		\node[above] at (11.25,0.3) {$n-|M|$};
		
		\node[above] at (0.5,2.25) {$M_1$};
		\node[above] at (0.5,1.25) {$M_2$};
		
	\end{tikzpicture}
	\caption{Optimal schedule for $P2|G =(V, E),p_{j}=1|\sum C_j$, as generated by Algorithm~\ref{Algo1}}\label{Fig6}
\end{figure}

\begin{proof} This result follows from Theorem~\ref{thm2}, together with Corollary~1 and a modified version of Algorithm~A2 in~\cite{boudhar2000scheduling}. For completeness, we provide the details of the proof. Suppose by contradiction that the schedule produced by Algorithm~\ref{Algo1} is not optimal. Then, there exists another schedule with a smaller total completion time. This would require that at least one job is completed earlier than in the schedule produced by Algorithm~\ref{Algo1}. However, the jobs corresponding to the matching \( M \) are already scheduled as early as possible, in parallel pairs. Therefore, the only way to improve the schedule would be to schedule one of the unmatched jobs earlier, in parallel with another unmatched job. This would imply that these two jobs are not in conflict, meaning that an edge exists between them in the agreement graph \( \overline{G} \). Consequently, the matching \( M \) would not be maximum, contradicting its maximality. Thus, the schedule produced by Algorithm~\ref{Algo1} is optimal. As a maximum matching can be computed in polynomial time, the problem $P2 \mid \text{ConfG} = (V, E),\ p_j = 1 \mid \sum C_j $ is solvable in polynomial time.
	
	\begin{algorithm}[H]
		\caption{Optimal schedule for $P2|\text{ConfG} = (V, E),p_j=1|\sum C_j$}\label{Algo1}
		\begin{algorithmic}[1]
			\State Find a maximum matching $M$ in $\overline{G}$.
			\State Schedule each pair of jobs from the matching to run in parallel at the earliest available time, in any order (see Figure~\ref{Fig6}).
			\State Schedule the remaining jobs (those not in the matching) after the matched jobs, in disjoint time intervals, starting at the earliest possible time and in any order (see Figure~\ref{Fig6}).
		\end{algorithmic}
	\end{algorithm}

	It remains to compute the optimal total completion time. We have:
	$$
	\sum_{j=1}^{n} C_j = 2(1+2+\ldots+|M|) + ((|M|+1)+(|M|+2)\ldots+(n-|M|)) 
	$$
	
	$$
	\sum_{j=1}^{n} C_j = 2 \times \left(\dfrac{|M| (|M| + 1)}{2}\right) + 
	\dfrac{n+1}{2} (n-2|M|)
	$$
	
	$$
	\sum_{j=1}^{n} C_j = \dfrac{n(n+1)}{2}+|M|(|M|-n)
	$$
\end{proof}

For further results on special cases of \(B1 \mid G=(V,E),\, b=2,\, p_j=1 \mid C_{\max}\), we refer the reader to~\cite{boudhar2009scheduling}.

We now show that when the number of machines increases to three, the problem with unit-time jobs becomes strongly NP-hard.

\begin{figure}[h!]
	\centering 
	\begin{tikzpicture}[scale=0.7]
		\draw[thick] (1,1) rectangle (1,4);
		\draw[thick] (2,1) rectangle (2,4);
		\draw[thick] (3,1) rectangle (3,4);
		\draw[thick] (5,1) rectangle (5,4);
		\draw[thick] (6,1) rectangle (6,4);
		\draw[thick] (8,1) rectangle (8,4);
		\draw[thick] (9,1) rectangle (9,4);
		
		\draw[->,>=latex] (1,1) to (10,1);
		\draw[thick] (1,2) rectangle (9,2);
		\draw[thick] (1,3) rectangle (9,3);
		\draw[thick] (1,4) rectangle (9,4);
		
		\tikzstyle{every node}=[];
		\node at (4,1.5) {$\cdots$};
		\node at (4,2.5) {$\cdots$};
		\node at (4,3.5) {$\cdots$};
		
		\node at (7,1.5) {$\cdots$};
		\node at (7,2.5) {$\cdots$};
		\node at (7,3.5) {$\cdots$};
		
		\node[below=2pt] at (1,1) {$0$};
		\node[below=2pt] at (2,1) {$1$};
		\node[below=2pt] at (3,1) {$2$};
		\node[below=2pt] at (6,1) {$j$};
		\node[below=2pt] at (9,1) {$q$};
		\node[below=2pt] at (10,1) {time};
		
		\node[above] at (0.5,3.25) {$M_3$};
		\node[above] at (0.5,2.25) {$M_2$};
		\node[above] at (0.5,1.25) {$M_1$};
		
		\node[above] at (1.5,4) {$V_1$};
		\node[above] at (2.5,4) {$V_2$};
		\node[above] at (4.,4) {$\cdots$};
		\node[above] at (5.5,4) {$V_j$};
		\node[above] at (7.,4) {$\cdots$};
		\node[above] at (8.5,4) {$V_q$};
		
	\end{tikzpicture}
	\caption{Feasible schedule corresponding to the reduction instance of Theorem~\ref{thm3}.}\label{Fig8}
\end{figure}

\begin{figure}[h!]
	\centering 
	\begin{tikzpicture}[scale=0.7]
		\draw[thick] (1,1) rectangle (1,4);
		\draw[thick] (2,1) rectangle (2,4);
		\draw[thick] (3,1) rectangle (3,4);
		\draw[thick] (6,1) rectangle (6,4);
		\draw[thick] (7,1) rectangle (7,4);
		\draw[thick] (8,1) rectangle (8,3);
		\draw[thick] (9,1) rectangle (9,2);
		
		\draw[->,>=latex] (1,1) to (10,1);
		\draw[thick] (1,2) rectangle (9,2);
		\draw[thick] (1,3) rectangle (8,3);
		\draw[thick] (1,4) rectangle (7,4);
		
		\tikzstyle{every node}=[];
		\node at (4.5,1.5) {$\cdots$};
		\node at (4.5,2.5) {$\cdots$};
		\node at (4.5,3.5) {$\cdots$};
		
		\node[below=2pt] at (1,1) {$0$};
		\node[below=2pt] at (2,1) {$1$};
		\node[below=2pt] at (3,1) {$2$};
		\node[below=2pt] at (7,1) {$q-1$};
		\node[below=2pt] at (8,1) {$q$};
		\node[below=2pt] at (9,1) {$q+1$};
		\node[below=2pt] at (10.3,1) {time};
		
		\node[above] at (0.5,3.25) {$M_3$};
		\node[above] at (0.5,2.25) {$M_2$};
		\node[above] at (0.5,1.25) {$M_1$};
		
	\end{tikzpicture}
	\caption{Impact of scheduling a job after time $q$ in the reduction instance of Theorem~\ref{thm3}.}\label{Fig9}
\end{figure}

\begin{theorem} \label{thm3}
	Problem $P3|\text{ConfG} = (V, E),p_j=1|\sum C_j$ is NP-hard in the strong sense. 
\end{theorem}

\begin{proof} The proof is by reduction from the Partition Into Triangles (PIT) problem, which is known to be NP-complete~\cite{GJ79}. The PIT problem is defined as follows: given a graph \( G' = (V', E') \) where \( |V'| = 3q \) and \( q \in \mathbb{N}^* \), the objective is to partition the vertex set \( V' \) into \( q \) subsets \( V'_1, V'_2, \dots, V'_q \), each of size 3, such that the vertices in each subset form a triangle (i.e., each subset is a complete subgraph of \( G' \)).

	Given an arbitrary instance of the PIT problem, we construct an instance of our scheduling problem as follows. The agreement graph for the scheduling problem is identical to the graph \( G' \), so the number of jobs is \( n = 3q \), with each job having a processing time of one unit. These jobs are to be processed on three identical machines.
	
	We show that the PIT problem has a solution if and only if there exists a feasible schedule with \( \sum C_j \leq \frac{3q(q+1)}{2} \).
	
	Assume that the PIT problem has a solution, meaning that the graph \( G' \) can be partitioned into triangles \( V'_j \) for \( j = 1, \dots, q \). In this case, it is possible to schedule the corresponding jobs on the machines such that each triangle is processed simultaneously, without violating any conflict constraints (see Figure \ref{Fig8}).
	
	The sum of completion times of the  $ 3q $ jobs is:
	$$
	\sum_{j=1}^{n} C_j = 3 \times 1 + 3 \times 2 + 3 \times 3 + \dots + 3 \times q
	$$
	
	Thus, 
	$$
	\sum_{j=1}^{n} C_j = \frac{3q(q + 1)}{2}.
	$$
	
	Conversely, suppose that the scheduling problem has a feasible schedule with $\sum C_j \leq \frac{3q(q + 1)}{2}$. If we show that all jobs are completed by time $q$, then the schedule must be free of idle time on the machines. In this case, a solution to the PIT problem can be derived by considering the jobs scheduled in parallel for each $j=1,\ldots,n$, as these parallel jobs must form a clique in the agreement graph since they are scheduled simultaneously.
	
	Assume that at least one job is processed after time \( q \) (see Figure \ref{Fig9}). In this case, the sum of the completion times is given by:
	
	$$
	\sum_{j=1}^{n} C_j = 3 \times (1 + 2 + \ldots + (q - 1)) + 2q + (q + 1)
	$$
	
	Simplifying this expression:
	
	$$
	\sum_{j=1}^{n} C_j = 3 \left( \frac{(q - 1)q}{2} \right) + 3q + 1
	$$
	
	We have:
	
	$$
	3 \left( \frac{(q - 1)q}{2} \right) + 3q + 1 > \frac{3q(q + 1)}{2}
	$$
	
	This contradicts the assumption that \( \sum C_j \leq \frac{3q(q + 1)}{2} \). Therefore, all jobs must be completed by time \( q \). The reduction is polynomial, thus concluding the proof.
	
\end{proof}

We now show that the problem with unit-time jobs becomes polynomial on three machines when the agreement graph is bipartite.

\begin{theorem} 
	Problem $P3|\text{ConfG} = (V, E),p_j=1|\sum C_j$ can be solved in polynomial time for $G$ being a complement of a bipartite graph.
\end{theorem}

\begin{proof} In a bipartite graph, any clique contains at most two vertices. Therefore, the problem is equivalent to the problem \( P2|\text{ConfG} = (V, E), p_j=1|\sum C_j \) with two machines, which is solvable in polynomial time (Theorem~\ref{thm2}).
\end{proof} 

\section{Mathematical formulations}  \label{sec:MILP}
This section introduces three mathematical formulations for $Pm|ConfG = (V,E)|\sum C_j$, adapted from existing models in the literature for scheduling variants on parallel machines.

\subsection{Time-indexed formulation}
Discrete-time models have been studied in the literature to address various machine scheduling problems, including those involving parallel machines (see, e.g.,\cite{velez2017changeover, KRAMER2021, wolsey1997mip}). Building on these models, we present a formulation that incorporates the conflict constraints between jobs. This formulation makes use of a binary decision variable $x_{jt}$, which equals 1 if job $ J_{j} $ starts processing at time $t$, and 0 otherwise. Let $T$ represent an upper bound on the job completion time. In the computational experiments, we set $T$ equal to the sum of the processing times of all jobs. The first ILP formulation, denoted by \mFA, is summarized as follows.

\begin{subequations}\label{P1}
	\begin{align}
		(\mFA)\quad \min \quad & \sum_{j \in J} \sum_{t=0}^{T-p_j} (t+p_j) x_{jt} \label{P1:obj}\\ 
		s.t. \quad & \sum_{t=0}^{T-p_j} x_{jt} =1, && j \in J \label{P1:c1}\\
		&  \sum_{j \in J} \sum_{s=\max\{0,t-p_j+1\}}^{\min\{t,T-p_j\}} x_{js} \leq m, && t \in\{0,\ldots,T-1\} \label{P1:c2} \\
		&  \sum_{s=\max\{0,t-p_j+1\}}^{\min\{t,T-p_j\}} x_{js} +  \sum_{s=\max\{0,t-p_k+1\}}^{\min\{t,T-p_k\}} x_{ks} \leq 1, && j, k \in J \textrm{ in conflict},  j < k,\notag\\
		&&& t \in\{0,\ldots,T-1\} \label{P1:c3}\\	
		&x_{jt} \in \{0,1\}, && j \in J, t \in \{0,\ldots,T-p_j\} \label{P1:c4} 
	\end{align}
\end{subequations}

The objective function~\eqref{P1:obj} minimizes the total completion time of the jobs. Constraints~\eqref{P1:c1} enforce that each job has exactly one starting time. Constraints~\eqref{P1:c2} ensure that at most $m$ jobs are processed at any time.  Constraints~\eqref{P1:c3} ensure that two conflicting jobs cannot be processed at the same time on different machines.

\subsection{Time and precedence relationship-indexed formulation}
This formulation is inspired by the one-commodity model, which has been applied in various contexts, including parallel machine scheduling~\cite{KRAMER2021}. It incorporates both time-indexed and precedence relationship variables. We introduce the binary decision variable \(x_{jkt}\), which takes value 1 if job \(J_k\) is scheduled as the immediate successor of job \(J_j\) on the same machine, and starts processing at time \(t\), and 0 otherwise. Note that “immediate successor” refers to the sequencing of jobs on a machine, and does not imply that \(J_k\) starts at the completion time of \(J_j\). We also introduce the binary decision variable $y_{jk}$, defined for each pair of conflicting jobs $J_j$ and $J_k$, with $j<k$. It takes value 1 if $J_j$ completes before $J_k$ starts, and 0 otherwise. Let $J_0$ be a fictitious job with $p_0=0$ and $J'=J \cup \{J_0\}$. The formulation is summarized in (\mFB).

\begin{subequations}\label{P2}
	\begin{align}
		(\mFB)\quad\min \quad & \sum_{j \in J'} \sum_{k \in J\setminus \{j\}} \sum_{t=p_j}^{T-p_k} (t+p_k) x_{jkt} \label{P2:obj}\\ 
		s.t. \quad & \sum_{k \in J}\sum_{t=0}^{T-p_k}  x_{0kt} \leq m, &&  \label{P2:c1}\\
		&  \sum_{j \in J}\sum_{t=p_j}^{T}  x_{j0t} \leq m, &&   \label{P2:c2}\\
		& \sum_{j \in J' \setminus\{k\}} \sum_{t=p_j}^{T-p_k} x_{jkt} = 1, &&  k \in J  \label{P2:c3}\\
		&  \sum_{j \in J'\setminus \{k\}} x_{jkt} = \sum_{j \in J' \setminus \{k\}} \sum_{t'=t+p_k}^{T-p_j} x_{kjt'}, && k\in J,\notag \\
		&       && t \in\{0,\ldots,T-p_k\}    \label{P2:c4} \\
		& -T(1-y_{jk})+ \sum_{l \in J'\setminus \{j\}} \sum_{t=p_l}^{T-p_j}(t+p_j) x_{ljt}  \leq \sum_{l \in J' \setminus \{k\}} \sum_{t=p_l}^{T-p_k}t x_{lkt}, &&   j  < k \in J \textrm{ in conflict},  \label{P2:c5} \\
		&  \sum_{l \in J'\setminus \{j\}} \sum_{t=p_l}^{T-p_j}t x_{ljt} \geq \sum_{l \in J' \setminus \{k\}} \sum_{t=p_l}^{T-p_k}(t+p_k) x_{lkt}-Ty_{jk}, &&   j  < k \in J \textrm{ in conflict}, \label{P2:c6} \\
		&x_{jkt} \in \{0,1\}, &&j, k\in J', j\neq k, \notag \\
		&&& t \in \{p_j,\ldots,T-p_k\} \label{P2:c7-time} \\
		&y_{jk} \in \{0,1\}, && j <  k \in J \textrm{ in conflict}.\label{P2:c7} 
	\end{align}
\end{subequations} 
The objective function~\eqref{P2:obj} minimizes the total completion times of the jobs. Constraints~\eqref{P2:c1} and~\eqref{P2:c2}  ensure that $J_0$ has $m$ immediate successors and $m$ immediate predecessors respectively. Constraints~\eqref{P2:c3} ensure that each job in $J$ has one immediate predecessor on a given machine. Constraints~\eqref{P2:c4} ensure that the immediate successor of a job starts after the completion time of this job. Constraints~\eqref{P2:c5} and~\eqref{P2:c6} ensure that two conflicting jobs cannot be processed at the same time on two different machines.

\subsection{Precedence relationship-indexed formulation}
This formulation is based on the precedence relationships between jobs and utilizes the decision variable $x_{jk}$, which equals 1 if job $ J_{k} $ is processed immediately after $J_j$ on the same machine, and 0 otherwise, and the decision variable $ y_{jk}$, which is equal to 1 if job $ J_{j} $ is scheduled any time before $J_k$, and 0 otherwise. We also define the continuous variable $C_j$, which represents the completion time of job  $J_j$. Similarly to (\mFB), we introduce a fictitious job $J_0$ with $p_0=0$, and let $J'=J \cup \{J_0\}$. This is summarized in (\mFC).

\begin{subequations}\label{P3}
	\begin{align}
		(\mFC)\qquad \min \quad & \sum_{j \in J}  C_{j} \label{P3:obj}\\ 
		s.t. \quad & \sum_{j \in J}  x_{0j} \leq m, &&  \label{P3:c1}\\
		&  \sum_{j \in J}  x_{j0} \leq m, &&   \label{P3:c2}\\
		& \sum_{j \in J'} x_{jk} = 1, &&  k \in J  \label{P3:c3}\\
		&  \sum_{k \in J'} x_{jk} = 1, &&  j \in J  \label{P3:c4}\\ 
		& -T(1-x_{jk})+ C_{j}  \leq  C_{k} -p_k, &&  j, k \in J, j \neq k,  \label{P3:c5} \\
		& -T(1-y_{jk})+ C_{j}  \leq  C_{k} -p_k, &&  j, k \in J \textrm{ in conflict},  j < k,\label{P3:c6} \\
		&  C_j-p_j \geq C_k-Ty_{jk}, &&  j, k \in J \textrm{ in conflict}, j < k, \label{P3:c7} \\
		&x_{jk} \in \{0,1\}, && j, k \in J', \label{P3:c8} \\
		&y_{jk} \in \{0,1\}, &&j, k \in J \textrm{ in conflict},  j < k,\label{P3:c9} \\
		&C_{j} \geq p_j, &&j \in J.\label{P3:c10} 
	\end{align}
\end{subequations}

	The objective function~\eqref{P3:obj} minimizes the total completion times of the jobs. Constraints~\eqref{P3:c1} and~\eqref{P3:c2}  ensure that $J_0$ has $m$ immediate successors and $m$ immediate predecessors respectively. Constraints~\eqref{P3:c3} and~\eqref{P3:c4} ensure that each job in $J$ has one immediate predecessor and one immediate successor on a given machine respectively. Constraints~\eqref{P3:c5} ensure that the immediate successor of a job starts after the completion time of this job. Constraints~\eqref{P3:c6} and~\eqref{P3:c7} ensure that two conflicting jobs cannot be processed at the same time on two different machines.
	
	\section{Lower bounds} \label{sec:lowerbounds}
	We present in this section three different classes of lower bounds for $Pm|ConfG = (V,E)|\sum C_j$.
	
	\subsection{Conflict constraint relaxation-based lower bounds}
	By relaxing the conflict constraints between the jobs, our problem is reduced to $P||\sum C_{j}$. This problem can be solved optimally by applying the SPT rule, where jobs are processed in non-decreasing order of their processing times. The corresponding $\sum C_{j}$ is a valid lower bound for $P|ConfG=(V,E)|\sum C_{j}$, denoted as $LB_1$.
	
	\subsection{Weighted independent set-based lower bounds}
	The idea of this lower bound is to find a set of conflicting jobs with maximum total processing time. These jobs should be processed in disjoint time intervals, thus, if we schedule them in non-decreasing order of their processing times, the corresponding sum of completion times is a valid lower bound for $P|ConfG=(V,E)|\sum C_{j}$.
	
	A set of conflicting jobs forms an independent set in the agreement graph $\overline{G}$. We assign a weight to each job in $V$ equal to its processing time, and let $\overline{G}_w=(V,\overline{E},W)$ be the resulting weighted agreement graph. The problem of finding an independent set of maximum weight is known to be NP-hard~\cite{GJ79}. Therefore, we use three greedy algorithms from~\cite{SMK03}, namely GWMIN, GWMIN2, and GWMAX. We denote the resulting lower bounds as $LB_2$, $LB_3$, and $LB_4$ respectively.

	\subsection{MILP relaxation-based lower bound}
	
	We derive three lower bounds, denoted by \(LB_5\), \(LB_6\), and \(LB_7\), from the linear relaxations of \mFA, \mFB, and \mFC, respectively. In addition, since each formulation is solved under a fixed time limit, we record the best lower bound obtained by the solver at termination for each formulation, denoted by \(LB'_5\), \(LB'_6\), and \(LB'_7\), respectively. The first set of bounds allows us to compare the strength of the linear relaxations, while the second provides insight into how effectively the solver strengthens these initial bounds during the solution process. The latter may also provide stronger reference bounds for the considered instances, yielding tighter optimality gaps and allowing optimality to be established for additional instances when the lower and upper bounds coincide.

	\section{Genetic algorithm} \label{sec:GA}
	Since their introduction by~\citet{Holland1975}, Genetic Algorithms (GAs) have been widely used to solve NP-hard combinatorial optimization problems, particularly those with many local optima scattered across the search space. GAs operate on a population of candidate solutions, evolving them through genetic operations—selection, crossover, mutation, and replacement. Each solution is encoded as a chromosome and evaluated by a problem-specific fitness function. In each iteration, selected parents generate offspring through recombination (crossover and mutation), and a replacement mechanism updates the population.  
	
	We present in this section a GA for the problem $Pm |ConfG = (V,E) | \sum C_j$. The general framework of the algorithm is outlined in Algorithm~\ref{GAalgo} (inspired by~\cite{prins1994}) and consists of several key procedures: initial population generation (Section~\ref{sec:initialpopulation}); selection and replacement mechanisms (Section~\ref{sec:selectionreplacement}); mutation and crossover operators (Section~\ref{sec:crossovermutation}); and fitness evaluation (Section~\ref{schedulebuilders}), where the fitness of a chromosome $X$ is its total completion time, denoted $\sum C_j(X)$. The algorithm maintains a diverse population of solutions with varying total completion times. The function \texttt{Rand} generates uniformly distributed numbers in the interval $[0,1]$. 
	
		\begin{algorithm}[H]
		\caption{General framework of the GA.}\label{GAalgo}
		\begin{algorithmic}[1]
			\State Generate an initial population;
			\While{none of the stopping criterion is met}
			\State Select two parents $P_1$ and $P_2$;
			\State Apply a crossover operator to $P_1$ and $P_2$;
			\State Randomly select one offspring $K$;
			\If{\texttt{Rand}$<p_m$}
			\State $L:=K$,
			\State Mutate $L$ using a mutation operator;
			\If {$\sum C_j(L)$ is not redundant} 
			\State $K:=L$;
			\EndIf
			\EndIf
			\If{$\sum C_j(K)$ is not redundant}
			\State Select a chromosome $R$ to be replaced;
			\State Replace $R$ by $K$;
			\EndIf
			\EndWhile
		\end{algorithmic}
		\hspace*{\algorithmicindent} \textbf{Outputs}: the best $\sum C_j$ with the corresponding schedule.
	\end{algorithm}

	The algorithm runs until one of the following conditions is met: (i) \( I_{\max} \) iterations are completed, (ii) a chromosome with a total completion time equal to the best lower bound is found (an optimal solution is found), or (iii) the maximum number of iterations without improvement, \( I_{\text{noimpv}} \), is reached. The values of \( I_{\max} \) and \( I_{\text{noimpv}} \) will be determined based on the experiments in Section~\ref{sec:experiments}.
\subsection{Encoding and decoding}\label{schedulebuilders}
We use a permutation-based representation to encode solutions, where each chromosome is represented as a permutation $\pi$ of the jobs. This encoding guarantees the feasibility of chromosomes throughout the evolution process. To evaluate a chromosome, we decode the corresponding permutation using the algorithms that we will present in this section. The total completion time obtained is then the fitness value of the chromosome.

\begin{figure}[H]
	\resizebox{0.8\textwidth}{!}{%
	\centering
	
	\begin{minipage}[t]{0.62\textwidth}
		
		\begin{tikzpicture}[baseline=(current bounding box.north),
		x=0.55cm,y=0.7cm]
		
		\draw[->] (0,0) -- (12,0);
		\foreach \x in {0,...,11}
		\draw (\x,0.08) -- (\x,-0.08) node[below] {\scriptsize \x};
		
		\node[left] at (0,1) {$M_1$};
		\node[left] at (0,2) {$M_2$};
		
		\draw (0,0.7) rectangle (1,1.3);
		\node at (0.5,1) {$J_1$};
		
		\draw (0,1.7) rectangle (2,2.3);
		\node at (1,2) {$J_2$};
		
		\draw (7,1.7) rectangle (10,2.3);
		\node at (8.5,2) {$J_3$};

		\draw (1,0.7) rectangle (7,1.3);
		\node at (3.5,1) {$J_4$};
		
		\node at (6,-1.2) {(b)};
		
	\end{tikzpicture}
	
		\vspace{0.5cm}
	
	\begin{tikzpicture}[baseline=(current bounding box.north),
		x=0.55cm,y=0.7cm]
		
		\draw[->] (0,0) -- (12,0);
		\foreach \x in {0,...,11}
		\draw (\x,0.08) -- (\x,-0.08) node[below] {\scriptsize \x};
		
		\node[left] at (0,1) {$M_1$};
		\node[left] at (0,2) {$M_2$};
		
		\draw (0,0.7) rectangle (1,1.3);
		\node at (0.5,1) {$J_1$};
		
		\draw (2,0.7) rectangle (5,1.3);
		\node at (3.5,1) {$J_3$};
		
		\draw (5,0.7) rectangle (11,1.3);
		\node at (8,1) {$J_4$};
		
		\draw (0,1.7) rectangle (2,2.3);
		\node at (1,2) {$J_2$};
		
		\node at (6,-1.2) {(c)};
		
	\end{tikzpicture}

	\end{minipage}
	\hspace{-0.3cm}
	\begin{minipage}[t]{0.30\textwidth}
		\centering
\begin{tikzpicture}[baseline=(current bounding box.north),
	every node/.style={circle,draw,minimum size=7mm}
	]
		
		\node (J3) at (1.2,0) {$J_3$};
		\node (J1) at (0,1.5) {$J_1$};
		\node (J2) at (-1.2,0) {$J_2$};
		\node (J4) at (0,-1.5) {$J_4$};
		
		\draw (J1)--(J3);
		\draw (J1)--(J4);
		\draw (J2)--(J3);
		\draw (J3)--(J4);
		
		\node[draw=none,rectangle] at (0,-2.5) {(a)};
		
	\end{tikzpicture}
		
	\end{minipage} }
	
	\caption{Example illustrating that restricting the search to non-delay schedules may exclude an optimal solution. (a) Conflict graph. (b) An optimal non-delay schedule with $\sum C_j=20$. (c) An optimal schedule that is active but not non-delay, with $\sum C_j=19$.}
	\label{fig:active-nondelay}
\end{figure}

Algorithms~\ref{evaluation:fifo},~\ref{evaluation:GT}, and~\ref{evaluation:ectf} generate active schedules, where no job can start earlier without rendering the problem infeasible. In contrast, Algorithm~\ref{evaluation:ND} produces non-delay schedules, in which no machine remains idle if a job that can be processed on it is available~\cite{MP08}. Note that restricting the search to non-delay schedules does not guarantee the presence of an optimal schedule, whereas the set of active schedules always includes at least one optimal schedule. Figure~\ref{fig:active-nondelay} shows an instance with two machines and four jobs for which an optimal non-delay schedule is not optimal and also shows an optimal schedule that is active but not non-delay.

	\begin{algorithm}[H]
		\caption{Build active schedules by first in first out}\label{evaluation:fifo} 
		\hspace*{\algorithmicindent} \textbf{Inputs}: permutation $\pi$ of jobs,  processing times $(p_{j})_{1 \leq j \leq n}$, adjacency matrix $A$ of $G$  
		\begin{algorithmic}[1]
			\State Initialize the earliest starting times $(s_{j}^i)_{1\leq i \leq m, 1 \leq j \leq n}$ to zero;
			\While {$\pi \neq \emptyset $}
			\State Select the first job $J_j$ of $\pi$ and scheduled it at $s_j$;
			\State $C_{j}=s_{j}+p_{j}$;
			\State $\pi=\pi \setminus \{J_{j}\}$;
			\State Update the earliest starting times of the remaining jobs that conflict with $J_j$ or correspond to the machine on which $J_j$ has been scheduled;
			\EndWhile      
		\end{algorithmic}
		\hspace*{\algorithmicindent} \textbf{Outputs}: $\sum_{j=1}^{n}C_{j}$.
	\end{algorithm}

	Let $s_{j}^i$ be the earliest time that job $J_j$ can start on machine $M_i$, and let $s_{j} =\min_{i=0}^m s_{j}^i$ denote its overall earliest start time. In Algorithm~\ref{evaluation:fifo}, jobs are inserted sequentially according to the order in $\pi$, with each job starting as early as possible while respecting the conflict constraints. Clearly, the resulting schedule is active. Once a job is scheduled, the start times of the remaining jobs that either share the same machine or conflict with the inserted job, and whose intervals overlap with the scheduled job, are updated, while considering the conflict constraints with the already scheduled jobs.

	\begin{algorithm}[H]
		\caption{Build active schedules by Giffler and Thompson mechanism~\cite{GT1960}}\label{evaluation:GT} 
		\hspace*{\algorithmicindent} \textbf{Inputs}:  permutation $\pi$ of jobs,  processing times $(p_{j})_{1 \leq j \leq n}$, adjacency matrix $A$ of $G$  
		\begin{algorithmic}[1]
			\State Initialize the earliest starting times $(s_{j}^i)_{1\leq i \leq m, 1 \leq j \leq n}$ to zero;
			\While {$\pi \neq \emptyset $} 
			\State Select a job $J_{j'}$ of $\pi$ with the minimum earliest completion time $s_{j'}+p_{j'}$; \label{ctactive1alg2}
			\State Select the first job $J_{j}$ of $\pi$ that is in conflict with $J_{j'}$ (including $J_{j'}$) such that $s_{j} < s_{j'}+p_{j'}$; \label{ctactive2} %
			\State $C_{j}=s_{j}+p_{j}$;
			\State $\pi=\pi \setminus \{J_{j}\}$;
			\For{ $(i'',j'') \in (M \times \pi)$ such that  $J_{j''}$ is in conflict with $J_{j}$ or $J_{j}$ is scheduled on $M_{i''}$} \label{feasibilityalg2}
			\If{$s_{j''}^{i''} < s_{j}+p_{j}$ } \label{step:updatesj1}
			\State $s_{j''}^{i''}=s_{j}+p_{j}$; 
			\State Update $s_{j''}$;  \label{step:updatesj2}
			\EndIf
			\EndFor
			\EndWhile      
		\end{algorithmic}
		\hspace*{\algorithmicindent} \textbf{Outputs}: $\sum_{j=1}^{n}C_{j}$.
	\end{algorithm}

	Algorithm~\ref{evaluation:GT} generates active schedules following the procedure of~\citet{GT1960}. This algorithm selects at each iteration the first job in $\pi$ with the earliest completion time, denoted as $J_{j'}$. Then, the first job of $\pi$ that conflicts with $J_{j'}$ (including $J_{j'}$ itself) and has a starting time earlier than the expected completion time of $J_{j'}$ is inserted into the schedule. After that, the starting times on the machine processing the inserted job, as well as the starting times of the jobs in conflict with the inserted job, are updated as shown in Steps~\ref{step:updatesj1}-\ref{step:updatesj2} of Algorithm~\ref{evaluation:GT}.
	
	A similar algorithm for the open shop with conflict graph was proposed in~\cite{tellache2023genetic}, where Lemma 1 shows that the produced schedule is active. This proof can be easily adapted to demonstrate that Algorithm~\ref{evaluation:GT} also generates active schedules.

	Algorithm~\ref{evaluation:ectf} generates active schedules similar to Algorithm~\ref{evaluation:GT}, except that at each iteration, it schedules the first job of $\pi$ with the minimum earliest completion time. We show in the following that Algorithm~\ref{evaluation:ectf} produces active schedules.

	\begin{lemma}
		Algorithm~\ref{evaluation:ectf} produces feasible schedules that are active.
	\end{lemma}

	\begin{proof}
		The for loop in Step~\ref{feasibility} ensures that each remaining job in $\pi$ starts no earlier than the completion time of all conflicting jobs that have already been scheduled, thereby guaranteeing the feasibility of the schedule.
		
		In Step~\ref{ctactive1}, the selected job $J_j$ has the earliest completion time among all remaining jobs. This implies that all other jobs have either the same or a later earliest completion time. In the for loop of Step~\ref{feasibility}, we update the earliest starting times that are less than the completion time of $J_j$, setting them to be equal to the completion time of $J_j$ if either (a) the job is assigned to the same machine as $J_j$ or (b) it conflicts with $J_j$.
		
		The jobs whose earliest starting times were modified cannot start earlier at their previous times without making the schedule infeasible. This is because these jobs initially had starting times before the completion of $J_j$ but finished afterward. Scheduling them at their previous starting times would either (1) place them on the same machine as $J_j$ at the same time or (2) assign them to different machines but process them simultaneously with $J_j$, which is infeasible due to conflicts. Thus, the produced schedule is active.
		
	\end{proof}

	\begin{algorithm}[H]
		\caption{Build active schedules by earliest completion time first}\label{evaluation:ectf} 
		\hspace*{\algorithmicindent} \textbf{Inputs}:  permutation $\pi$ of jobs,  processing times $(p_{j})_{1 \leq j \leq n}$, adjacency matrix $A$ of $G$  
		\begin{algorithmic}[1]
			\State Initialize the earliest starting times $(s_{j}^i)_{1\leq i \leq m, 1 \leq j \leq n}$ to zero;
			\While {$\pi \neq \emptyset $} 
			\State Select the first job $J_{j}$ of $\pi$ with the minimum earliest completion time $s_{j'}+p_{j'}$; \label{ctactive1}
			\State $C_{j}=s_{j}+p_{j}$;
			\State $\pi=\pi \setminus \{J_{j}\}$;
			\For{ $(i',j') \in (M \times \pi)$ such that  $J_{j'}$ is in conflict with $J_{j}$ or $J_{j}$ is scheduled on $M_{i'}$} \label{feasibility}
			\If{$s_{j'}^{i'} < s_{j}+p_{j}$ }
			\State $s_{j'}^{i'}=s_{j}+p_{j}$;
			\State Update $s_{j'}$;  
			\EndIf
			\EndFor
			\EndWhile      
		\end{algorithmic}
		\hspace*{\algorithmicindent} \textbf{Outputs}: $\sum_{j=1}^{n}C_{j}$.
	\end{algorithm}

	\begin{algorithm}[H]
	\caption{Build non-delay schedules}\label{evaluation:ND} 
	\hspace*{\algorithmicindent} \textbf{Inputs}:  permutation $\pi$ of jobs, processing times $(p_{j})_{1 \leq j \leq n}$, adjacency matrix $A$ of $G$ 
	\begin{algorithmic}[1]
		\State Initialize the earliest starting times $(s_{j}^i)_{1\leq i \leq m, 1 \leq j \leq n}$  to zero;
		\While {$\pi \neq \emptyset $} 
		\State Select the first job $J_{j}$ of $\pi$ with the minimum earliest starting time $s_{j}$;\label{CTnondelay}
		\State $C_{j}=s_{j}+p_{j}$;
		\State $\pi=\pi\setminus \{J_{j}\}$;
		\For{ $(i',j') \in (M \times \pi)$ such that  $J_{j'}$ is in conflict with $J_{j}$ or $J_{j}$ is scheduled on $M_{i'}$} \label{Fornondelay}
		\If{$s_{j'}^{i'} < s_{j}+p_{j}$ }
		\State $s_{j'}^{i'}=s_{j}+p_{j}$;
		\State Update $s_{j'}$; 
		\EndIf
		\EndFor
		\EndWhile
	\end{algorithmic}
	\hspace*{\algorithmicindent} \textbf{Outputs}: $\sum_{j=1}^{n}C_{j}$.
\end{algorithm}

	Algorithm~\ref{evaluation:ND} generates non-delay schedules by ensuring that no machine remains idle when a job is available to start on it. At each iteration, the algorithm selects the first job in $\pi$ with the earliest starting time. Then, it updates the starting times on the machine processing the selected job, as well as those in conflict with the selected job, setting their starting times to the completion time of the selected job if their current starting times are earlier. We can show that the produced schedule is non-delay; the proof follows directly from Lemma 2 of~\cite{tellache2023genetic} and can be easily adapted.

	\subsection{Generation of the initial population}\label{sec:initialpopulation}
	The initial population is either generated entirely at random or seeded with chromosomes constructed by sorting jobs according to eight different rules. These include sorting in decreasing or increasing order of processing time $p_{j}$, conflict degree $c_{j}$ (the number of jobs that cannot be processed in parallel with job $J_{j}$), the ratio $c_{j}/p_{j}$, and the ratio $a_{j}/p_{j}$, where $a_{j}$ represents the agreement degree of $J_{j}$ (the number of jobs that can be processed in parallel with $J_{j}$). Therefore, the generation of the initial population requires both a chromosome generation procedure, i.e., permutations of the $n$ jobs encoding candidate solutions (generated either randomly under a uniform distribution and/or using the eight sorting rules described above), and a decoding procedure based on Algorithms~\ref{evaluation:fifo},~\ref{evaluation:GT},~\ref{evaluation:ectf}, and~\ref{evaluation:ND} to assign a fitness value to each chromosome. The specific algorithm used depends on the GA variant and is detailed in the experimental section.
	
As shown in Algorithm~\ref{GAalgo}, admission to the population
is based on fitness values, with at most one chromosome retained
per fitness value, as in other scheduling and routing genetic
algorithms~\citep{prins2000competitive,liu2009effective}.
This criterion was also supported by \citet{prins2004simple},
who reported no significant improvement from using exact
duplicate detection instead of fitness-based criteria. When a new chromosome is generated, its fitness value is checked against those already present in the population. If a duplicate is detected, we attempt up to $T_{max}$ times to generate a chromosome with a unique fitness value. The goal is to construct a population of size $N_p$; however, if after  $T_{max}$ attempts a chromosome with a different fitness value cannot be obtained, the final population size is set to the number of successfully generated chromosomes with distinct fitness value.

	\subsection{Crossover and mutation operators}\label{sec:crossovermutation}
	The crossover operator is essential for effectively exploring the search space. When two parent chromosomes are selected, crossover generates two offspring by exchanging genetic information between them. We consider three crossover operators: one-point crossover (\texttt{X1}), introduced by \citet{Holland1975} and later adapted for chromosome permutations by \citet{REEVES1995}; order crossover (\texttt{OX}) \cite{Goldberg1989}; and linear order crossover (\texttt{LOX}) \cite{Falkenauer1991}.

	\begin{figure}[H]
		\begin{center}
	\begin{tikzpicture}[scale=0.75, transform shape,
		gene/.style={draw, minimum width=0.7cm, minimum height=0.55cm, inner sep=0pt},
		sel/.style={gene, fill=gray!40},
		lbl/.style={font=\small}
		]

		\begin{scope}[xshift=0cm]
			
			\node[lbl,left] at (-0.6,0) {Parent 1};
			\foreach \x/\txt in {0/J_{2},0.7/J_{1},1.4/J_{5},2.1/J_{3}}
			{
				\node[sel] at (\x,0) {$\txt$};
			}
		
		    \node[gene] at (2.8,0) {$J_{4}$};
		    \node[gene] at (3.5,0) {$J_{6}$};
			
			\draw[->] (2.1,0.9)--(2.1,0.35);
			\node at (2.1,1.1) {$p$};
			
			\node[lbl,left] at (-0.6,-1) {Parent 2};
			\foreach \x/\txt in {0/J_{1},0.7/J_{4},1.4/J_{6},2.1/J_{5}}
			{
				\node[gene,fill=gray!80] at (\x,-1) {$\txt$};
			}
		     \node[gene] at (2.8,-1) {$J_{3}$};
		     \node[gene] at (3.5,-1) {$J_{2}$};
			
			\node[lbl,left] at (-0.6,-2.7) {Child 1};
			\foreach \x/\txt in {0/J_{2},0.7/J_{1},1.4/J_{5},2.1/J_{3}}
			{
				\node[sel] at (\x,-2.7) {$\txt$};
			}
			\node[gene] at (2.8,-2.7) {$J_{4}$};
			\node[gene] at (3.5,-2.7) {$J_{6}$};
			
			\node[lbl,left] at (-0.6,-3.7) {Child 2};
			\foreach \x/\txt in {0/J_{1},0.7/J_{4},1.4/J_{6},2.1/J_{5}}
			{
				\node[gene,fill=gray!80] at (\x,-3.7) {$\txt$};
			}
			\node[gene] at (2.8,-3.7) {$J_{2}$};
			\node[gene] at (3.5,-3.7) {$J_{3}$};
			
			\node at (1.8,-5) {(a)};
		\end{scope}

		\begin{scope}[xshift=6cm]
			
			\node[lbl,left] at (-0.6,0) {Parent 1};
			\foreach[count=\i] \txt in {J_{2},J_{1},J_{5},J_{3},J_{6},J_{4}}
			{
				\node[gene] at ({0.7*(\i-1)},0) {$\txt$};
			}
			
			\draw[->] (1.4,0.9)--(1.4,0.35);
			\draw[->] (2.1,0.9)--(2.1,0.35);
			\node at (1.4,1.1) {$p$};
			\node at (2.1,1.1) {$q$};
			
			 \node[sel] at (1.4,0) {$J_{5}$};
			 \node[sel] at (2.1,0) {$J_{3}$};
			
			\node[lbl,left] at (-0.6,-1) {Parent 2};
			\foreach[count=\i] \txt in {J_{1},J_{4},J_{6},J_{5},J_{3},J_{2}}
			{
				\node[gene] at ({0.7*(\i-1)},-1) {$\txt$};
			}
		
		    \node[gene,fill=gray!80] at (1.4,-1) {$J_{6}$};
		    \node[gene,fill=gray!80] at (2.1,-1) {$J_{5}$};
			
			\node[lbl,left] at (-0.6,-2.7) {Child 1};
			
			\node[gene] at (0,-2.7) {$J_{1}$};
			\node[gene] at (0.7,-2.7) {$J_{4}$};
			
			\node[sel] at (1.4,-2.7) {$J_{5}$};
			\node[sel] at (2.1,-2.7) {$J_{3}$};
			
			\node[gene] at (2.8,-2.7) {$J_{6}$};
			\node[gene] at (3.5,-2.7) {$J_{2}$};
			
			\node[lbl,left] at (-0.6,-3.7) {Child 2};
			
			\node[gene] at (0,-3.7) {$J_{2}$};
			\node[gene] at (0.7,-3.7) {$J_{1}$};
			
			\node[gene,fill=gray!80] at (1.4,-3.7) {$J_{6}$};
			\node[gene,fill=gray!80] at (2.1,-3.7) {$J_{5}$};
			
			\node[gene] at (2.8,-3.7) {$J_{3}$};
			\node[gene] at (3.5,-3.7) {$J_{4}$};
			
			\node at (1.8,-5) {(b)};
		\end{scope}

		\begin{scope}[xshift=12cm]
			
			\node[lbl,left] at (-0.6,0) {Parent 1};
			\foreach[count=\i] \txt in {J_{2},J_{1},J_{5},J_{3},J_{6},J_{4}}
			{
				\node[gene] at ({0.7*(\i-1)},0) {$\txt$};
			}
			
			\draw[->] (1.4,0.9)--(1.4,0.35);
			\draw[->] (2.1,0.9)--(2.1,0.35);
			\node at (1.4,1.1) {$p$};
			\node at (2.1,1.1) {$q$};
			 \node[sel] at (1.4,0) {$J_{5}$};
			\node[sel] at (2.1,0) {$J_{3}$};
			
			\node[lbl,left] at (-0.6,-1) {Parent 2};
			\foreach[count=\i] \txt in {J_{1},J_{4},J_{6},J_{5},J_{3},J_{2}}
			{
				\node[gene] at ({0.7*(\i-1)},-1) {$\txt$};
			}
			
			 \node[gene,fill=gray!80] at (1.4,-1) {$J_{6}$};
			\node[gene,fill=gray!80] at (2.1,-1) {$J_{5}$};
			
			\node[lbl,left] at (-0.6,-2.7) {Child 1};
			
			\node[gene] at (0,-2.7) {$J_{4}$};
			\node[gene] at (0.7,-2.7) {$J_{6}$};
			
			\node[sel] at (1.4,-2.7) {$J_{5}$};
			\node[sel] at (2.1,-2.7) {$J_{3}$};
			
			\node[gene] at (2.8,-2.7) {$J_{2}$};
			\node[gene] at (3.5,-2.7) {$J_{1}$};
			
			\node[lbl,left] at (-0.6,-3.7) {Child 2};
			
			\node[gene] at (0,-3.7) {$J_{1}$};
			\node[gene] at (0.7,-3.7) {$J_{3}$};
			
			\node[gene,fill=gray!90] at (1.4,-3.7) {$J_{6}$};
			\node[gene,fill=gray!90] at (2.1,-3.7) {$J_{5}$};
			
			\node[gene] at (2.8,-3.7) {$J_{4}$};
			\node[gene] at (3.5,-3.7) {$J_{2}$};
			
			\node at (1.8,-5) {(c)};
		\end{scope}
		
	\end{tikzpicture}
	\caption{Crossover operators: (a) one-point crossover (X1), (b) linear order crossover (LOX), and (c) order crossover (OX).}
	\label{fig:crossover}
\end{center}
\end{figure}

In \texttt{X1}, a random cut point divides the parents, and subsequences are exchanged to generate two offspring. Each offspring inherits the prefix from one parent, while the remaining positions are completed using the order of appearance of the missing genes in the other parent. This mechanism ensures feasibility by guaranteeing that each element appears exactly once in the resulting permutations. In \texttt{LOX}, a random subsequence from one parent is directly copied to the child while maintaining its position, and the remaining jobs are inserted in empty slots in the order they appear in the second parent, scanned from left to right. \texttt{OX} follows a similar approach as \texttt{LOX} but inserts the remaining jobs after the copied subsequence, handling chromosomes as cyclic sequences (see Figure~\ref{fig:crossover}).
	
	\begin{figure}[H]
	\begin{center}
\begin{tikzpicture}[
	scale=0.7,
	transform shape,
	gene/.style={
		draw,
		minimum width=1cm,
		minimum height=0.8cm,
		inner sep=0pt
	}
	]

		\node[left] at (-1,0.1) {Chromosome};
		
		\foreach \i/\txt in {
			0/$J_{2}$,
			1/$J_{1}$,
			2/$J_{5}$,
			3/$J_{3}$,
			4/$J_{4}$,
			5/$J_{6}$}
		{
			\ifnum\i=1
			\node[gene,fill=gray!50] (a\i) at (\i,0) {\txt};
			\else\ifnum\i=4
			\node[gene,fill=gray!80] (a\i) at (\i,0) {\txt};
			\else
			\node[gene] (a\i) at (\i,0) {\txt};
			\fi\fi
		}
		
		\draw[->,thick] (1,1.0) -- (1,0.45);
		\node at (1,1.25) {$p$};
		
		\draw[->,thick] (4,1.0) -- (4,0.45);
		\node at (4,1.25) {$q$};
		
		\node[left] at (-1,-2) {Mutated chromosome};
		
		\foreach \i/\txt in {
			0/$J_{2}$,
			1/$J_{4}$,
			2/$J_{5}$,
			3/$J_{3}$,
			4/$J_{1}$,
			5/$J_{6}$}
		{
			\ifnum\i=1
			\node[gene,fill=gray!80] (am\i) at (\i,-2) {\txt};
			\else\ifnum\i=4
			\node[gene,fill=gray!50] (am\i) at (\i,-2) {\txt};
			\else
			\node[gene] (am\i) at (\i,-2) {\txt};
			\fi\fi
		}
		
		\node at (2.5,-3.0) {( a )};

		\begin{scope}[xshift=9cm]
			
			\node[left] at (-1.8,0.4) {};
			
			\foreach \i/\txt in {
				0/$J_{2}$,
				1/$J_{1}$,
				2/$J_{5}$,
				3/$J_{3}$,
				4/$J_{4}$,
				5/$J_{6}$}
			{
					\ifnum\i=4
				\node[gene,fill=gray!50] (bm\i) at (\i,0) {\txt};
				\else\ifnum\i=3
				\node[gene] (bm\i) at (\i,0) {\txt};
				\else
				\node[gene] (bm\i) at (\i,0) {\txt};
				\fi\fi
			}
			
			\draw[->,thick] (1,1.0) -- (1,0.45);
			\node at (1,1.25) {$p$};
			
			\draw[->,thick] (4,1.0) -- (4,0.45);
			\node at (4,1.25) {$q$};
			
			\foreach \i/\txt in {
				0/$J_{2}$,
				1/$J_{4}$,
				2/$J_{1}$,
				3/$J_{5}$,
				4/$J_{3}$,
				5/$J_{6}$}
			{
				\ifnum\i=1
				\node[gene,fill=gray!50] (bm\i) at (\i,-2) {\txt};
				\else\ifnum\i=3
				\node[gene] (bm\i) at (\i,-2) {\txt};
				\else
				\node[gene] (bm\i) at (\i,-2) {\txt};
				\fi\fi
			}
			
			\node at (2.5,-3.0) {( b )};
			
		\end{scope}
	\end{tikzpicture}
	\end{center}
\caption{Mutation operators: (a) swap . (b) move.} \label{fig:mutation}

Mutation helps maintain diversity within the population and operates on a single chromosome with a probability $p_m$. We consider two mutation operators: \texttt{swap} \cite{REEVES1995} and \texttt{move} \cite{REEVES1995}. In both cases, two positions are selected at random. In the \texttt{swap} operator, the jobs at these positions are exchanged, whereas in the \texttt{move} mutation, the job from one position is moved to the other (see Figure~\ref{fig:mutation}).

    \end{figure}
	\subsection{Selection and replacement mechanisms}\label{sec:selectionreplacement}
	Chromosomes in the current population are selected for reproduction and replacement based on their fitness values, using the ranking mechanism outlined by \citet{REEVES1995}. The population is sorted in descending order of fitness, with the chromosome in position $k$ assigned rank $k$. To select the two parents for crossover, the first parent is chosen with probability $2k/(N_p(N_p+1))$, where $k$ is the rank of the chromosome, and the second parent is selected with probability $1/N_p$. One of the two offspring is then selected at random and may undergo mutation with probability $p_m$. Finally, the chromosome to be replaced is randomly selected from those with ranks below the median, $\lfloor N_p/2 \rfloor $.

	\subsection{Hybridization with Local Search}\label{sec:LS}
To further improve the performance of the GA, we apply a local search procedure to each chromosome in the final population. In addition to the \texttt{move} and \texttt{swap} neighborhood operators, we introduce two additional operators: \texttt{Or-Opt}~\cite{Ilhan1976}, which is similar to \texttt{move} but inserts two adjacent jobs instead of one, and \texttt{2-Opt}~\cite{croes1958}, which resembles \texttt{swap} but also reverses the order of the jobs between the two swapped positions. 
	
The local search procedure applied to each chromosome in the final population is summarized in Algorithm~\ref{alg:ls}, where $f(s)$ denotes the fitness value of a solution $s$. The procedure proceeds iteratively as follows. At each step, all neighborhood operators are applied and the best neighboring chromosome is selected. If this chromosome improves the current chromosome, it replaces it and the procedure is repeated. Otherwise, the search terminates for that iteration. This entire improvement process is repeated for a predefined number of iterations for each chromosome.

	For each chromosome, both Algorithm~\ref{evaluation:ectf} and Algorithm~\ref{evaluation:ND} are used to evaluate the solution, and the lowest total completion time is retained. Algorithms~\ref{evaluation:fifo} and~\ref{evaluation:GT} are not used at this stage, as the experimental results presented in the next section indicate that they are outperformed by the other two algorithms.
	
	\begin{algorithm}[H]
	\caption{Local search procedure}\label{alg:ls}
	\begin{algorithmic}[1]
		\Require Initial solution $s$, maximum number of iterations $lsiter$
		
		\While{$lsiter$ has not been reached and the lower bound has not been attained}
		\State improved $\leftarrow$ true
		
		\While{improved and the lower bound has not been attained}
		\State improved $\leftarrow$ false
		\State $s_{1} \leftarrow \texttt{Or-Opt}(s)$
		\State $s_{2} \leftarrow \texttt{2-Opt}(s)$
		\State $s_{3} \leftarrow \texttt{Move}(s)$
		\State $s_{4} \leftarrow \texttt{Swap}(s)$
		
		\State $s^{*} \leftarrow \arg\min \{ f(s_1), f(s_2), f(s_3), f(s_4) \}$
		
		\If{$f(s^{*}) < f(s)$}
		\State $s \leftarrow s^{*}$
		\State improved $\leftarrow$ true
		\EndIf
		
		\EndWhile
		
		\EndWhile
		
		\State \Return $s$
	\end{algorithmic}
\end{algorithm}
	
	\section{Computational experiments} \label{sec:experiments}
	
	The mathematical models, lower bounds, and GAs described above were implemented in C++ and evaluated on a server equipped with 775 GiB of RAM and 64 CPU cores running at 2.25 GHz. The mathematical models were solved using Gurobi 11.0.2, with a time limit of 15 minutes per instance and formulation. Each model was warm-started using the feasible solution obtained by Algorithm~\ref{evaluation:GT} with the job permutation $\pi=(1,2,\ldots,n)$. The parameter $T$ was set to the sum of the processing times of all jobs.
	
	The experiments were conducted using instances derived from the benchmark set in~\citet{kowalczyk2018branch}. These instances were randomly generated with \(n \in \{20, 50, 100, 150\}\) and \(m \in \{3, 5, 8, 10, 12\}\). For each combination of \(n\) and \(m\), 20 instances were generated for six different classes. In Class 1, \( p_j \sim U[1, 10] \), in Class 2, \( p_j \sim U[1, 100] \), in Class 3, \( p_j \sim U[10, 20] \), and in Classes 4 and 5, \( p_j \sim U[90, 100] \). Finally, in Class 6, \( p_j \sim U[10, 100] \). In total, this resulted in 2400 instances.
	
	To construct instances of our problem, we augment these benchmark instances with randomly generated conflict graphs. We use the binomial random-graph model \(G(n,p)\), in which each pair of distinct vertices is joined independently with probability \(p\)~\cite{gilbert1959random}. The expected degree of a vertex is therefore \((n-1)p\). We consider three regimes for $p$: a fixed $p>0$, yielding an expected degree that grows linearly with $n$; $p=c/n$, for a fixed $c>0$, yielding an approximately constant expected degree; and $p=\log n/n$, yielding an expected degree that grows approximately logarithmically with $n$. Classical results show that, for \(p=c/n\), all connected components have size \(O(\log n)\) when \(c<1\), whereas a unique component containing a positive fraction of the vertices emerges when \(c>1\), with probability tending to one as \(n\to\infty\)~\cite{erdos1960evolution}. Connectivity has a different threshold, at the scale \(p=\log n/n\)~\cite{ER59}, where \(\log\) denotes the natural logarithm. These results motivate considering both fixed and size-dependent edge probabilities to assess the performance of our approaches under different conflict structures.

For each instance in the benchmark set of~\citet{kowalczyk2018branch}, we generated five conflict graphs for each fixed probability \(p\in\{0.2,0.5,0.8\}\), yielding 36,000 scheduling instances. For the second regime, \(p=c/n\), we considered \(c\in\{0.5,2\}\) and generated one conflict graph for each value of \(c\) and each benchmark instance, yielding 4,800 scheduling instances. For the third regime, \(p=\log n/n\), we generated one conflict graph per benchmark instance, yielding 2,400 scheduling instances. The complete test set therefore comprises 43,200 scheduling instances. These instances, together with the best known lower and upper bounds (or the proven optimal total completion time, when available), are publicly available at this \href{https://nourelhoudatellache.wordpress.com/test-instances/}{link}. We denote the three resulting instance sets by FD, CD, and LD, corresponding to the fixed-density, constant-degree, and logarithmic-degree regimes, respectively.

We first report the computational results for the FD instances, covering the lower bounds in Section~\ref{subsec:lb}, the mathematical formulations in Section~\ref{subsec:expmilps}, the GA parameter settings in Section~\ref{subsec:gasettings}, and the GA with local search in Section~\ref{subsec:galb}. The results for the CD and LD instances are then presented in Section~\ref{sec:CDLOG}.

\begin{table}[!htbp]
	\centering
	\fontsize{8}{9.5}\selectfont
	\setlength{\tabcolsep}{0pt}
	\renewcommand{\arraystretch}{0.94}
	\caption{Comparison between the lower bounds on the FD instances with $n=20$. Each $(p,m)$ group contains 600 instances.}
	\label{results:LB20}
	\noindent\begin{tabular*}{\linewidth}{@{\extracolsep{\fill}}lrrrrrrrrrrrrr@{}}
		\toprule
		& & \multicolumn{3}{c}{$LB_{1}$} & \multicolumn{3}{c}{$LB_{2}$} & \multicolumn{3}{c}{$LB_{3}$} & \multicolumn{3}{c}{$LB_{4}$} \\
		\cmidrule(lr){3-5} \cmidrule(lr){6-8} \cmidrule(lr){9-11} \cmidrule(lr){12-14}
		$p$ & $m$ & $b$ & $c$ & $d$ & $b$ & $c$ & $d$ & $b$ & $c$ & $d$ & $b$ & $c$ & $d$ \\ \midrule
		0.2 & 3 & 11.333 & 0.406 & 0.000 & 0.000 & 89.658 & 0.000 & 0.000 & 89.542 & 0.000 & 0.000 & 89.792 & 0.000 \\
		& 5 & 0.000 & 2.968 & 0.000 & 0.000 & 85.079 & 0.000 & 0.000 & 84.963 & 0.000 & 0.000 & 85.156 & 0.000 \\
		& 8 & 0.000 & 12.979 & 0.000 & 0.000 & 81.542 & 0.000 & 0.000 & 81.404 & 0.000 & 0.000 & 82.051 & 0.000 \\
		& 10 & 0.000 & 23.468 & 0.000 & 0.000 & 81.643 & 0.000 & 0.000 & 81.397 & 0.000 & 0.000 & 81.856 & 0.000 \\
		& 12 & 0.000 & 28.404 & 0.000 & 0.000 & 81.428 & 0.000 & 0.000 & 81.207 & 0.000 & 0.000 & 81.553 & 0.000 \\
		\midrule
		0.5 & 3 & 0.000 & 3.135 & 0.000 & 0.000 & 76.889 & 0.000 & 0.000 & 76.728 & 0.000 & 0.000 & 78.005 & 0.000 \\
		& 5 & 0.000 & 22.857 & 0.000 & 0.000 & 72.587 & 0.000 & 0.000 & 72.238 & 0.000 & 0.000 & 73.482 & 0.000 \\
		& 8 & 0.000 & 42.997 & 0.000 & 0.000 & 72.255 & 0.000 & 0.000 & 72.017 & 0.000 & 0.000 & 73.536 & 0.000 \\
		& 10 & 0.000 & 50.603 & 0.000 & 0.000 & 72.630 & 0.000 & 0.000 & 72.342 & 0.000 & 0.000 & 73.825 & 0.000 \\
		& 12 & 0.000 & 54.386 & 0.000 & 0.000 & 73.071 & 0.000 & 0.000 & 72.717 & 0.000 & 0.000 & 73.703 & 0.000 \\
		\midrule
		0.8 & 3 & 0.000 & 15.572 & 0.000 & 0.000 & 45.041 & 0.000 & 0.000 & 44.922 & 0.000 & 0.000 & 47.123 & 0.000 \\
		& 5 & 0.000 & 44.000 & 0.000 & 0.000 & 46.353 & 0.000 & 0.000 & 46.208 & 0.000 & 0.000 & 47.972 & 0.000 \\
		& 8 & 0.000 & 59.668 & 0.000 & 0.167 & 46.042 & 0.000 & 0.000 & 45.915 & 0.000 & 0.167 & 47.957 & 0.000 \\
		& 10 & 0.000 & 65.327 & 0.000 & 0.000 & 46.834 & 0.000 & 0.000 & 46.879 & 0.000 & 0.000 & 49.121 & 0.000 \\
		& 12 & 0.000 & 67.756 & 0.000 & 0.000 & 46.356 & 0.000 & 0.000 & 46.428 & 0.000 & 0.000 & 48.232 & 0.000 \\
	\end{tabular*}
	\par\medskip
	\noindent\begin{tabular*}{\linewidth}{@{\extracolsep{\fill}}lrrrrrrrrrrrrr@{}}
		\toprule
		& & \multicolumn{4}{c}{$LB_{5}$} & \multicolumn{4}{c}{$LB_{6}$} & \multicolumn{4}{c}{$LB_{7}$} \\
		\cmidrule(lr){3-6} \cmidrule(lr){7-10} \cmidrule(lr){11-14}
		$p$ & $m$ & $a$ & $b$ & $c$ & $d$ & $a$ & $b$ & $c$ & $d$ & $a$ & $b$ & $c$ & $d$ \\ \midrule
		0.2 & 3 & 0 & 60.000 & 0.090 & 9.586 & 0 & 0.000 & 54.430 & 126.660 & 0 & 0.000 & 69.112 & 0.018 \\
		& 5 & 0 & 7.500 & 1.369 & 7.288 & 0 & 0.000 & 37.130 & 123.955 & 0 & 0.000 & 54.900 & 0.017 \\
		& 8 & 0 & 0.167 & 8.683 & 5.939 & 0 & 0.000 & 27.984 & 124.394 & 0 & 0.000 & 45.224 & 0.017 \\
		& 10 & 0 & 0.000 & 16.565 & 4.717 & 0 & 0.000 & 29.418 & 116.804 & 0 & 0.000 & 44.100 & 0.017 \\
		& 12 & 0 & 0.000 & 16.562 & 4.845 & 0 & 0.000 & 32.292 & 118.378 & 0 & 0.000 & 43.874 & 0.017 \\
		\midrule
		0.5 & 3 & 0 & 0.000 & 2.435 & 25.224 & 0 & 0.000 & 55.650 & 127.256 & 0 & 0.000 & 70.052 & 0.018 \\
		& 5 & 0 & 0.000 & 20.891 & 16.569 & 0 & 0.000 & 49.771 & 125.379 & 0 & 0.000 & 64.439 & 0.018 \\
		& 8 & 0 & 0.000 & 39.497 & 12.655 & 0 & 0.000 & 52.800 & 126.745 & 0 & 0.000 & 64.154 & 0.017 \\
		& 10 & 0 & 0.000 & 45.832 & 11.402 & 0 & 0.000 & 54.456 & 124.199 & 0 & 0.000 & 63.903 & 0.017 \\
		& 12 & 0 & 0.000 & 46.315 & 11.592 & 0 & 0.000 & 56.870 & 119.380 & 0 & 0.000 & 64.225 & 0.017 \\
		\midrule
		0.8 & 3 & 0 & 0.000 & 14.873 & 37.017 & 0 & 0.000 & 61.313 & 128.929 & 0 & 0.000 & 74.166 & 0.018 \\
		& 5 & 0 & 0.000 & 42.543 & 23.818 & 0 & 0.000 & 63.542 & 123.873 & 0 & 0.000 & 74.220 & 0.018 \\
		& 8 & 0 & 0.000 & 57.167 & 19.242 & 0 & 0.000 & 66.620 & 126.507 & 0 & 0.000 & 74.636 & 0.018 \\
		& 10 & 0 & 0.000 & 61.968 & 17.672 & 0 & 0.000 & 68.039 & 125.253 & 0 & 0.000 & 74.656 & 0.017 \\
		& 12 & 0 & 0.000 & 62.068 & 17.842 & 0 & 0.000 & 69.509 & 121.901 & 0 & 0.000 & 74.723 & 0.017 \\
	\end{tabular*}
	\par\medskip
	\noindent\begin{tabular*}{\linewidth}{@{\extracolsep{\fill}}lrrrrrrrrrrrrr@{}}
		\toprule
		& & \multicolumn{4}{c}{$LB'_{5}$} & \multicolumn{4}{c}{$LB'_{6}$} & \multicolumn{4}{c}{$LB'_{7}$} \\
		\cmidrule(lr){3-6} \cmidrule(lr){7-10} \cmidrule(lr){11-14}
		$p$ & $m$ & $a$ & $b$ & $c$ & $d$ & $a$ & $b$ & $c$ & $d$ & $a$ & $b$ & $c$ & $d$ \\ \midrule
		0.2 & 3 & 0 & \textbf{100.000} & \textbf{0.000} & 50.327 & 36 & 0.000 & 59.285 & 910.168 & 0 & 0.000 & 32.647 & 900.158 \\
		& 5 & 0 & \textbf{100.000} & \textbf{0.000} & 56.653 & 35 & 6.333 & 42.472 & 877.440 & 0 & 17.333 & 7.883 & 775.459 \\
		& 8 & 0 & \textbf{100.000} & \textbf{0.000} & 60.683 & 73 & 31.833 & 33.883 & 691.964 & 0 & 98.667 & 0.027 & 17.299 \\
		& 10 & 0 & \textbf{100.000} & \textbf{0.000} & 57.469 & 68 & 34.667 & 34.550 & 671.850 & 0 & \textbf{100.000} & \textbf{0.000} & 1.051 \\
		& 12 & 0 & \textbf{100.000} & \textbf{0.000} & 42.403 & 72 & 34.333 & 34.015 & 665.696 & 0 & \textbf{100.000} & \textbf{0.000} & 0.913 \\
		\midrule
		0.5 & 3 & 0 & \textbf{88.500} & \textbf{0.517} & 577.848 & 21 & 0.000 & 57.568 & 911.183 & 0 & 16.833 & 15.456 & 830.334 \\
		& 5 & 0 & 38.833 & 5.419 & 608.847 & 49 & 1.667 & 54.491 & 905.712 & 0 & \textbf{98.167} & \textbf{0.158} & 182.287 \\
		& 8 & 0 & 44.667 & 4.421 & 585.893 & 48 & 3.500 & 55.135 & 900.596 & 0 & \textbf{99.500} & \textbf{0.033} & 126.281 \\
		& 10 & 0 & 44.833 & 4.624 & 585.126 & 58 & 5.667 & 53.615 & 893.657 & 0 & \textbf{99.000} & \textbf{0.085} & 114.437 \\
		& 12 & 0 & 45.500 & 4.837 & 584.604 & 46 & 3.333 & 55.083 & 904.452 & 0 & \textbf{99.667} & \textbf{0.021} & 119.336 \\
		\midrule
		0.8 & 3 & 0 & \textbf{80.500} & \textbf{0.987} & 751.534 & 23 & 0.000 & 64.721 & 911.479 & 0 & 19.833 & 11.451 & 888.512 \\
		& 5 & 0 & \textbf{68.667} & \textbf{2.980} & 747.448 & 16 & 0.000 & 66.410 & 909.619 & 0 & 33.167 & 10.297 & 882.421 \\
		& 8 & 0 & \textbf{78.667} & \textbf{1.370} & 744.556 & 20 & 0.000 & 67.167 & 911.635 & 0 & 22.333 & 10.994 & 878.103 \\
		& 10 & 0 & \textbf{78.000} & \textbf{1.577} & 743.364 & 18 & 0.000 & 66.049 & 914.115 & 0 & 23.000 & 10.759 & 872.294 \\
		& 12 & 0 & \textbf{79.667} & \textbf{1.082} & 741.217 & 12 & 0.000 & 67.177 & 910.405 & 0 & 21.000 & 11.428 & 878.853 \\
		\bottomrule
	\end{tabular*}
	\par\smallskip
	\begin{minipage}{\linewidth}\fontsize{8}{9.5}\selectfont
		$a$: Number of instances where Gurobi could not find a lower bound within the time limit (for LPs, no optimal LP result). \par 
		$b$ (\%): Number of times in percentage a given lower bound yields the best bound. \par
		$c$ (\%): Average percentage deviation from the best lower bound over instances with a valid $LB$. \par
		$d$ (s): Average CPU time in seconds.
	\end{minipage}
\end{table}

\begin{table}[!htbp]
	\centering
	\fontsize{8}{10}\selectfont
	\setlength{\tabcolsep}{0pt}
	\renewcommand{\arraystretch}{1.0}
	\caption{Comparison of the lower bounds on the FD instances with $n=50$, $m\in\{3,5,8\}$. Each $(p,m)$ contains 600 instances.}
	\label{results:LB50_MDHD}
	\noindent\begin{tabular*}{\linewidth}{@{\extracolsep{\fill}}lrrrrrrrrrrrrr@{}}
		\toprule
		& & \multicolumn{3}{c}{$LB_{1}$} & \multicolumn{3}{c}{$LB_{2}$} & \multicolumn{3}{c}{$LB_{3}$} & \multicolumn{3}{c}{$LB_{4}$} \\
		\cmidrule(lr){3-5} \cmidrule(lr){6-8} \cmidrule(lr){9-11} \cmidrule(lr){12-14}
		$p$ & $m$ & $b$ & $c$ & $d$ & $b$ & $c$ & $d$ & $b$ & $c$ & $d$ & $b$ & $c$ & $d$ \\ \midrule
		0.2 & 3 & 44.833 & 0.040 & 0.000 & 0.000 & 97.197 & 0.000 & 0.000 & 97.140 & 0.000 & 0.000 & 97.274 & 0.000 \\
		& 5 & 13.500 & 0.333 & 0.000 & 0.000 & 95.587 & 0.000 & 0.000 & 95.460 & 0.000 & 0.000 & 95.694 & 0.000 \\
		& 8 & 1.667 & 1.609 & 0.000 & 0.000 & 93.683 & 0.000 & 0.000 & 93.586 & 0.000 & 0.000 & 93.896 & 0.000 \\
		\midrule
		0.5 & 3 & \textbf{61.833} & \textbf{0.141} & 0.000 & 0.000 & 92.421 & 0.000 & 0.000 & 92.252 & 0.000 & 0.000 & 92.982 & 0.000 \\
		& 5 & 34.667 & \textbf{1.203} & 0.000 & 0.000 & 88.015 & 0.000 & 0.000 & 87.790 & 0.000 & 0.000 & 88.953 & 0.000 \\
		& 8 & 25.167 & \textbf{5.191} & 0.000 & 0.000 & 83.158 & 0.000 & 0.000 & 82.819 & 0.000 & 0.000 & 84.284 & 0.000 \\
		\midrule
		0.8 & 3 & \textbf{64.333} & \textbf{0.928} & 0.000 & 0.000 & 72.627 & 0.000 & 0.000 & 72.383 & 0.000 & 0.000 & 74.592 & 0.000 \\
		& 5 & 37.500 & \textbf{5.645} & 0.000 & 0.000 & 58.772 & 0.000 & 0.000 & 58.416 & 0.000 & 0.000 & 61.371 & 0.000 \\
		& 8 & 0.000 & 17.819 & 0.000 & 0.000 & 46.145 & 0.000 & 0.000 & 45.858 & 0.000 & 0.000 & 50.050 & 0.000 \\
	\end{tabular*}
	\par\medskip
	\noindent\begin{tabular*}{\linewidth}{@{\extracolsep{\fill}}lrrrrrrrrr@{}}
		\toprule
		& & \multicolumn{4}{c}{$LB_{5}$} & \multicolumn{4}{c}{$LB_{7}$} \\
		\cmidrule(lr){3-6} \cmidrule(lr){7-10}
		$p$ & $m$ & $a$ & $b$ & $c$ & $d$ & $a$ & $b$ & $c$ & $d$ \\ \midrule
		0.2 & 3 & 224 & 44.000 & \textbf{0.008} & 487.768 & 0 & 0.000 & 86.150 & 0.030 \\
		& 5 & 144 & 22.833 & \textbf{0.102} & 401.866 & 0 & 0.000 & 78.019 & 0.030 \\
		& 8 & 20 & 28.333 & \textbf{0.807} & 296.381 & 0 & 0.000 & 67.896 & 0.040 \\
		\midrule
		0.5 & 3 & 371 & 3.000 & 0.264 & 594.790 & 0 & 0.000 & 86.165 & 0.030 \\
		& 5 & 208 & 20.167 & 1.245 & 486.240 & 0 & 0.000 & 78.228 & 0.031 \\
		& 8 & 151 & 10.167 & 5.704 & 403.686 & 0 & 0.000 & 69.155 & 0.054 \\
		\midrule
		0.8 & 3 & 386 & 2.000 & 2.477 & 610.087 & 0 & 0.000 & 86.286 & 0.033 \\
		& 5 & 226 & 24.833 & 8.495 & 519.688 & 0 & 0.000 & 79.331 & 0.032 \\
		& 8 & 198 & 12.000 & \textbf{15.250} & 477.149 & 0 & 0.000 & 73.151 & 0.042 \\
	\end{tabular*}
	\par\medskip
	\noindent\begin{tabular*}{\linewidth}{@{\extracolsep{\fill}}lrrrrrrrrr@{}}
		\toprule
		& & \multicolumn{4}{c}{$LB'_{5}$} & \multicolumn{4}{c}{$LB'_{7}$} \\
		\cmidrule(lr){3-6} \cmidrule(lr){7-10}
		$p$ & $m$ & $a$ & $b$ & $c$ & $d$ & $a$ & $b$ & $c$ & $d$ \\ \midrule
		0.2 & 3 & 0 & \textbf{62.500} & 37.500 & 598.216 & 0 & 0.000 & 72.368 & 900.134 \\
		& 5 & 2 & \textbf{70.500} & 29.264 & 650.120 & 0 & 0.000 & 56.139 & 900.135 \\
		& 8 & 2 & \textbf{83.167} & 16.555 & 848.698 & 0 & 0.000 & 35.017 & 900.153 \\
		\midrule
		0.5 & 3 & 4 & 35.500 & 64.262 & 892.552 & 0 & 0.000 & 67.430 & 900.144 \\
		& 5 & 4 & \textbf{45.167} & 54.530 & 903.320 & 0 & 0.000 & 48.201 & 900.133 \\
		& 8 & 30 & \textbf{66.333} & 30.175 & 903.987 & 0 & 0.000 & 25.764 & 900.145 \\
		\midrule
		0.8 & 3 & 200 & 33.833 & 49.250 & 909.174 & 0 & 0.000 & 57.841 & 900.133 \\
		& 5 & 200 & \textbf{37.667} & 43.500 & 908.040 & 0 & 0.000 & 35.656 & 900.118 \\
		& 8 & 200 & \textbf{46.500} & 17.981 & 908.298 & 0 & 41.833 & 16.584 & 900.128 \\
		\bottomrule
	\end{tabular*}
	\par\smallskip
	\begin{minipage}{\linewidth}\fontsize{8}{9.5}\selectfont
	$a$: Number of instances where Gurobi could not find a lower bound within the time limit (for LPs, no optimal LP result). \par 
$b$ (\%): Number of times in percentage a given lower bound yields the best bound. \par
$c$ (\%): Average percentage deviation from the best lower bound over instances with a valid $LB$. \par
$d$ (s): Average CPU time in seconds.
\end{minipage}
\end{table}

\begin{table}[!htbp]
	\centering
\fontsize{8}{10}\selectfont
\setlength{\tabcolsep}{0pt}
\renewcommand{\arraystretch}{1.05}
	
	\caption{Comparison between the lower bounds on the instances derived
		from~\citet{kowalczyk2018branch} with $n=50$ and $m \in \{10,12\}$, $n=100$, and $n=150$.}
	\label{results:LB50100150}
	
\begin{tabular*}{\textwidth}{@{\extracolsep\fill}lllrrrrrrrrrrrrrrrr}
		\toprule
		&&&\multicolumn{3}{c}{$LB_1$}
		&\multicolumn{3}{c}{$LB_2$}
		&\multicolumn{3}{c}{$LB_3$}
		&\multicolumn{3}{c}{$LB_4$}
		&\multicolumn{4}{c}{$LB'_7$}\\
		\cline{4-6}\cline{7-9}\cline{10-12}\cline{13-15}\cline{16-19}
		$p$&$n$&$m$
		&$b$&$c$&$d$
		&$b$&$c$&$d$
		&$b$&$c$&$d$
		&$b$&$c$&$d$
		&$a$&$b$&$c$&$d$\\
		\midrule
		$0.2$ & 50 & 10 & \textbf{100.000} & \textbf{0.000} & 0.000 & 0.000 & 92.074 & 0.000 & 0.000 & 91.860 & 0.000 & 0.000 & 92.242 & 0.000 & 0 & 0.000 & 19.765 & 900.146 \\
		& & 12 & \textbf{70.833} & \textbf{1.072} & 0.000 & 0.000 & 91.328 & 0.000 & 0.000 & 91.087 & 0.000 & 0.000 & 91.550 & 0.000 & 0 & 29.667 & 9.147 & 900.132 \\
		\cmidrule{2-19}
		& 100 & 3 & \textbf{100.000} & \textbf{0.000} & 0.000 & 0.000 & 99.084 & 0.000 & 0.000 & 99.051 & 0.000 & 0.000 & 99.142 & 0.000 & 0 & 0.000 & 85.189 & 900.176 \\
		& & 5 & \textbf{100.000} & \textbf{0.000} & 0.000 & 0.000 & 98.534 & 0.000 & 0.000 & 98.454 & 0.000 & 0.000 & 98.555 & 0.000 & 0 & 0.000 & 75.887 & 925.543 \\
		& & 8 & \textbf{100.000} & \textbf{0.000} & 0.000 & 0.000 & 97.697 & 0.000 & 0.000 & 97.658 & 0.000 & 0.000 & 97.821 & 0.000 & 0 & 0.000 & 62.835 & 900.181 \\
		& & 10 & \textbf{100.000} & \textbf{0.000} & 0.000 & 0.000 & 97.158 & 0.000 & 0.000 & 97.068 & 0.000 & 0.000 & 97.337 & 0.000 & 0 & 0.000 & 54.639 & 900.171 \\
\bottomrule
	\end{tabular*}
\par\smallskip
	\begin{minipage}{\linewidth}\fontsize{8}{9.5}\selectfont
		$a$: Number of instances where Gurobi could not find a lower bound within the time limit (for LPs, no optimal LP result). \par 
	$b$ (\%): Number of times in percentage a given lower bound yields the best bound. \par
	$c$ (\%): Average percentage deviation from the best lower bound over instances with a valid $LB$. \par
	$d$ (s): Average CPU time in seconds.
\end{minipage}
\end{table}

\begin{figure}[!htbp]
	\centering
	\resizebox{\linewidth}{!}{%
		%
%
	}
	\caption{Comparison between the lower bounds on the FD instances with $n=20$.}
	\label{fig:lb_machine_bars_n20}
\end{figure}

\begin{figure}[!htbp]
	\centering
	\resizebox{\linewidth}{!}{%
		%
%
	}
	\caption{Comparison between the lower bounds on the FD instances with $n=50$. Dashed separators mark the change from eight compared bounds at $m\in\{3,5,8\}$ to five at $m\in\{10,12\}$.}
	\label{fig:lb_machine_bars_n50}
\end{figure}

\begin{figure}[!htbp]
	\centering
	\resizebox{\linewidth}{!}{%
		%
%
	}
	\caption{Comparison between the lower bounds on the FD instances with $n\in\{100,150\}$. }
	\label{fig:lb_machine_bars_n100150}
\end{figure}

\subsection{Performance of the lower bounds} \label{subsec:lb}

This section evaluates the performance of the lower bounds introduced in Section~\ref{sec:lowerbounds} on the FD instances. The results are reported in Table~\ref{results:LB20} for $n=20$, Table~\ref{results:LB50_MDHD} for $n=50$ and $m \in \{3,5,8\}$, and Table~\ref{results:LB50100150} for the remaining instances. The results are also summarized in Figures~\ref{fig:lb_machine_bars_n20}, \ref{fig:lb_machine_bars_n50}, and~\ref{fig:lb_machine_bars_n100150}

For each lower bound, column $a$ reports the number of instances for which the corresponding bound could not be obtained within the 15-minute time limit. For LP-relaxation bounds, this corresponds to instances for which the relaxation was not solved to optimality, whereas for final MILP solver bounds, it corresponds to instances for which Gurobi did not return a lower bound. Average percentage deviations are computed only over instances for which the corresponding bound is available. The percentage deviation is defined as $100(BestLB-LB_i)/BestLB$, where $BestLB$ denotes the
largest available lower bound for each instance.

The lower bounds reported in these tables depend on the instance size and their computational performance. For $n=20$, we report the results for all lower bounds. However, the bounds derived from \mFB\ were generally outperformed by the others and required more computational resources, with $LB'_6$ already unavailable for 595 of the 9,000 instances with $n=20$. We therefore excluded these bounds from the subsequent experiments. Thus, for $n=50$ and $m\in\{3,5,8\}$, we report the results for the remaining lower bounds. For these instances, the bounds derived from \mFA\ also became computationally demanding, with $LB_5$ and $LB'_5$ unavailable for 1,928 and 642 of the 5,400 instances, respectively. Consequently, for the remaining instance classes, we restrict the comparison to $LB_1$--$LB_4$ and $LB'_7$. 

As shown in Tables~\ref{results:LB20}, \ref{results:LB50_MDHD}, and~\ref{results:LB50100150} and Figures~\ref{fig:lb_machine_bars_n20}, \ref{fig:lb_machine_bars_n50}, and~\ref{fig:lb_machine_bars_n100150}, $LB_1$ generally outperforms $LB_2$, $LB_3$, and $LB_4$, which are based on weighted independent sets. However, the difference tends to decrease as the conflict density $p$ and the number of machines $m$ increase. In particular, for $n=20$, $p=0.8$, and $m\in\{8,10,12\}$, the latter bounds outperform $LB_1$ in terms of average percentage deviation. This behavior is closely related to the conflict density, as increasing the density strengthens the weighted independent set-based bounds. Moreover, increasing the number of machines relaxes the machine-capacity constraints underlying $LB_1$, while conflicts have no impact on this bound, as they are not taken into account in its computation. Among the weighted independent set-based bounds, $LB_2$ and $LB_3$ generally outperform $LB_4$, with $LB_3$ yielding slightly stronger bounds than $LB_2$ in most cases. Recall that $LB_2$ and $LB_3$ are obtained using the greedy heuristics GWMIN and GWMIN2~\cite{SMK03}, respectively, which follow a similar procedure: at each iteration, a vertex is selected according to a specific rule and removed along with its neighbors. The selected vertices then form an independent set. In contrast, $LB_4$ is obtained using GWMAX~\cite{SMK03}, which iteratively removes vertices according to a selection rule until no edges remain, with the remaining vertices forming an independent set.

Comparing the LP-relaxation bounds $LB_5$, $LB_6$, and $LB_7$, we observe that $LB_5$ generally provides the strongest bounds, followed by $LB_6$ and $LB_7$. This ordering can be explained by the structures of the underlying formulations. The time-indexed formulation \mFA\ explicitly accounts for machine capacity and conflict constraints at each time period, which remain effective in its LP relaxation. Although \mFB\ also uses time-indexed variables, it models conflicts through big-$M$ precedence constraints, which become weaker when the precedence variables take fractional values. Similarly, the precedence-based formulation \mFC\ relies on big-$M$ inequalities to model both sequencing and conflict constraints, resulting in a weaker LP relaxation.

Regarding computational efficiency, $LB_7$ is by far the fastest to compute. The higher computational cost of $LB_6$ can be attributed, at least in part, to the larger size of \mFB, whose variables are indexed by pairs of jobs and time periods. Given its weaker bounds and higher computational cost compared with $LB_5$, we excluded $LB_6$ from the subsequent experiments. For $n=50$ and $m\in\{3,5,8\}$, $LB_5$ continues to provide smaller average percentage deviations than $LB_7$ across all instance groups. However, its computational requirements increase considerably with the instance size, and Gurobi fails to solve its LP relaxation to optimality within the time limit for 1,928 of the 5,400 instances. In contrast, $LB_7$ is successfully computed for all instances, with average computation times ranging from 0.030 to 0.054 seconds. Consequently, $LB_5$ was also excluded from the experiments on the remaining instance classes.

The lower bounds obtained from the MILPs show a different trend from their LP-relaxation counterparts. In particular, $LB'_6$ generally provides the weakest bounds, while $LB'_5$ outperforms $LB'_7$ in most instance groups, except mainly for instances with $n=20$ and medium conflict density, where $LB'_7$ yields stronger bounds. This shows that, despite its weaker LP relaxation, \mFC\ can produce competitive lower bounds when solving the MILP. Moreover, for some instances, Gurobi successfully solves the LP relaxation to optimality but fails to return a lower bound from the corresponding MILP within the time limit. This difference is attributed to the distinct solution procedures used in the separate LP and MILP runs, as solving the MILP may involve substantial presolve and root-node processing before a bound is obtained.  Comparing each formulation final bound with its LP, we observe that, for $n=20$, $LB'_5\geq LB_5$ and $LB'_7>LB_7$ on every instance. Thus, despite its weak LP, \mFC\ substantially strengthens its bounds during the solution process. However, improvement is not systematic: $LB'_6$ is weaker than $LB_6$ on some instances, and the same occurs for
$LB'_5$ relative to $LB_5$ when $n=50$. As the instance size increases, the computational limitations of \mFA\ become more pronounced, leading us to retain only $LB'_7$ for the remaining instance classes.

Overall, for $n=20$, $LB'_5$ provides the best lower bound on 76.52\% of the instances, followed by $LB'_7$ on 56.57\%. These bounds also achieve the smallest average percentage deviations from $BestLB$, at 1.85\% and 7.42\%, respectively. In comparison, $LB'_6$, $LB_5$, and $LB_1$ attain $BestLB$ on only 8.09\%, 4.51\%, and 0.76\% of the instances, respectively, while the remaining bounds rarely or never attain the best value. For $n=50$ and $m\in\{3,5,8\}$, $LB'_5$ remains the most frequent best bound, attaining $BestLB$ on 53.46\% of the instances, followed by $LB_1$ and $LB_5$ on 31.50\% and 18.59\%, respectively. However, $LB_1$ achieves the smallest average deviation, at 3.66\%, compared with 3.88\% for $LB_5$, 38.29\% for $LB'_5$, and 46.11\% for $LB'_7$. For $n=50$ and $m\in\{10,12\}$, $LB'_7$ performs best among the five reported bounds, attaining $BestLB$ on 50.39\% of the instances, compared with 46.89\% for $LB_1$. It also achieves the smallest average deviation, at 6.09\%, compared with 8.71\% for $LB_1$. The remaining bounds attain $BestLB$ on fewer than 2\% of the instances each. For $n\in\{100,150\}$, $LB_1$ provides the best reported bound on 98.13\% of the instances, with an average deviation of only 0.19\%, whereas $LB'_7$ attains the best bound on 1.87\%, with an average deviation of 60.56\%. Overall, these results show the effectiveness of the bounds derived from the mathematical formulations for smaller instances and the increasing importance of the computationally inexpensive bound $LB_1$ for larger instances.

\begin{table}[!htbp]
	\centering
	\fontsize{8}{9.5}\selectfont
	\setlength{\tabcolsep}{0pt}
	\renewcommand{\arraystretch}{1.12}
	\caption{Comparison of the mathematical formulations on the FD instances with $n=20$. Each $(p,m)$ contains 600 instances.}
	\label{results:UB2050}
	\noindent\begin{tabular*}{\linewidth}{@{\extracolsep{\fill}}lrrrrrrrrrrrrr@{}}
		\toprule
		& & \multicolumn{4}{c}{$\mathrm{F}_1$} & \multicolumn{4}{c}{$\mathrm{F}_2$} & \multicolumn{4}{c}{$\mathrm{F}_3$} \\
		\cmidrule(lr){3-6} \cmidrule(lr){7-10} \cmidrule(lr){11-14}
		$p$ & $m$ & $a$ & $b$ & $c$ & $d$ & $a$ & $b$ & $c$ & $d$ & $a$ & $b$ & $c$ & $d$ \\ \midrule
		0.2 & 3 & \textbf{100.000} & \textbf{0.000} & \textbf{100.000} & \textbf{0.000} & 10.833 & 1.091 & 0.000 & 59.770 & 32.667 & 0.192 & 0.000 & 32.815 \\
		& 5 & \textbf{100.000} & \textbf{0.000} & \textbf{100.000} & \textbf{0.000} & 9.500 & 3.563 & 6.333 & 44.194 & 49.000 & 0.194 & 17.333 & 8.088 \\
		& 8 & \textbf{100.000} & \textbf{0.000} & \textbf{100.000} & \textbf{0.000} & 35.000 & 5.302 & 31.833 & 35.246 & 99.833 & 0.000 & 98.667 & 0.028 \\
		& 10 & \textbf{100.000} & \textbf{0.000} & \textbf{100.000} & \textbf{0.000} & 39.667 & 5.764 & 34.500 & 35.747 & \textbf{100.000} & \textbf{0.000} & \textbf{100.000} & 0.000 \\
		& 12 & \textbf{100.000} & \textbf{0.000} & \textbf{100.000} & \textbf{0.000} & 38.333 & 5.372 & 34.333 & 35.131 & \textbf{100.000} & \textbf{0.000} & \textbf{100.000} & 0.000 \\
		\midrule
		0.5 & 3 & \textbf{65.500} & 1.043 & \textbf{46.167} & \textbf{2.439} & 0.833 & 5.125 & 0.000 & 60.011 & 47.833 & \textbf{0.622} & 14.333 & 16.862 \\
		& 5 & 45.667 & 4.058 & 38.833 & 8.818 & 4.667 & 10.397 & 1.667 & 57.281 & \textbf{99.000} & \textbf{0.009} & \textbf{97.500} & \textbf{0.214} \\
		& 8 & 54.500 & 2.732 & 44.667 & 6.734 & 8.333 & 9.954 & 3.500 & 57.486 & \textbf{99.833} & \textbf{0.001} & \textbf{99.333} & \textbf{0.046} \\
		& 10 & 55.000 & 2.086 & 44.833 & 6.415 & 9.000 & 10.655 & 5.500 & 56.213 & \textbf{99.667} & \textbf{0.002} & \textbf{98.833} & \textbf{0.096} \\
		& 12 & 53.667 & 2.295 & 45.500 & 6.785 & 6.000 & 10.414 & 3.000 & 57.571 & \textbf{99.833} & \textbf{0.000} & \textbf{99.667} & \textbf{0.021} \\
		\midrule
		0.8 & 3 & 33.500 & 6.415 & \textbf{22.500} & \textbf{15.839} & 1.167 & 9.990 & 0.000 & 70.108 & \textbf{72.333} & \textbf{0.716} & 2.667 & 21.405 \\
		& 5 & 30.833 & 6.472 & \textbf{22.833} & \textbf{16.053} & 1.167 & 10.266 & 0.000 & 70.848 & \textbf{78.500} & \textbf{0.320} & 6.000 & 18.628 \\
		& 8 & 28.667 & 6.535 & \textbf{22.333} & \textbf{13.933} & 0.833 & 10.352 & 0.000 & 71.195 & \textbf{85.667} & \textbf{0.255} & 6.500 & 18.224 \\
		& 10 & 31.667 & 5.690 & \textbf{23.000} & \textbf{13.169} & 0.500 & 10.304 & 0.000 & 70.248 & \textbf{83.667} & \textbf{0.207} & 6.500 & 17.583 \\
		& 12 & 29.000 & 5.755 & \textbf{22.167} & \textbf{12.839} & 1.833 & 10.126 & 0.000 & 71.128 & \textbf{87.333} & \textbf{0.145} & 5.167 & 18.246 \\
		\bottomrule
	\end{tabular*}
	\par\smallskip
	\begin{minipage}{\linewidth}\fontsize{8}{9.5}\selectfont
		$a$ (\%): Number of times in percentage a given upper bound yields the best bound. \par
		$b$ (\%): Average percentage deviation from the best upper bound. \par $c$ (\%): Number of times in percentage Gurobi completed with status ``optimal solution found''.  \par $d$ (\%): Average gap reported by Gurobi in percentage.
	\end{minipage}
\end{table}

\begin{table}[!htbp]
	\centering
	\fontsize{8}{10}\selectfont
	\setlength{\tabcolsep}{0pt}
	\renewcommand{\arraystretch}{1.05}
	\caption{Comparison of $\mathrm{F}_1$ and $\mathrm{F}_3$ on the FD instances with $n=50$, $m\in\{3,5,8\}$. Each $(p,m)$ group contains 600 instances.}
	\label{results:UB50_MDHD}
	\noindent\begin{tabular*}{\linewidth}{@{\extracolsep{\fill}}lrrrrrrrrr@{}}
		\toprule
		& & \multicolumn{4}{c}{$\mathrm{F}_1$} & \multicolumn{4}{c}{$\mathrm{F}_3$} \\
		\cmidrule(lr){3-6} \cmidrule(lr){7-10}
		$p$ & $m$ & $a$ & $b$ & $c$ & $d$ & $a$ & $b$ & $c$ & $d$ \\ \midrule
		0.2 & 3 & 53.500 & 0.160 & \textbf{42.000} & \textbf{37.607} & \textbf{61.000} & \textbf{0.061} & 0.000 & 72.413 \\
		& 5 & 39.500 & 0.776 & \textbf{32.333} & \textbf{30.023} & \textbf{60.833} & \textbf{0.431} & 0.000 & 56.648 \\
		& 8 & 35.500 & 1.852 & \textbf{10.167} & \textbf{20.986} & \textbf{64.667} & \textbf{1.391} & 0.000 & 38.553 \\
		\midrule
		0.5 & 3 & 25.167 & 1.083 & \textbf{3.333} & \textbf{64.832} & \textbf{76.833} & \textbf{0.510} & 0.000 & 68.245 \\
		& 5 & 19.833 & 2.501 & \textbf{0.000} & 62.740 & \textbf{82.167} & \textbf{0.805} & \textbf{0.000} & \textbf{56.487} \\
		& 8 & 15.167 & 4.115 & \textbf{0.000} & 58.668 & \textbf{87.333} & \textbf{0.473} & \textbf{0.000} & \textbf{53.926} \\
		\midrule
		0.8 & 3 & 10.167 & 2.170 & \textbf{0.000} & \textbf{59.433} & \textbf{93.000} & \textbf{0.164} & \textbf{0.000} & 65.406 \\
		& 5 & 2.833 & 3.714 & \textbf{0.000} & 66.973 & \textbf{97.667} & \textbf{0.031} & \textbf{0.000} & \textbf{64.203} \\
		& 8 & 0.667 & 5.331 & \textbf{0.000} & 63.554 & \textbf{99.500} & \textbf{0.005} & \textbf{0.000} & \textbf{63.368} \\
		\bottomrule
	\end{tabular*}
	\par\smallskip
	\begin{minipage}{\linewidth}\fontsize{8}{9.5}\selectfont 
		$a$ (\%): Number of times in percentage a given upper bound yields the best bound. \par
		$b$ (\%): Average percentage deviation from the best upper bound. \par 
		$c$ (\%): Number of times in percentage Gurobi completed with status ``optimal solution found''. \par 
		$d$ (\%): Average gap reported by Gurobi in percentage. \par 
	\end{minipage}
\end{table}

\begin{table}[!htbp]
	\scriptsize
	\setlength{\tabcolsep}{0pt}
	\caption{\label{results:UB100150} Results of \mFC\ on instances derived from~\cite{kowalczyk2018branch} with $n = 50$, $100$, and $150$}
			\begin{tabular*}{\textwidth}{@{\extracolsep{\fill}}lrrrrrrr}
				\toprule
				
				&&\multicolumn{6}{c}{\mFC}\\
				\cline{3-8}
				&&\multicolumn{2}{c}{$p=0.2$}&\multicolumn{2}{c}{$p=0.5$}&\multicolumn{2}{c}{$p=0.8$}\\
				\cline{3-4}  \cline{5-6}  \cline{7-8}
				$n$&$m$&$a$&$b$&$a$&$b$&$a$&$b$\\
				\midrule
				50&$10$  &   0.000  &   31.218&   0.000&   52.900&   0.000  &   63.294\\
				&$12$ &   0.000  &   27.975&   0.000&   52.968&   0.000&   63.059\\
				\midrule
				100&$3$ &   0.000  &   85.200&   0.000&   81.752&   0.000&   77.922\\
				&$5$  &   0.000  &   76.041&   0.000&   73.178 &   0.000&   76.492\\
				&$8$  &   0.000  &   64.240&   0.000&   70.316&   0.000&   76.248\\
				&$10$  &   0.000  &   58.558&   0.000&   70.218&   0.000 &   76.410\\
				&$12$ &   0.000  &   54.531&   0.000 &   70.101&   0.000&   76.239\\
				\midrule
				150&$3$  &   0.000  &   89.959&   0.000&   87.118&   0.000&   83.265\\
				&$5$  &   0.000  &   83.592 &   0.000&   80.282&   0.000&   81.603\\
				&$8$  &   0.000  &   74.954&   0.000&   77.072&   0.000&   81.578\\
				&$10$  &   0.000  &   70.114&   0.000&   76.829&   0.000&   81.475\\
				&$12$  &   0.000  &   66.312&   0.000&   76.837&   0.000 &   81.522\\
				\bottomrule
			\end{tabular*}
			\begin{tablenotes} \fontsize{8}{9.5}\selectfont
				\item $a$(\%): Number of times in percentage Gurobi completed with status "optimal solution found".
				\item $b$(\%): Average gap reported by Gurobi in percentage.
			\end{tablenotes}
\end{table}

\begin{table}[!htbp]
	\centering
	\small
	\setlength{\tabcolsep}{0pt}
	\renewcommand{\arraystretch}{1.10}
	\caption{Comparison of the best upper bounds obtained from the mathematical formulations and the best lower bounds on the FD instances.}
	\label{results:BestLBBestUB}
	\begin{tabular*}{\linewidth}{@{\extracolsep{\fill}}lrrrrrrr@{}}
		\toprule
		& & \multicolumn{2}{c}{$p=0.2$} & \multicolumn{2}{c}{$p=0.5$} & \multicolumn{2}{c}{$p=0.8$} \\
		\cmidrule(lr){3-4}\cmidrule(lr){5-6}\cmidrule(lr){7-8}
		$n$ & $m$ & $a$ (\%) & $b$ (\%) & $a$ (\%) & $b$ (\%) & $a$ (\%) & $b$ (\%) \\
	\midrule
	20 & 3 & 100.000 & 0.000 & 55.167 & 0.974 & 24.833 & 10.108 \\
	& 5 & 100.000 & 0.000 & 99.333 & 0.048 & 27.000 & 8.612 \\
	& 8 & 100.000 & 0.000 & 99.833 & 0.012 & 27.667 & 7.593 \\
	& 10 & 100.000 & 0.000 & 99.833 & 0.009 & 28.500 & 7.183 \\
	& 12 & 100.000 & 0.000 & 100.000 & 0.000 & 26.833 & 7.251 \\
	\midrule
	50 & 3 & 38.833 & 0.069 & 3.333 & 1.981 & 0.000 & 18.831 \\
	& 5 & 32.000 & 0.653 & 0.000 & 15.395 & 0.000 & 44.030 \\
	& 8 & 10.167 & 3.864 & 0.000 & 37.507 & 0.000 & 54.137 \\
	& 10 & 0.000 & 13.648 & 0.000 & 50.469 & 0.000 & 62.761 \\
	& 12 & 0.000 & 20.223 & 0.000 & 52.941 & 0.000 & 62.803 \\
	\midrule
	100 & 3 & 0.333 & 0.060 & 0.000 & 1.213 & 0.000 & 13.407 \\
	& 5 & 0.000 & 0.597 & 0.000 & 10.929 & 0.000 & 43.624 \\
	& 8 & 0.000 & 3.580 & 0.000 & 36.109 & 0.000 & 63.228 \\
	& 10 & 0.000 & 8.274 & 0.000 & 47.603 & 0.000 & 69.893 \\
	& 12 & 0.000 & 14.041 & 0.000 & 55.148 & 0.000 & 72.655 \\
	\midrule
	150 & 3 & 0.667 & 0.038 & 0.000 & 0.763 & 0.000 & 9.737 \\
	& 5 & 0.000 & 0.396 & 0.000 & 7.683 & 0.000 & 40.112 \\
	& 8 & 0.000 & 2.599 & 0.000 & 31.638 & 0.000 & 61.605 \\
	& 10 & 0.000 & 5.464 & 0.000 & 44.158 & 0.000 & 68.736 \\
	& 12 & 0.000 & 10.020 & 0.000 & 52.686 & 0.000 & 73.511 \\
	\midrule
	\multicolumn{2}{l}{Overall} & 29.100 & 4.176 & 22.875 & 22.363 & 6.742 & 39.991 \\
	\bottomrule
	\end{tabular*}
	\par\smallskip
	\begin{minipage}{\linewidth}\fontsize{8}{9.5}\selectfont\raggedright
	$a(\%)$: percentage of instances for which $\mathrm{BestLB}=\mathrm{BestUB}$. \par
	$b(\%)$: mean percentage gap, $100(\mathrm{BestUB}-\mathrm{BestLB})/\mathrm{BestUB}$, over all instances in the group.
	\end{minipage}
\end{table}

\subsection{Performance of the mathematical formulations}
\label{subsec:expmilps}

This section evaluates the performance of the mathematical formulations introduced in Section~\ref{sec:MILP}, with a time limit of 15 minutes per instance and formulation. The results are reported in Table~\ref{results:UB2050} for $n=20$, Table~\ref{results:UB50_MDHD} for $n=50$ and $m\in\{3,5,8\}$, and Table~\ref{results:UB100150} for the remaining instances.

As in the lower-bound experiments, the formulations considered depend on the instance size. For $n=20$, we report the results for all three formulations. Since \mFB\ was generally outperformed by the other formulations and required substantially more computational resources, it was excluded from the subsequent experiments. Thus, for $n=50$ and $m\in\{3,5,8\}$, we report the results for \mFA\ and \mFC. Despite the use of dual simplex for F1 to reduce memory consumption, the formulation remained computationally demanding. Moreover, Gurobi failed to return a finite lower bound within the time limit for 642 of the 5,400 instances, preventing the computation of a finite optimality gap. Consequently, only \mFC\ was retained for the remaining instance classes. It should be noted that, a feasible solution, and hence an upper bound, was available for all instances through the warm start, even when Gurobi did not return a finite lower bound.

The percentage deviation from the best upper bound is calculated as $((UB_i-BestUB)/\allowbreak BestUB)\times 100$, where $UB_i$ denotes the best total completion time obtained by formulation (F$_i$), and $BestUB=\min\{UB_i\mid i=1,\ldots,3\}$ is the best upper bound among the formulations considered for each instance. The optimality gap reported by Gurobi is defined as $(|UB-LB|/|UB|)\times 100$, where $UB$ is the incumbent objective value (i.e., $UB_i$ for formulation (F$_i$)) and $LB$ is the corresponding solver lower bound. Since a feasible solution was available for all instances through the warm start, an upper bound was always obtained. However, for some instances, Gurobi did not return a finite lower bound within the time limit, preventing the computation of a finite optimality gap. The number of such instances corresponds to the values reported in column $a$ for $LB'_5$, $LB'_6$, and $LB'_7$ in Tables~\ref{results:LB20}--\ref{results:LB50100150}. Consequently, the average optimality gaps are computed only over instances for which a finite gap is available. The average computation times of the three formulations are reported in the same tables for the corresponding lower bounds.

As shown in Table~\ref{results:UB2050}, \mFB\ performs worse overall than \mFA\ and \mFC\ on instances with $n=20$, both in upper-bound quality and in proving optimality. Comparing \mFA\ and \mFC\ reveals complementary strengths: \mFA\ proves optimality more frequently, consistent with its stronger linear-relaxation bounds reported in the previous section, whereas \mFC\ provides the best upper bounds overall. The performance of \mFA\ relative to \mFC\ deteriorates as the instance size and conflict density increase, mainly due to the computational demands of its time-indexed structure, with $T=\sum_j p_j$. For $n=20$, \mFA\ solves 55.522\% of instances to proven optimality, including all instances with $p=0.2$, compared with 50.167\% for \mFC\ and 8.044\% for \mFB. Their average solver gaps are 6.868\%, 10.150\%, and 56.812\%, respectively. However, \mFC\ attains the best upper bound in 82.344\% of instances, compared with 61.867\% for \mFA\ and 11.178\% for \mFB, with average deviations of 0.178\%, 2.872\%, and 7.912\%, respectively. Moreover, among the instances for which \mFA\ does not prove optimality, \mFC\ returns a strictly better upper bound in 85.736\% of cases.

For $n=50$ and $m\in\{3,5,8\}$, Table~\ref{results:UB50_MDHD} shows a similar trend. Formulation \mFC\ consistently outperforms \mFA\ in terms of upper-bound quality, attaining the best upper bound on 80.333\% of the instances, compared with 22.481\% for \mFA, with average deviations of 0.430\% and 2.411\%, respectively. However, \mFA\ proves optimality for 9.759\% of these instances, whereas \mFC\ does not prove optimality for any instance. The average Gurobi gaps are 50.131\% for \mFA\ and 59.916\% for \mFC. For the remaining classes, namely $n=50$ with $m\in\{10,12\}$ and all instances with $n\in\{100,150\}$, only \mFC\ was evaluated. As shown in Table~\ref{results:UB100150}, it does not prove optimality for any of these instances and has an average solver gap of 71.427\%.

As shown in the previous section, some lower bounds outperform those obtained from the mathematical formulations for certain instance classes. We therefore compare the best upper bound obtained from the formulations ($\mathrm{BestUB}$) with the best lower bound ($\mathrm{BestLB}$) for each instance. This comparison provides a tighter assessment of solution quality and can establish optimality for additional instances. Table~\ref{results:BestLBBestUB} presents these results. Columns $a$ and $b$ report, respectively, the percentage of instances for which the two bounds coincide and the average percentage gap, computed as $100(\mathrm{BestUB}-\mathrm{BestLB})/\mathrm{BestUB}$. For $n=20$, the bounds coincide in 72.600\% of instances, with an average gap of 2.786\% improving over the results obtained by the three models. For $n=50$ and $m\in\{3,5,8\}$, the bounds coincide in 9.370\% of instances, with an average gap of 19.607\%, improving over the results of both \mFA\ and \mFC. For $n=50$ with $m\in\{10,12\}$ and all instances with $n\in\{100,150\}$, using the best available lower bounds reduces the average gap from 71.427\% reported by Gurobi for \mFC\ to 30.899\% and establishes optimality for 0.028\% of instances, whereas \mFC\ alone certifies none. Across all 36,000 instances, the bounds coincide in 19.572\% of cases, with an average gap of 22.177\%.

\begin{figure}[!htbp]
	\centering
	\resizebox{0.85\linewidth}{!}{%
		\begin{tikzpicture}
			\begin{groupplot}[
				group style={
					group size=1 by 3,
					vertical sep=9mm,
					x descriptions at=edge bottom
				},
				ybar,
				/pgf/bar width=8pt,
				width=0.98\linewidth,
				height=4.2cm,
				xmin=0.5,
				xmax=11.5,
				ymin=0,
				xtick={1,2,3,4,5,6,7,8,9,10,11},
				xticklabels={
					random,
					hybrid,
					\texttt{move},
					\texttt{swap},
					\texttt{X1},
					\texttt{LOX},
					\texttt{OX},
					Algorithm~\ref{evaluation:fifo},
					Algorithm~\ref{evaluation:GT},
					Algorithm~\ref{evaluation:ectf},
					Algorithm~\ref{evaluation:ND}
				},
				xticklabel style={
					rotate=45,
					anchor=east,
					font=\scriptsize
				},
				ylabel={avg\_dev},
				ylabel style={font=\small},
				yticklabel style={
					font=\scriptsize,
					/pgf/number format/fixed,
					/pgf/number format/precision=2
				},
				scaled y ticks=false,
				title style={font=\small},
				ymajorgrids=true,
				grid style={gray!15},
				axis line style={gray!60},
				tick align=outside,
				point meta=y,
				nodes near coords={
					\pgfmathprintnumber[
					fixed,
					precision=3,
					fixed zerofill
					]{\pgfplotspointmeta}
				},
				every node near coord/.append style={
					font=\scriptsize,
					text=black,
					anchor=south,
					yshift=1pt
				}
				]
				
				\nextgroupplot[
				title={$p=0.2$},
				ymax=0.23,
				ytick={0,0.05,0.10,0.15,0.20}
				]
				
				\addplot[
				fill=red!75,
				draw=black,
				line width=0.2pt
				] coordinates {
					(1,0.124103766666667)
					(2,0.0909845)
					(3,0.107881866666667)
					(4,0.1072065)
					(5,0.109867266666667)
					(6,0.102024233333333)
					(7,0.110740766666667)
					(8,0.116765866666667)
					(9,0.193555533333333)
					(10,0.0586784)
					(11,0.0611763333333333)
				};
				
				\draw[black!70,line width=0.8pt]
				(axis cs:2.5,0) -- (axis cs:2.5,0.23);
				
				\draw[black!70,line width=0.8pt]
				(axis cs:4.5,0) -- (axis cs:4.5,0.23);
				
				\draw[black!70,line width=0.8pt]
				(axis cs:7.5,0) -- (axis cs:7.5,0.23);
				
				\nextgroupplot[
				title={$p=0.5$},
				ymax=0.90,
				ytick={0,0.20,0.40,0.60,0.80}
				]
				
				\addplot[
				fill=blue!65,
				draw=black,
				line width=0.2pt
				] coordinates {
					(1,0.578926525)
					(2,0.559191375)
					(3,0.56989085)
					(4,0.568224575)
					(5,0.573963575)
					(6,0.56221175)
					(7,0.5710003)
					(8,0.557366625)
					(9,0.764011125)
					(10,0.49383775)
					(11,0.461014975)
				};
				
				\draw[black!70,line width=0.8pt]
				(axis cs:2.5,0) -- (axis cs:2.5,0.90);
				
				\draw[black!70,line width=0.8pt]
				(axis cs:4.5,0) -- (axis cs:4.5,0.90);
				
				\draw[black!70,line width=0.8pt]
				(axis cs:7.5,0) -- (axis cs:7.5,0.90);
				
				\nextgroupplot[
				title={$p=0.8$},
				ymax=1.90,
				ytick={0,0.30,0.60,0.90,1.20,1.50,1.80}
				]
				
				\addplot[
				fill=brown!80,
				draw=black,
				line width=0.2pt
				] coordinates {
					(1,1.473348)
					(2,1.43280075)
					(3,1.45304225)
					(4,1.4531065)
					(5,1.45408275)
					(6,1.44229525)
					(7,1.46284725)
					(8,1.40397675)
					(9,1.659355)
					(10,1.3941735)
					(11,1.35479425)
				};
				
				\draw[black!70,line width=0.8pt]
				(axis cs:2.5,0) -- (axis cs:2.5,1.90);
				
				\draw[black!70,line width=0.8pt]
				(axis cs:4.5,0) -- (axis cs:4.5,1.90);
				
				\draw[black!70,line width=0.8pt]
				(axis cs:7.5,0) -- (axis cs:7.5,1.90);
				
			\end{groupplot}
		\end{tikzpicture}%
	}
	
	\caption{\label{APDplots}Effects of the variation of the four GA's components on the average deviation from the best lower bound (avg\_dev).}
\end{figure}

\begin{table}[!htbp]
	\centering
	\scriptsize
	\setlength{\tabcolsep}{0pt}
	\caption{Results of \texttt{GA} on the instances derived from~\citet{kowalczyk2018branch} for $p=0.2$.}
	\label{results:GAKLD}
	\begin{tabular*}{\textwidth}{@{\extracolsep{\fill}}lrrrrrrrrrrrrrr@{}}
		\toprule
		& & \multicolumn{5}{c}{Best} & \multicolumn{8}{c}{Average} \\
		\cmidrule(lr){3-7}\cmidrule(lr){8-15}
		$n$ & $m$ & $a$ & $b$ & $c$ & $d$ & $e$ & $a$ & $b$ & $c$ & $d$ & $e$ & $f$ & $g$ & $h$ \\
		\midrule
		20 & 3 & 94.833 & 0.000 & - & - & - & 63.000 & 0.000 & - & - & - & 252.3 & 22373.3 & 1.7 \\
		& 5 & 82.333 & 0.001 & - & - & - & 50.500 & 0.001 & - & - & - & 238.7 & 29872.9 & 2.2 \\
		& 8 & 55.333 & 0.008 & - & - & - & 55.000 & 0.008 & - & - & - & 222.2 & 22674.3 & 1.8 \\
		& 10 & 51.833 & 0.010 & - & - & - & 51.833 & 0.010 & - & - & - & 174.4 & 24086.9 & 1.9 \\
		& 12 & 50.667 & 0.011 & - & - & - & 50.667 & 0.011 & - & - & - & 177.7 & 24680.2 & 2.0 \\
		\midrule
		50 & 3 & 42.060 & 0.000 & 13.624 & 0.000 & 0.001 & 8.584 & 0.001 & 1.362 & 0.000 & 0.001 & 300.0 & 80919.4 & 6.4 \\
		& 5 & 3.125 & 0.002 & 0.245 & 0.003 & 0.007 & 0.000 & 0.004 & 0.000 & 0.003 & 0.007 & 300.0 & 106537.3 & 9.8 \\
		& 8 & 4.918 & 0.008 & 0.000 & 0.023 & 0.022 & 0.000 & 0.013 & 0.000 & 0.025 & 0.020 & 289.9 & 127824.9 & 13.1 \\
		& 10 & - & - & 0.000 & 0.107 & 0.049 & - & - & 0.000 & 0.111 & 0.045 & 288.4 & 105428.3 & 11.7 \\
		& 12 & - & - & 0.000 & 0.197 & 0.052 & - & - & 0.000 & 0.199 & 0.050 & 291.9 & 74314.3 & 8.8 \\
		\midrule
		100 & 3 & 100.000 & 0.000 & 25.920 & 0.000 & 0.000 & 0.000 & 0.000 & 3.679 & 0.000 & 0.000 & 300.0 & 91423.3 & 11.3 \\
		& 5 & - & - & 0.000 & 0.002 & 0.004 & - & - & 0.000 & 0.002 & 0.004 & 300.0 & 122192.4 & 19.8 \\
		& 8 & - & - & 0.000 & 0.013 & 0.023 & - & - & 0.000 & 0.015 & 0.022 & 300.0 & 158820.9 & 31.7 \\
		& 10 & - & - & 0.000 & 0.038 & 0.049 & - & - & 0.000 & 0.043 & 0.045 & 300.0 & 177629.2 & 38.6 \\
		& 12 & - & - & 0.000 & 0.088 & 0.069 & - & - & 0.000 & 0.095 & 0.063 & 300.0 & 171645.6 & 43.2 \\
		\midrule
		150 & 3 & 100.000 & 0.000 & 29.195 & 0.000 & 0.000 & 0.000 & 0.000 & 3.020 & 0.000 & 0.000 & 300.0 & 95094.4 & 17.7 \\
		& 5 & - & - & 0.000 & 0.001 & 0.003 & - & - & 0.000 & 0.001 & 0.003 & 300.0 & 129802.8 & 35.4 \\
		& 8 & - & - & 0.000 & 0.007 & 0.019 & - & - & 0.000 & 0.008 & 0.019 & 300.0 & 159944.3 & 56.5 \\
		& 10 & - & - & 0.000 & 0.021 & 0.035 & - & - & 0.000 & 0.023 & 0.033 & 300.0 & 186285.0 & 76.8 \\
		& 12 & - & - & 0.000 & 0.053 & 0.055 & - & - & 0.000 & 0.058 & 0.050 & 300.0 & 211182.0 & 94.0 \\
		\bottomrule
	\end{tabular*}
\begin{tablenotes} \fontsize{8}{9.5}\selectfont
	\item $a$ (\%): Percentage of instances with proven optimality
	for which the GA TCT equals the optimal TCT.
	
	\item $b$: Average deviation of the GA TCT from the optimal
	TCT over instances with proven optimality.
	
	\item $c$ (\%): Percentage of instances without proven
	optimality for which the GA TCT equals the best lower bound.
	
	\item $d$: Average deviation of the GA TCT from the best
	lower bound over instances without proven optimality.
	
	\item $e$: Average deviation of the GA TCT from the best
	upper bound from the MILPs ($\mathrm{BestUB}$) over instances
	without proven optimality.
	
	\item $f$: Average population size.
	$g$: Average number of generations.
	$h$ (s): CPU time in seconds.
\end{tablenotes}
\end{table}

\begin{table}[!htbp]
	\centering
	\scriptsize
	\setlength{\tabcolsep}{0pt}
	\caption{Results of \texttt{GA} on the instances derived from~\citet{kowalczyk2018branch} for $p=0.5$.}
	\label{results:GAKMD}
	\begin{tabular*}{\textwidth}{@{\extracolsep{\fill}}lrrrrrrrrrrrrrr@{}}
		\toprule
		& & \multicolumn{5}{c}{Best} & \multicolumn{8}{c}{Average} \\
		\cmidrule(lr){3-7}\cmidrule(lr){8-15}
		$n$ & $m$ & $a$ & $b$ & $c$ & $d$ & $e$ & $a$ & $b$ & $c$ & $d$ & $e$ & $f$ & $g$ & $h$ \\
		\midrule
		20 & 3 & 64.048 & 0.003 & 0.000 & 0.020 & 0.003 & 27.190 & 0.004 & 0.000 & 0.020 & 0.003 & 333.5 & 59155.4 & 4.0 \\
		& 5 & 19.799 & 0.019 & 0.000 & 0.082 & -0.004 & 19.799 & 0.019 & 0.000 & 0.082 & -0.004 & 341.1 & 41028.6 & 3.1 \\
		& 8 & 17.863 & 0.023 & 0.000 & 0.085 & -0.004 & 17.863 & 0.023 & 0.000 & 0.085 & -0.004 & 338.9 & 41614.8 & 3.1 \\
		& 10 & 23.873 & 0.020 & 0.000 & 0.076 & -0.018 & 23.873 & 0.020 & 0.000 & 0.076 & -0.018 & 337.6 & 38793.0 & 3.0 \\
		& 12 & 19.167 & 0.021 & - & - & - & 19.167 & 0.021 & - & - & - & 338.8 & 41034.4 & 3.1 \\
		\midrule
		50 & 3 & 0.000 & 0.004 & 0.000 & 0.008 & 0.013 & 0.000 & 0.005 & 0.000 & 0.009 & 0.012 & 400.0 & 119599.7 & 10.2 \\
		& 5 & - & - & 0.000 & 0.110 & 0.066 & - & - & 0.000 & 0.116 & 0.061 & 399.9 & 110761.1 & 10.3 \\
		& 8 & - & - & 0.000 & 0.483 & 0.083 & - & - & 0.000 & 0.489 & 0.079 & 400.0 & 79231.4 & 8.3 \\
		& 10 & - & - & 0.000 & 0.886 & 0.076 & - & - & 0.000 & 0.893 & 0.072 & 400.0 & 79081.4 & 9.0 \\
		& 12 & - & - & 0.000 & 1.003 & 0.072 & - & - & 0.000 & 1.011 & 0.069 & 400.0 & 78935.1 & 9.4 \\
		\midrule
		100 & 3 & - & - & 0.000 & 0.003 & 0.009 & - & - & 0.000 & 0.003 & 0.009 & 400.0 & 135137.8 & 20.1 \\
		& 5 & - & - & 0.000 & 0.053 & 0.065 & - & - & 0.000 & 0.058 & 0.060 & 400.0 & 161070.0 & 27.4 \\
		& 8 & - & - & 0.000 & 0.395 & 0.121 & - & - & 0.000 & 0.413 & 0.110 & 400.0 & 111730.3 & 24.3 \\
		& 10 & - & - & 0.000 & 0.700 & 0.121 & - & - & 0.000 & 0.722 & 0.109 & 400.0 & 111199.4 & 26.2 \\
		& 12 & - & - & 0.000 & 0.990 & 0.119 & - & - & 0.000 & 1.015 & 0.108 & 400.0 & 109012.4 & 28.7 \\
		\midrule
		150 & 3 & - & - & 0.000 & 0.001 & 0.006 & - & - & 0.000 & 0.002 & 0.006 & 400.0 & 141651.2 & 36.0 \\
		& 5 & - & - & 0.000 & 0.031 & 0.049 & - & - & 0.000 & 0.034 & 0.046 & 400.0 & 167429.7 & 49.3 \\
		& 8 & - & - & 0.000 & 0.317 & 0.112 & - & - & 0.000 & 0.333 & 0.101 & 400.0 & 133282.7 & 54.2 \\
		& 10 & - & - & 0.000 & 0.604 & 0.117 & - & - & 0.000 & 0.625 & 0.105 & 400.0 & 128539.6 & 56.9 \\
		& 12 & - & - & 0.000 & 0.895 & 0.116 & - & - & 0.000 & 0.921 & 0.103 & 400.0 & 129733.5 & 65.0 \\
		\bottomrule
	\end{tabular*}
\begin{tablenotes} \fontsize{8}{9.5}\selectfont
	\item $a$ (\%): Percentage of instances with proven optimality
	for which the GA TCT equals the optimal TCT.
	
	\item $b$: Average deviation of the GA TCT from the optimal
	TCT over instances with proven optimality.
	
	\item $c$ (\%): Percentage of instances without proven
	optimality for which the GA TCT equals the best lower bound.
	
	\item $d$: Average deviation of the GA TCT from the best
	lower bound over instances without proven optimality.
	
	\item $e$: Average deviation of the GA TCT from the best
	upper bound from the MILPs ($\mathrm{BestUB}$) over instances
	without proven optimality.
	
	\item $f$: Average population size.
	$g$: Average number of generations.
	$h$ (s): CPU time in seconds.
\end{tablenotes}
\end{table}

\begin{table}[!htbp]
	\centering
	\scriptsize
	\setlength{\tabcolsep}{0pt}
	\caption{Results of \texttt{GA} on the instances derived from~\citet{kowalczyk2018branch} for $p=0.8$.}
	\label{results:GAKHD}
	\begin{tabular*}{\textwidth}{@{\extracolsep{\fill}}lrrrrrrrrrrrrrr@{}}
		\toprule
		& & \multicolumn{5}{c}{Best} & \multicolumn{8}{c}{Average} \\
		\cmidrule(lr){3-7}\cmidrule(lr){8-15}
		$n$ & $m$ & $a$ & $b$ & $c$ & $d$ & $e$ & $a$ & $b$ & $c$ & $d$ & $e$ & $f$ & $g$ & $h$ \\
		\midrule
		20 & 3 & 34.228 & 0.017 & 0.000 & 0.163 & -0.002 & 30.872 & 0.017 & 0.000 & 0.163 & -0.002 & 585.7 & 54433.3 & 3.9 \\
		& 5 & 28.395 & 0.018 & 0.000 & 0.148 & -0.010 & 27.778 & 0.018 & 0.000 & 0.148 & -0.010 & 590.8 & 53567.0 & 3.9 \\
		& 8 & 30.120 & 0.017 & 0.000 & 0.132 & -0.011 & 24.699 & 0.017 & 0.000 & 0.132 & -0.011 & 588.6 & 54221.6 & 4.1 \\
		& 10 & 29.825 & 0.015 & 0.000 & 0.125 & -0.010 & 26.316 & 0.015 & 0.000 & 0.125 & -0.010 & 587.1 & 54162.9 & 4.1 \\
		& 12 & 27.329 & 0.018 & 0.000 & 0.124 & -0.011 & 24.224 & 0.018 & 0.000 & 0.124 & -0.011 & 588.7 & 55252.5 & 4.3 \\
		\midrule
		50 & 3 & - & - & 0.000 & 0.131 & 0.088 & - & - & 0.000 & 0.137 & 0.083 & 700.0 & 136267.5 & 11.0 \\
		& 5 & - & - & 0.000 & 0.630 & 0.103 & - & - & 0.000 & 0.640 & 0.098 & 700.0 & 127000.0 & 11.4 \\
		& 8 & - & - & 0.000 & 1.129 & 0.089 & - & - & 0.000 & 1.140 & 0.084 & 700.0 & 127571.3 & 12.5 \\
		& 10 & - & - & 0.000 & 1.622 & 0.087 & - & - & 0.000 & 1.638 & 0.082 & 700.0 & 126195.0 & 13.0 \\
		& 12 & - & - & 0.000 & 1.647 & 0.084 & - & - & 0.000 & 1.663 & 0.078 & 700.0 & 127737.1 & 14.0 \\
		\midrule
		100 & 3 & - & - & 0.000 & 0.070 & 0.077 & - & - & 0.000 & 0.077 & 0.071 & 700.0 & 189374.8 & 22.1 \\
		& 5 & - & - & 0.000 & 0.552 & 0.135 & - & - & 0.000 & 0.571 & 0.124 & 700.0 & 183164.3 & 26.8 \\
		& 8 & - & - & 0.000 & 1.391 & 0.130 & - & - & 0.000 & 1.422 & 0.118 & 700.0 & 182100.2 & 31.8 \\
		& 10 & - & - & 0.000 & 1.908 & 0.133 & - & - & 0.000 & 1.946 & 0.122 & 700.0 & 183118.3 & 36.4 \\
		& 12 & - & - & 0.000 & 2.268 & 0.132 & - & - & 0.000 & 2.309 & 0.120 & 700.0 & 181676.8 & 39.2 \\
		\midrule
		150 & 3 & - & - & 0.000 & 0.043 & 0.060 & - & - & 0.000 & 0.047 & 0.057 & 700.0 & 203111.0 & 35.9 \\
		& 5 & - & - & 0.000 & 0.471 & 0.129 & - & - & 0.000 & 0.491 & 0.117 & 700.0 & 222217.7 & 52.6 \\
		& 8 & - & - & 0.000 & 1.293 & 0.129 & - & - & 0.000 & 1.324 & 0.117 & 700.0 & 223761.7 & 68.1 \\
		& 10 & - & - & 0.000 & 1.820 & 0.127 & - & - & 0.000 & 1.857 & 0.116 & 700.0 & 223288.0 & 78.9 \\
		& 12 & - & - & 0.000 & 2.328 & 0.127 & - & - & 0.000 & 2.374 & 0.115 & 700.0 & 220933.5 & 89.2 \\
		\bottomrule
	\end{tabular*}
\begin{tablenotes} \fontsize{8}{9.5}\selectfont
	\item $a$ (\%): Percentage of instances with proven optimality
	for which the GA TCT equals the optimal TCT.
	
	\item $b$: Average deviation of the GA TCT from the optimal
	TCT over instances with proven optimality.
	
	\item $c$ (\%): Percentage of instances without proven
	optimality for which the GA TCT equals the best lower bound.
	
	\item $d$: Average deviation of the GA TCT from the best
	lower bound over instances without proven optimality.
	
	\item $e$: Average deviation of the GA TCT from the best
	upper bound from the MILPs ($\mathrm{BestUB}$) over instances
	without proven optimality.
	
	\item $f$: Average population size.
	$g$: Average number of generations.
	$h$ (s): CPU time in seconds.
\end{tablenotes}
\end{table}

\subsection{Setting of the GA}\label{subsec:gasettings}
We first analyze the impact of the four main components of the GA on its performance: (a) the initial population, either randomly generated or seeded with some specific job permutations; (b) the chromosome evaluation algorithm, using Algorithms~\ref{evaluation:fifo},~\ref{evaluation:GT}, and~\ref{evaluation:ectf} for active schedules, and Algorithm~\ref{evaluation:ND} for non-delay schedules; (c) the crossover type (\texttt{X1}, \texttt{OX}, or \texttt{LOX}); and (d) the mutation type (\texttt{swap} or \texttt{move}). We tested all 48 variants on instance groups \((n,m) \in \{(20,3), (50,5), (100,8), (150,10)\}\), each with the three values of \( p \). Each variant was run five times per instance, and the average total completion time was recorded. The parameters on these experiments were set as follows: \( N_p = 50 \), \( I_{\max} = 100 \times N_p \times n \), \( I_{\text{noimp}} = 50000 \), \( p_m = 0.2 \), and \( T_{\max} = 1000 \).

Figure~\ref{APDplots} illustrates the effects of varying the four components on the average percentage deviation from the best lower bound \(BestLB\) across the three values of \(p\). The deviation is calculated as \(((\mathit{tct} - BestLB) / BestLB) \), where \(\mathit{tct}\) is the total completion time obtained by the corresponding variant of the GA. In the figure, "random" refers to an initial population generated randomly, while "hybrid" denotes a population seeded with specific permutations. Note that, for each setting of each of the four components, we report the average over all possible combinations of the remaining components. For instance, in the case of the “random” initial population, the reported value corresponds to the average over $4 \times 3 \times 2 = 24$ different configurations, obtained from all combinations of the four chromosome evaluation algorithms, three crossover operators, and two mutation operators.

We observed from the implementation (see Figure~\ref{APDplots}) that seeding the initial population with specific permutations consistently led to better performance across all instance groups. Regarding the crossover operators, variants using LOX outperformed those using X1 and OX in all cases. For mutation operators, the swap mutation achieved the best results for instances with \(p = 0.2\) and \(p = 0.5\), while the move mutation performed slightly better on instances with \(p = 0.8\). Among the four components, the algorithm used to evaluate the chromosomes had the most significant impact on the GA's performance. Algorithms~\ref{evaluation:fifo} and~\ref{evaluation:GT} were consistently outperformed by Algorithms~\ref{evaluation:ectf} and~\ref{evaluation:ND} across all instance groups. Between the latter two, Algorithm~\ref{evaluation:ectf} showed slightly better performance for \(p = 0.2\), whereas Algorithm~\ref{evaluation:ND} outperformed it for \(p = 0.5\) and \(p = 0.8\).

Regarding the GA parameters, we conducted preliminary experiments in which we varied the population size $N_p \in \{50, 100, 150, 200, 300, 400, 500, 600\}$ and the mutation rate $p_m \in \{0.2, 0.5,1\}$, resulting in 24 combinations. All combinations were tested on the instance groups \((n,m) \in \{(20,3), (50,5), (100,8), (150,10)\}\). In all experiments, the GA was configured with   Algorithm~\ref{evaluation:ND} for chromosome evaluation, an initial population seeded with specific permutations, LOX crossover, and the swap mutation operator. These preliminary experiments, whose detailed results are not
reported here for conciseness, indicate that systematic mutation (\(p_m = 1\)) outperformed the probabilistic mutation (\(p_m \in \{0.2,0.5\}\)) in the majority of instance groups. As for population size, increasing \(N_p\) generally led to higher CPU time—as expected—but did not necessarily improve performance. Specifically, the best average performance was observed for \(N_p = 300\) when \(p = 0.2\), for \(N_p = 400\) when \(p = 0.5\), and for \(N_p = 700\) when \(p = 0.8\).

We present in Tables~\ref{results:GAKLD},~\ref{results:GAKMD}, and~\ref{results:GAKHD} the results corresponding to the GA variant configured with Algorithm~\ref{evaluation:ND} for chromosome evaluation, an initial population seeded with specific permutations, LOX crossover, the swap mutation operator, and a systematic mutation strategy. The population size \(N_p\) was set according to the value of \(p\): 300 for \(p = 0.2\), 400 for \(p = 0.5\), and 700 for \(p = 0.8\). The remaining parameters are the same as those used in the previous experiments: \( I_{\max} = 100 \times N_p \times n \), \( I_{\text{noimp}} = 50000 \), and \( T_{\max} = 1000 \). The GA was run five times for each instance. The “Best” section reports the average results over all instances obtained by selecting the best run among the five executions for each instance, while the “Average” section reports the average results over all instances using the average value across the five runs for each instance. The performance of the GA is then evaluated separately on two sets of instances: instances with a proven optimal TCT for which the best lower bound coincides with the best upper bound obtained from the mathematical formulations, and instances without a proven optimal TCT for which these bounds do not coincide. For the instances with a proven optimal TCT, columns $a$ and $b$ compare the GA TCT with the proven optimal TCT. For the instances without a proven optimal TCT, columns $c$ and $d$ compare the GA TCT with the best lower bound, while column $e$ compares the GA TCT with the best upper bound obtained from the mathematical models. The corresponding relative deviations are computed as $(\mathrm{TCT}-\mathrm{OptTCT})/\mathrm{OptTCT}$, $(\mathrm{TCT}-\mathrm{BestLB})/\mathrm{BestLB}$, and $(\mathrm{BestUB}-\mathrm{TCT})/\mathrm{BestUB}$, respectively. A dash (“--”) in the tables indicates that no value is reported for the corresponding criterion. This occurs in two cases: (i) when no instance has a proven optimal TCT, which applies to columns $a$ and $b$ and corresponds to a value of $0\%$ in column $a$ of Table~\ref{results:BestLBBestUB}; and (ii) when all instances in the group are solved to proven optimality, which applies to columns $c$, $d$, and $e$  and corresponds to a value of $100\%$ in column $a$ of Table~\ref{results:BestLBBestUB}.

From Tables~\ref{results:GAKLD},~\ref{results:GAKMD}, and~\ref{results:GAKHD}, we observe that as the density of the conflict graph increases, the number of instances solved to proven optimality decreases, and the deviation from the best lower bound increases. We attribute this to the quality of the lower bounds: as \(p\) increases, the best-performing lower bound becomes \(LB_1\), which does not account for conflicts. Overall, for \(p = 0.2\) (respectively, \(p = 0.5\) and \(p = 0.8\)), the GA variant solved to optimality 2123 out of 3492 (695 out of 2745, and 242 out of 809) instances for which the best lower and upper bounds coincide, with an average deviation from the optimal total completion time of 0.005 (respectively, 0.018 and 0.017). On instances where the best lower and upper bounds differ, the GA solution coincided with the best lower bound in 380 out of 8508 cases for \(p = 0.2\). The average deviation from the best lower bound on these instances is 0.039 for \(p = 0.2\), 0.421 for \(p = 0.5\), and 0.955 for \(p = 0.8\). Meanwhile, the average deviation from $\mathrm{BestUB}$ is 0.027, 0.074, and 0.086 for \(p = 0.2\), \(p = 0.5\), and \(p = 0.8\), respectively, indicating that the GA improved upon the solutions obtained by the mathematical models within the given time limit of 15 minutes. In contrast, the GA requires at most 94 seconds on average for the largest instances. However, in some instances with \(n = 20\), the deviation from the best upper bound is negative, suggesting that the total completion times returned by the mathematical models—enhanced with warm starts—were better in those specific cases.

\begin{table}[!htbp]
	\centering
	\scriptsize
	\setlength{\tabcolsep}{0pt}
	\caption{Performance of GA-ls on instances derived from~\citet{kowalczyk2018branch} with $p=0.2$.}
	\label{results:GAKLDLS}
	\begin{tabular*}{\textwidth}{@{\extracolsep{\fill}}lrrrrrrrrrrrrrr@{}}
		\toprule
		& & \multicolumn{5}{c}{Best} & \multicolumn{8}{c}{Average} \\
		\cmidrule(lr){3-7}\cmidrule(lr){8-15}
		$n$ & $m$ & $a$ & $b$ & $c$ & $d$ & $e$ & $a$ & $b$ & $c$ & $d$ & $e$ & $f$ & $g$ & $h$ \\
		\midrule
		20 & 3 & 97.500 & 0.000 & - & - & - & 85.333 & 0.000 & - & - & - & 252.4 & 23316.5 & 1.8 \\
		& 5 & 85.167 & 0.001 & - & - & - & 76.000 & 0.001 & - & - & - & 238.7 & 33408.3 & 3.5 \\
		& 8 & 57.167 & 0.007 & - & - & - & 57.167 & 0.007 & - & - & - & 222.1 & 22596.4 & 3.7 \\
		& 10 & 56.333 & 0.008 & - & - & - & 56.333 & 0.008 & - & - & - & 174.4 & 24087.0 & 3.8 \\
		& 12 & 55.000 & 0.008 & - & - & - & 55.000 & 0.008 & - & - & - & 177.7 & 24675.4 & 4.2 \\
		\midrule
		50 & 3 & 89.700 & 0.000 & 25.068 & 0.000 & 0.001 & 75.966 & 0.000 & 19.346 & 0.000 & 0.001 & 300.0 & 82840.1 & 18.0 \\
		& 5 & 15.104 & 0.001 & 0.245 & 0.002 & 0.007 & 4.167 & 0.002 & 0.245 & 0.003 & 0.007 & 300.0 & 106442.4 & 31.8 \\
		& 8 & 11.475 & 0.006 & 0.000 & 0.022 & 0.023 & 0.000 & 0.009 & 0.000 & 0.024 & 0.021 & 289.9 & 127396.1 & 43.1 \\
		& 10 & - & - & 0.000 & 0.106 & 0.050 & - & - & 0.000 & 0.109 & 0.047 & 288.4 & 108105.7 & 46.0 \\
		& 12 & - & - & 0.000 & 0.196 & 0.053 & - & - & 0.000 & 0.197 & 0.052 & 291.7 & 73939.0 & 50.4 \\
		\midrule
		100 & 3 & 100.000 & 0.000 & 58.863 & 0.000 & 0.000 & 100.000 & 0.000 & 58.528 & 0.000 & 0.000 & 300.0 & 92224.8 & 19.4 \\
		& 5 & - & - & 0.167 & 0.001 & 0.005 & - & - & 0.000 & 0.002 & 0.004 & 300.0 & 123240.1 & 53.9 \\
		& 8 & - & - & 0.000 & 0.013 & 0.023 & - & - & 0.000 & 0.015 & 0.022 & 300.0 & 162292.4 & 80.3 \\
		& 10 & - & - & 0.000 & 0.038 & 0.049 & - & - & 0.000 & 0.042 & 0.046 & 300.0 & 176199.2 & 97.3 \\
		& 12 & - & - & 0.000 & 0.087 & 0.070 & - & - & 0.000 & 0.094 & 0.064 & 300.0 & 169424.9 & 108.4 \\
		\midrule
		150 & 3 & 100.000 & 0.000 & 64.765 & 0.000 & 0.000 & 100.000 & 0.000 & 64.765 & 0.000 & 0.000 & 300.0 & 95956.5 & 34.3 \\
		& 5 & - & - & 0.833 & 0.001 & 0.003 & - & - & 0.167 & 0.001 & 0.003 & 300.0 & 127248.6 & 107.1 \\
		& 8 & - & - & 0.000 & 0.007 & 0.019 & - & - & 0.000 & 0.008 & 0.019 & 300.0 & 160233.1 & 159.1 \\
		& 10 & - & - & 0.000 & 0.021 & 0.035 & - & - & 0.000 & 0.023 & 0.033 & 300.0 & 188731.1 & 203.2 \\
		& 12 & - & - & 0.000 & 0.053 & 0.055 & - & - & 0.000 & 0.057 & 0.051 & 300.0 & 209185.7 & 240.1 \\
		\bottomrule
	\end{tabular*}
	\begin{tablenotes} \fontsize{8}{9.5}\selectfont
		\item $a$ (\%): Percentage of instances with proven optimality
		for which the GA TCT equals the optimal TCT.
		
		\item $b$: Average deviation of the GA TCT from the optimal
		TCT over instances with proven optimality.
		
		\item $c$ (\%): Percentage of instances without proven
		optimality for which the GA TCT equals the best lower bound.
		
		\item $d$: Average deviation of the GA TCT from the best
		lower bound over instances without proven optimality.
		
		\item $e$: Average deviation of the GA TCT from the best
		upper bound from the MILPs ($\mathrm{BestUB}$) over instances
		without proven optimality.
		
		\item $f$: Average population size.
		$g$: Average number of generations.
		$h$ (s): CPU time in seconds.
	\end{tablenotes}
\end{table}

\begin{table}[!htbp]
	\centering
	\scriptsize
	\setlength{\tabcolsep}{0pt}
	\caption{Performance of GA-ls on instances derived from~\citet{kowalczyk2018branch} with $p=0.5$.}
	\label{results:GAKMDLS}
	\begin{tabular*}{\textwidth}{@{\extracolsep{\fill}}lrrrrrrrrrrrrrr@{}}
		\toprule
		& & \multicolumn{5}{c}{Best} & \multicolumn{8}{c}{Average} \\
		\cmidrule(lr){3-7}\cmidrule(lr){8-15}
		$n$ & $m$ & $a$ & $b$ & $c$ & $d$ & $e$ & $a$ & $b$ & $c$ & $d$ & $e$ & $f$ & $g$ & $h$ \\
		\midrule
		20 & 3 & 71.299 & 0.003 & 0.000 & 0.020 & 0.003 & 45.015 & 0.003 & 0.000 & 0.020 & 0.003 & 333.6 & 62686.6 & 8.2 \\
		& 5 & 21.141 & 0.017 & 0.000 & 0.082 & -0.004 & 21.141 & 0.017 & 0.000 & 0.082 & -0.004 & 341.2 & 41092.6 & 7.9 \\
		& 8 & 19.866 & 0.020 & 0.000 & 0.085 & -0.004 & 19.866 & 0.020 & 0.000 & 0.085 & -0.004 & 338.9 & 41615.5 & 9.2 \\
		& 10 & 25.876 & 0.017 & 0.000 & 0.076 & -0.018 & 25.876 & 0.017 & 0.000 & 0.076 & -0.018 & 337.6 & 38860.8 & 9.6 \\
		& 12 & 20.167 & 0.018 & - & - & - & 20.167 & 0.018 & - & - & - & 338.8 & 40962.8 & 11.1 \\
		\midrule
		50 & 3 & 5.000 & 0.002 & 0.000 & 0.007 & 0.014 & 0.000 & 0.003 & 0.000 & 0.008 & 0.013 & 400.0 & 119423.4 & 34.1 \\
		& 5 & - & - & 0.000 & 0.108 & 0.067 & - & - & 0.000 & 0.114 & 0.063 & 400.0 & 110470.2 & 44.2 \\
		& 8 & - & - & 0.000 & 0.482 & 0.084 & - & - & 0.000 & 0.486 & 0.081 & 400.0 & 79463.1 & 53.3 \\
		& 10 & - & - & 0.000 & 0.885 & 0.077 & - & - & 0.000 & 0.889 & 0.074 & 400.0 & 79525.2 & 62.3 \\
		& 12 & - & - & 0.000 & 1.001 & 0.073 & - & - & 0.000 & 1.006 & 0.071 & 400.0 & 78749.4 & 68.5 \\
		\midrule
		100 & 3 & - & - & 0.000 & 0.002 & 0.010 & - & - & 0.000 & 0.003 & 0.009 & 400.0 & 135463.8 & 61.3 \\
		& 5 & - & - & 0.000 & 0.052 & 0.065 & - & - & 0.000 & 0.057 & 0.061 & 400.0 & 159880.8 & 84.0 \\
		& 8 & - & - & 0.000 & 0.394 & 0.122 & - & - & 0.000 & 0.410 & 0.112 & 400.0 & 111057.4 & 99.3 \\
		& 10 & - & - & 0.000 & 0.699 & 0.121 & - & - & 0.000 & 0.719 & 0.111 & 400.0 & 110877.7 & 114.8 \\
		& 12 & - & - & 0.000 & 0.989 & 0.119 & - & - & 0.000 & 1.012 & 0.109 & 400.0 & 109210.4 & 126.4 \\
		\midrule
		150 & 3 & - & - & 0.000 & 0.001 & 0.006 & - & - & 0.000 & 0.001 & 0.006 & 400.0 & 143198.5 & 127.0 \\
		& 5 & - & - & 0.000 & 0.030 & 0.050 & - & - & 0.000 & 0.033 & 0.047 & 400.0 & 169113.6 & 177.6 \\
		& 8 & - & - & 0.000 & 0.315 & 0.113 & - & - & 0.000 & 0.331 & 0.102 & 400.0 & 132754.0 & 218.8 \\
		& 10 & - & - & 0.000 & 0.601 & 0.118 & - & - & 0.000 & 0.621 & 0.107 & 400.0 & 128897.0 & 254.4 \\
		& 12 & - & - & 0.000 & 0.893 & 0.116 & - & - & 0.000 & 0.917 & 0.105 & 400.0 & 128730.6 & 283.9 \\
		\bottomrule
	\end{tabular*}
	\begin{tablenotes} \fontsize{8}{9.5}\selectfont
		\item $a$ (\%): Percentage of instances with proven optimality
		for which the GA TCT equals the optimal TCT.
		
		\item $b$: Average deviation of the GA TCT from the optimal
		TCT over instances with proven optimality.
		
		\item $c$ (\%): Percentage of instances without proven
		optimality for which the GA TCT equals the best lower bound.
		
		\item $d$: Average deviation of the GA TCT from the best
		lower bound over instances without proven optimality.
		
		\item $e$: Average deviation of the GA TCT from the best
		upper bound from the MILPs ($\mathrm{BestUB}$) over instances
		without proven optimality.
		
		\item $f$: Average population size.
		$g$: Average number of generations.
		$h$ (s): CPU time in seconds.
	\end{tablenotes}
\end{table}

\begin{table}[!htbp]
	\centering
	\scriptsize
	\setlength{\tabcolsep}{0pt}
	\caption{Performance of GA-ls on instances derived from~\citet{kowalczyk2018branch} with $p=0.8$.}
	\label{results:GAKHDLS}
	\begin{tabular*}{\textwidth}{@{\extracolsep{\fill}}lrrrrrrrrrrrrrr@{}}
		\toprule
		& & \multicolumn{5}{c}{Best} & \multicolumn{8}{c}{Average} \\
		\cmidrule(lr){3-7}\cmidrule(lr){8-15}
		$n$ & $m$ & $a$ & $b$ & $c$ & $d$ & $e$ & $a$ & $b$ & $c$ & $d$ & $e$ & $f$ & $g$ & $h$ \\
		\midrule
		20 & 3 & 37.584 & 0.013 & 0.000 & 0.162 & -0.001 & 36.913 & 0.013 & 0.000 & 0.162 & -0.001 & 585.7 & 55724.2 & 10.9 \\
		& 5 & 31.481 & 0.015 & 0.000 & 0.147 & -0.009 & 31.481 & 0.015 & 0.000 & 0.147 & -0.009 & 590.9 & 55123.0 & 13.1 \\
		& 8 & 31.325 & 0.013 & 0.000 & 0.131 & -0.010 & 30.723 & 0.013 & 0.000 & 0.131 & -0.010 & 588.7 & 55967.6 & 16.1 \\
		& 10 & 30.409 & 0.014 & 0.000 & 0.124 & -0.009 & 30.409 & 0.014 & 0.000 & 0.124 & -0.009 & 587.2 & 55590.2 & 17.3 \\
		& 12 & 27.950 & 0.017 & 0.000 & 0.123 & -0.010 & 27.950 & 0.017 & 0.000 & 0.123 & -0.010 & 588.6 & 55896.9 & 19.9 \\
		\midrule
		50 & 3 & - & - & 0.000 & 0.129 & 0.090 & - & - & 0.000 & 0.134 & 0.086 & 700.0 & 137830.2 & 43.3 \\
		& 5 & - & - & 0.000 & 0.628 & 0.105 & - & - & 0.000 & 0.634 & 0.101 & 700.0 & 127299.5 & 59.7 \\
		& 8 & - & - & 0.000 & 1.125 & 0.091 & - & - & 0.000 & 1.134 & 0.087 & 700.0 & 127448.2 & 78.5 \\
		& 10 & - & - & 0.000 & 1.617 & 0.089 & - & - & 0.000 & 1.629 & 0.085 & 700.0 & 127580.7 & 90.5 \\
		& 12 & - & - & 0.000 & 1.641 & 0.085 & - & - & 0.000 & 1.653 & 0.081 & 700.0 & 128138.7 & 101.7 \\
		\midrule
		100 & 3 & - & - & 0.000 & 0.068 & 0.078 & - & - & 0.000 & 0.074 & 0.073 & 700.0 & 189204.6 & 70.2 \\
		& 5 & - & - & 0.000 & 0.550 & 0.136 & - & - & 0.000 & 0.569 & 0.125 & 700.0 & 183525.8 & 98.4 \\
		& 8 & - & - & 0.000 & 1.389 & 0.130 & - & - & 0.000 & 1.418 & 0.120 & 700.0 & 183952.1 & 134.4 \\
		& 10 & - & - & 0.000 & 1.904 & 0.135 & - & - & 0.000 & 1.940 & 0.124 & 700.0 & 184563.5 & 157.3 \\
		& 12 & - & - & 0.000 & 2.264 & 0.132 & - & - & 0.000 & 2.305 & 0.121 & 700.0 & 182842.8 & 178.0 \\
		\midrule
		150 & 3 & - & - & 0.000 & 0.042 & 0.061 & - & - & 0.000 & 0.046 & 0.058 & 700.0 & 203058.0 & 143.9 \\
		& 5 & - & - & 0.000 & 0.470 & 0.130 & - & - & 0.000 & 0.489 & 0.118 & 700.0 & 221995.6 & 217.4 \\
		& 8 & - & - & 0.000 & 1.289 & 0.130 & - & - & 0.000 & 1.321 & 0.118 & 700.0 & 221178.0 & 294.7 \\
		& 10 & - & - & 0.000 & 1.816 & 0.128 & - & - & 0.000 & 1.855 & 0.116 & 700.0 & 219190.4 & 342.7 \\
		& 12 & - & - & 0.000 & 2.324 & 0.128 & - & - & 0.000 & 2.368 & 0.117 & 700.0 & 223297.9 & 397.0 \\
		\bottomrule
	\end{tabular*}
	\begin{tablenotes} \fontsize{8}{9.5}\selectfont
		\item $a$ (\%): Percentage of instances with proven optimality
		for which the GA TCT equals the optimal TCT.
		
		\item $b$: Average deviation of the GA TCT from the optimal
		TCT over instances with proven optimality.
		
		\item $c$ (\%): Percentage of instances without proven
		optimality for which the GA TCT equals the best lower bound.
		
		\item $d$: Average deviation of the GA TCT from the best
		lower bound over instances without proven optimality.
		
		\item $e$: Average deviation of the GA TCT from the best
		upper bound from the MILPs ($\mathrm{BestUB}$) over instances
		without proven optimality.
		
		\item $f$: Average population size.
		$g$: Average number of generations.
		$h$ (s): CPU time in seconds.
	\end{tablenotes}
\end{table}

\subsection{Performance of GA with Local Search} \label{subsec:galb}

In this section, we evaluate the impact of applying local search to the final population of the GA (denoted as GA-ls). As detailed in Section~\ref{sec:LS}, the local search procedure described in Algorithm~\ref{alg:ls} is applied independently to each chromosome of the final population. For each chromosome, the neighborhood structures are explored iteratively following Algorithm~\ref{alg:ls}, until either the maximum number of iterations $lsiter$ is reached or the objective value reaches the best lower bound.

We apply local search to the GA variant selected in the previous section. To determine an appropriate value of $lsiter$, we conducted preliminary experiments with $lsiter \in \{200, 300, 400, 500, 700, 1000\}$. These experiments, whose detailed results are omitted for conciseness, showed that increasing $lsiter$ generally improves solution quality at the expense of additional CPU time, although the improvements are not always significant. Based on these observations, we set $lsiter=500$ for instances with $n=20$ and $n=50$, and $lsiter=700$ for instances with $n=100$ and $n=150$. The results obtained with these settings are summarized in Tables~\ref{results:GAKLDLS}--\ref{results:GAKHDLS}.

Tables~\ref{results:GAKLDLS}--\ref{results:GAKHDLS} present the results of GA-ls, following the same distinction as in Tables~\ref{results:GAKLD}--\ref{results:GAKHD} between instances with proven optimal TCT, for which the best upper bound obtained from the MILPs coincides with the best lower bound, and instances without proven optimal TCT, for which the best upper bound obtained from the MILPs does not coincide with the best lower bound. For the former, the results are reported in columns $a$ and $b$, while for the latter, they are reported in columns $c$, $d$, and $e$. The deviation measures are computed using the same expressions as in the previous section.

We observe performance improvements when comparing the GA with local search (GA-ls) to the GA without local search across most instances. Specifically, for instances with a proven optimal TCT, GA-ls solved to optimality 2{,}358 out of 3{,}492 cases for \( p = 0.2 \), 758 out of 2{,}745 for \( p = 0.5 \), and 256 out of 809 for \( p = 0.8 \). The corresponding average deviations from the optimal total completion time are 0.004, 0.016, and 0.014, respectively. For instances without a proven optimal TCT, the GA-ls solution matched the best known lower bound in 837 out of 8{,}508 cases for \( p = 0.2 \). The average deviation from the best lower bound in these cases is 0.038 for \( p = 0.2 \), 0.419 for \( p = 0.5 \), and 0.952 for \( p = 0.8 \). Correspondingly, the average deviation from the best upper bound found so far over the instances without a proven optimal TCT is 0.027, 0.075, and 0.087, respectively. In total, GA-ls solved to optimality 3{,}195 instances for \( p = 0.2 \), 758 for \( p = 0.5 \), and 256 for \( p = 0.8 \) out of 12{,}000 instances. The overall average deviation from the best lower bounds across all instances is 0.028 for \( p = 0.2 \), 0.327 for \( p = 0.5 \), and 0.889 for \( p = 0.8 \).

\section{Conclusion}
This paper addressed the problem of scheduling jobs on identical parallel machines under conflict constraints, where certain jobs cannot be executed simultaneously on different machines. These constraints are given by a simple undirected conflict graph. The objective is to minimize the total completion time. We showed that the problem is strongly NP-hard even with only two machines and when job processing times are restricted to just two distinct values. However, for unit-time jobs, we established a polynomial equivalence with a special case of batch scheduling and derived an algorithm based on maximum matching, showing that the problem is polynomially solvable on two machines for any conflict graph. The problem becomes strongly NP-hard when the number of machines increases to three. Still considering unit-time jobs, we further showed that the problem remains polynomially solvable on three machines if the conflict graph is the complement of a bipartite graph. Also, when the conflict graph is the complement of a star, the problem is tractable for arbitrary processing times.

To solve the general problem on $m$ machines, we proposed three mathematical formulations and a GA. We evaluated their performance on a wide range of benchmark instances derived from those of \citet{kowalczyk2018branch}, augmented with conflict graphs of varying densities. In addition to the fixed-density instances, we considered the CD02, CD05, and LOG instances, with edge probabilities $p=2/n$, $p=0.5/n$, and $p=\log n/n$, respectively. We compared the performance of the mathematical models and the GA against lower bounds on the total completion time. The computational results show that the precedence-based model, when enhanced with a warm start, provided better upper bounds overall, whereas the time-indexed formulation proved optimality more frequently. The precedence-based model solved 12.542\% of the fixed-density instances to proven optimality within a 15-minute time limit, with an average optimality gap reported by Gurobi of 0.544. However, comparing the best upper bounds returned by the formulations to the best lower bounds revealed that at least 19.572\% of these instances are optimally solved, and the average gap between these bounds was reduced to 0.222. For the GA, we conducted extensive experiments to tune its components and parameters. The best-performing variant, combined with local search, outperformed the mathematical models with warm starts while requiring less CPU time. The average deviation of the GA's solutions from the best upper bounds obtained by the mathematical formulations across all fixed-density instances was 0.051, indicating that the GA produced better solutions.

Possible extensions of this work include further investigating the complexity of the problem for specific classes of conflict graphs—for example, for two machines with conflict graphs that are complements of bipartite graphs—motivated by the fact that the problem with unit-time jobs is polynomially solvable on three machines for complements of bipartite graphs. From a methodological perspective, deriving tighter bounds on $T$, possibly by adapting the approach of~\cite{kramer2019enhanced} to account for conflict constraints, could improve the efficiency of the time-indexed formulation. Designing more efficient lower bounds to better assess algorithmic performance, particularly for large-scale instances with high conflict densities. Additionally, developing exact algorithms based on decomposition techniques (e.g., branch-and-price~\cite{kowalczyk2018branch}
and branch-and-cut~\cite{tellache2025variant}) or arc-flow formulations—known for their effectiveness in related scheduling problems on identical machines with total completion time minimization—could further enhance solution quality.

\bibliographystyle{unsrtnat}
\bibliography{biblio}

@article{ER59,
	author={Erd{\H{o}}s, P{\'a}l and R{\'e}nyi, Alfr{\'e}d},
	title={On random graphs {I}},
	Journal={Publicationes Mathematicae},
	pages = {290--297},
	volume    = {6},
	year={1959},
}

@article{SMK03,
	author={Sakai, Shuichi and Togasaki, Mitsunori and Yamazaki, Koichi},
	title={A note on greedy algorithms for the maximum weighted independent set problem},
	journal={Discrete Applied Mathematics},
	year={2003},
	pages = {313--322},
	volume=126
}

@article{KRAMER2021,
	title = {Mathematical formulations for scheduling jobs on identical parallel machines with family setup times and total weighted completion time minimization},
	journal = {European Journal of Operational Research},
	volume = {289},
	number = {3},
	pages = {825--840},
	year = {2021},
	author = {Kramer, Arthur and Iori, Manuel and Lacomme, Philippe},
	
}

@article{kowalczyk2018branch,
	title={A branch-and-price algorithm for parallel machine scheduling using ZDDs and generic branching},
	author={Kowalczyk, Daniel and Leus, Roel},
	journal={INFORMS Journal on Computing},
	volume={30},
	number={4},
	pages={768--782},
	year={2018},
	publisher={INFORMS}
}

@article{GT1960,
	author = {Giffler, Bernard and Thompson, Gerald L.},
	journal = {Operations Research},
	number = {4},
	pages = {487--503},
	publisher = {INFORMS},
	title = {Algorithms for Solving Production-Scheduling Problems},
	volume = {8},
	year = {1960}
}

@article{velez2017changeover,
	title={Changeover formulations for discrete-time mixed-integer programming scheduling models},
	author={Velez, Sara and Dong, Yachao and Maravelias, Christos T.},
	journal={European Journal of Operational Research},
	volume={260},
	number={3},
	pages={949--963},
	year={2017},
	publisher={Elsevier}
}

@article{wolsey1997mip,
	title={MIP modelling of changeovers in production planning and scheduling problems},
	author={Wolsey, Laurence A.},
	journal={European Journal of Operational Research},
	volume={99},
	number={1},
	pages={154--165},
	year={1997},
	publisher={Elsevier}
}

@book{Holland1975,
	author    = {Holland, John Henry},
	title     = {Adaptation in Natural and Artificial Systems},
	year      = {1975},
	publisher = {University of Michigan Press},
	address   = {Ann Arbor, MI},
}

@article{prins1994,
	title={An overview of scheduling problems arising in satellite communications},
	author={Prins, Christian},
	journal={Journal of the Operational Research Society},
	volume={45},
	number={6},
	pages={611--623},
	year={1994},
}

@article{tellache2023genetic,
	title={A genetic algorithm for scheduling open shops with conflict graphs to minimize the makespan},
	author={Tellache, Nour ElHouda and Kerbache, Laoucine},
	journal={Computers \& Operations Research},
	volume={156},
	pages={106247},
	year={2023},
	publisher={Elsevier}
}

@book{MP08,
	author = {Pinedo, Michael L.},
	title = {Scheduling: Theory, Algorithms, and Systems},
	year = {2008},
	publisher = {Springer Publishing Company},
	address =     {NewYork},
	edition = {3rd},
}

@article{REEVES1995,
	author = {Reeves, Colin R.},
	title = {A genetic algorithm for flowshop sequencing},
	journal = {Computers \& Operations Research},
	volume = {22},
	number = {1},
	pages = {5--13},
	year = {1995},
}

@book{Goldberg1989,
	author = {Goldberg, David Edward},
	title = {Genetic Algorithms in Search, Optimization and Machine Learning},
	year = {1989},
	isbn = {0201157675},
	publisher = {Addison-Wesley Longman Publishing Co., Inc.},
	address = {USA},
	edition = {1st},
}

@inproceedings{Falkenauer1991,	
	author={Falkenauer, Emanuel and Bouffouix, S.},
	booktitle={Proceedings. 1991 IEEE International Conference on Robotics and Automation}, 
	title={A genetic algorithm for job shop}, 
	year={1991},
	volume={01},
	pages={824--829}
}

@article{GLLR79,
	author = {Graham, Ronald Lewis and Lawler, Eugene Leighton and Lenstra, Jan Karel and Rinnooy Kan, Alexander Hendrik George},
	journal = {Annals of discrete mathematics},
	pages = {287--326},
	title = {Optimization and approximation in deterministic sequencing and scheduling: a survey},
	volume = 5,
	year = 1979
}

@article{mohabeddine2019new,
	title={New results in two identical machines scheduling with agreement graphs},
	author={Mohabeddine, Amine and Boudhar, Mourad},
	journal={Theoretical Computer Science},
	volume={779},
	pages={37--46},
	year={2019},
	publisher={Elsevier}
}

@article{bendraouche2015scheduling,
	title={Scheduling: Agreement graph vs resource constraints},
	author={Bendraouche, Mohamed and Boudhar, Mourad and Oulamara, Ammar},
	journal={European Journal of Operational Research},
	volume={240},
	number={2},
	pages={355--360},
	year={2015},
	publisher={Elsevier}
}

@article{gardi2009mutual,
	title={Mutual exclusion scheduling with interval graphs or related classes, part I},
	author={Gardi, Fr{\'e}d{\'e}ric},
	journal={Discrete Applied Mathematics},
	volume={157},
	number={1},
	pages={19--35},
	year={2009},
	publisher={Elsevier}
}

@article{baker1996mutual,
	title={Mutual exclusion scheduling},
	author={Baker, Brenda S. and Coffman, Jr., Edward Grady},
	journal={Theoretical Computer Science},
	volume={162},
	number={2},
	pages={225--243},
	year={1996},
	publisher={Elsevier}
}

@article{bendraouche2012scheduling,
	title={Scheduling jobs on identical machines with agreement graph},
	author={Bendraouche, Mohamed and Boudhar, Mourad},
	journal={Computers \& operations research},
	volume={39},
	number={2},
	pages={382--390},
	year={2012},
	publisher={Elsevier}
}

@article{bendraouche2016scheduling,
	title={Scheduling with agreements: new results},
	author={Bendraouche, Mohamed and Boudhar, Mourad},
	journal={International Journal of Production Research},
	volume={54},
	number={12},
	pages={3508--3522},
	year={2016},
	publisher={Taylor \& Francis}
}

@incollection{karp2009reducibility,
	title={Reducibility among combinatorial problems},
	author={Karp, Richard Manning},
	booktitle={50 Years of Integer Programming 1958-2008: from the Early Years to the State-of-the-Art},
	pages={219--241},
	year={2009},
	address   = {Berlin, Heidelberg},
	publisher={Springer}
}

@book{GJ79,
	author =       {Garey, Michael R. and Johnson, David Stifler},
	title =        {Computers and intractability: A Guide to the Theory of NP-Completeness},
	publisher =    {Freeman},
	address =     {New York},
	year =         {1979}
}

@article{TellacheBoudharDAM2017,
	author = {Tellache, Nour ElHouda and Boudhar, Mourad},
	title = "Open shop scheduling problems with conflict graphs ",
	journal = "Discrete Applied Mathematics ",
	volume = "227",
	number = "",
	pages = "103 - 120",
	year = "2017", 
}

@article{TELLACHE2019154,
	author = {Tellache, Nour ElHouda and Boudhar, Mourad and Yalaoui, Farouk},
	title = {Two-machine open shop problem with agreement graph},
	journal = {Theoretical Computer Science},
	volume = {796},
	pages = {154--168},
	year = {2019}
}

@article{TELLACHE202185,
	title = {New complexity results for shop scheduling problems with agreement graphs},
	author = {Tellache, Nour ElHouda},
	journal = {Theoretical Computer Science},
	volume = {889},
	pages = {85-95},
	year = {2021},
}

@article{tellache2018ANOR,
	title={Flow shop scheduling problem with conflict graphs},
	author={Tellache, Nour ElHouda and Boudhar, Mourad},
	journal={Annals of Operations Research},
	volume={261},
	pages={339--363},
	year={2018},
	publisher={Springer}
}

@article{tellache2017two,
	title={Two-machine flow shop problem with unit-time operations and conflict graph},
	author={Tellache, Nour ElHouda and Boudhar, Mourad},
	journal={International Journal of Production Research},
	volume={55},
	number={6},
	pages={1664--1679},
	year={2017},
	publisher={Taylor \& Francis}
}

@book{conway1967theory,
	author    = {Conway, Richard Walter and Maxwell, William L. and Miller, Louis W.},
	title     = {Theory of Scheduling},
	publisher = {Addison-Wesley Publishing Company},
	address   = {Reading, Massachusetts},
	year      = {1967}
}

@inproceedings{cai2018approximation,
	title={Approximation algorithms for two-machine flow-shop scheduling with a conflict graph},
	author={Cai, Yinhui and Chen, Guangting and Chen, Yong and Goebel, Randy and Lin, Guohui and Liu, Longcheng and Zhang, An},
	booktitle={Computing and Combinatorics: 24th International Conference, COCOON 2018, Qing Dao, China, July 2-4, 2018, Proceedings 24},
	pages={205--217},
	year={2018},
	organization={Springer}
}

@phdthesis{Ilhan1976,
	author       = {Or, Ilhan},
	title        = {Traveling salesman-type combinatorial problems and their relation to the logistics of regional blood banking},
	school       = {Northwestern University},
	year         = {1976},
	type         = {Ph.D. dissertation}
}

@article{croes1958,
	title={A method for solving traveling-salesman problems},
	author={Croes, G. A.},
	journal={Operations research},
	volume={6},
	number={6},
	pages={791--812},
	year={1958},
	publisher={INFORMS}
}

@article{bruno1974scheduling,
	title={Scheduling independent tasks to reduce mean finishing time},
	author={Bruno, James and Coffman, Jr., Edward Grady and Sethi, Ravi},
	journal={Communications of the ACM},
	volume={17},
	number={7},
	pages={382--387},
	year={1974},
	publisher={ACM New York, NY, USA}
}

@article{kramer2019enhanced,
	title={Enhanced arc-flow formulations to minimize weighted completion time on identical parallel machines},
	author={Kramer, Arthur and Dell'Amico, Mauro and Iori, Manuel},
	journal={European Journal of Operational Research},
	volume={275},
	number={1},
	pages={67--79},
	year={2019},
	publisher={Elsevier}
}

@article{baptiste2007complexity,
	title={The complexity of mean flow time scheduling problems with release times},
	author={Baptiste, Philippe and Brucker, Peter and Chrobak, Marek and D{\"u}rr, Christoph and Kravchenko, Svetlana A. and Sourd, Francis},
	journal={Journal of Scheduling},
	volume={10},
	number={2},
	pages={139--146},
	year={2007},
	publisher={Springer}
}

@article{wang2019complexity,
	title={A complexity analysis of parallel scheduling unit-time jobs with in-tree precedence constraints while minimizing the mean flow time},
	author={Wang, Tianyu and Bellenguez-Morineau, Odile},
	journal={Journal of Scheduling},
	volume={22},
	pages={709--714},
	year={2019},
	publisher={Springer}
}

@article{smith1956,
	author = {Smith, Wayne E.},
	title = {Various optimizers for single-stage production},
	journal = {Naval Research Logistics Quarterly},
	volume = {3},
	number = {1-2},
	pages = {59-66},
	year = {1956}
}

@article{skutella2000ptas,
	title={A PTAS for minimizing the total weighted completion time on identical parallel machines},
	author={Skutella, Martin and Woeginger, Gerhard J.},
	journal={Mathematics of Operations Research},
	volume={25},
	number={1},
	pages={63--75},
	year={2000},
	publisher={INFORMS}
}

@article{boudhar2000scheduling,
	title={Scheduling on a batch machine with job compatibilities},
	author={Boudhar, Mourad and Finke, Gerd},
	journal={JORBEL-Belgian Journal of Operations Research, Statistics, and Computer Science},
	volume={40},
	number={1-2},
	pages={69--80},
	year={2000}
}

@article{boudhar2009scheduling,
	title={Scheduling with undirected graphs: motivations and reviews},
	author={Boudhar, Mourad},
	journal={International Journal of Operational Research},
	volume={6},
	number={3},
	pages={405--419},
	year={2009},
	publisher={Inderscience Publishers}
}

@article{fowler2022survey,
	title={A survey of scheduling with parallel batch (p-batch) processing},
	author={Fowler, John W. and M{\"o}nch, Lars},
	journal={European journal of operational research},
	volume={298},
	number={1},
	pages={1--24},
	year={2022},
	publisher={Elsevier}
}

@article{gilbert1959random,
	author={Gilbert, Edgar Nelson},
	title={Random graphs},
	journal={The Annals of Mathematical Statistics},
	volume={30},
	number={4},
	pages={1141--1144},
	year={1959}
}

@article{erdos1960evolution,
	author={Erd{\H{o}}s, P{\'a}l and R{\'e}nyi, Alfr{\'e}d},
	title={On the evolution of random graphs},
	journal={Publications of the Mathematical Institute of the Hungarian Academy of Sciences},
	volume={5},
	pages={17--61},
	year={1960}
}

@misc{TellacheAzerine2026,
	title        = {Minimizing the makespan in job shop scheduling under conflict graph constraints},
	author       = {Tellache, Nour ElHouda and Azerine, Abdenour},
	year         = {2026},
	eprint       = {2609.04161},
	doi          = {10.48550/arXiv.2609.04161}
}

@article{kawaguchi1986worst,
	author  = {Kawaguchi, Tsuyoshi and Kyan, Seiki},
	title   = {Worst Case Bound of an LRF Schedule for the Mean Weighted Flow-Time Problem},
	journal = {SIAM Journal on Computing},
	volume  = {15},
	number  = {4},
	pages   = {1119--1129},
	year    = {1986},
	doi     = {10.1137/0215081}
}

@article{tellache2025variant,
	title={On a variant of the minimum path cover problem in acyclic digraphs: Computational complexity results and exact method},
	author={Tellache, Nour ElHouda and Baldacci, Roberto},
	journal={Networks},
	volume={86},
	number={3},
	pages={325--357},
	year={2025},
	publisher={Wiley Online Library}
}

@article{prins2000competitive,
	author  = {Prins, Christian},
	title   = {Competitive genetic algorithms for the open-shop
	scheduling problem},
	journal = {Mathematical Methods of Operations Research},
	year    = {2000},
	volume  = {52},
	number  = {3},
	pages   = {389--411}
}

@article{prins2004simple,
	author  = {Prins, Christian},
	title   = {A simple and effective evolutionary algorithm
	for the vehicle routing problem},
	journal = {Computers \& Operations Research},
	year    = {2004},
	volume  = {31},
	number  = {12},
	pages   = {1985--2002}
}

@article{liu2009effective,
	author  = {Liu, Shuguang and Huang, Weilai and Ma, Huiming},
	title   = {An effective genetic algorithm for the fleet size
	and mix vehicle routing problems},
	journal = {Transportation Research Part E:
	Logistics and Transportation Review},
	year    = {2009},
	volume  = {45},
	number  = {3},
	pages   = {434--445}
}

\newpage

\appendix

\section{Computational results for constant and logarithmic expected degrees} \label{sec:CDLOG}

This appendix presents the computational results for the CD and LD instances. Their expected edge densities are given by $p=c/n$,
with $c\in\{0.5,2\}$, and $p=\ln(n)/n$, respectively.
For the instance sizes $n\in\{20,50,100,150\}$ considered in
our experiments, these densities range approximately from
0.0033 to 0.025 for $c=0.5$, from 0.0133 to 0.100 for $c=2$,
and from 0.0334 to 0.150 for the LD instances. Thus, these
graphs are sparser than the FD graphs with $p=0.2$, with
the difference increasing as $n$ grows. We focus on the
similarities and differences relative to the results in
Section~\ref{sec:experiments}, particularly for the FD
instances with $p=0.2$. Overall, the results
suggest that the lower conflict densities postpone the computational limitations observed for the FD instances
to larger instance sizes.

\begin{table}[h]
	\centering
	\fontsize{8}{9.5}\selectfont
	\setlength{\tabcolsep}{0pt}
	\renewcommand{\arraystretch}{0.94}
	\caption{Comparison of the lower bounds on the CD and LD instances with $n=20$. Each $(p,m)$ contains 120 instances.}
	\label{results:LB20CDLOG}
	\noindent\begin{tabular*}{\linewidth}{@{\extracolsep{\fill}}lrrrrrrrrrrrrr@{}}
		\toprule
		& & \multicolumn{3}{c}{$LB_{1}$} & \multicolumn{3}{c}{$LB_{2}$} & \multicolumn{3}{c}{$LB_{3}$} & \multicolumn{3}{c}{$LB_{4}$} \\
		\cmidrule(lr){3-5} \cmidrule(lr){6-8} \cmidrule(lr){9-11} \cmidrule(lr){12-14}
		$p$ & $m$ & $b$ & $c$ & $d$ & $b$ & $c$ & $d$ & $b$ & $c$ & $d$ & $b$ & $c$ & $d$ \\ \midrule
		$0.5/n$ & 3 & 88.333 & 0.013 & 0.000 & 0.000 & 95.599 & 0.000 & 0.000 & 95.516 & 0.000 & 0.000 & 95.076 & 0.000 \\
		& 5 & 44.167 & 0.208 & 0.000 & 0.000 & 92.853 & 0.000 & 0.000 & 92.728 & 0.000 & 0.000 & 92.281 & 0.000 \\
		& 8 & 33.333 & 0.402 & 0.000 & 0.000 & 91.216 & 0.000 & 0.000 & 91.352 & 0.000 & 0.000 & 90.325 & 0.000 \\
		& 10 & 16.667 & 1.545 & 0.000 & 0.000 & 88.899 & 0.000 & 0.000 & 88.809 & 0.000 & 0.000 & 88.253 & 0.000 \\
		& 12 & 7.500 & 2.070 & 0.000 & 0.000 & 88.302 & 0.000 & 0.000 & 88.033 & 0.000 & 0.000 & 87.027 & 0.000 \\
		\midrule
		$2/n$ & 3 & 38.333 & 0.155 & 0.000 & 0.000 & 92.366 & 0.000 & 0.000 & 92.143 & 0.000 & 0.000 & 92.201 & 0.000 \\
		& 5 & 1.667 & 0.884 & 0.000 & 0.000 & 89.361 & 0.000 & 0.000 & 89.183 & 0.000 & 0.000 & 88.945 & 0.000 \\
		& 8 & 0.833 & 3.821 & 0.000 & 0.000 & 85.714 & 0.000 & 0.000 & 85.333 & 0.000 & 0.000 & 85.454 & 0.000 \\
		& 10 & 0.000 & 9.176 & 0.000 & 0.000 & 84.644 & 0.000 & 0.000 & 84.409 & 0.000 & 0.000 & 84.280 & 0.000 \\
		& 12 & 0.000 & 12.509 & 0.000 & 0.000 & 84.383 & 0.000 & 0.000 & 84.146 & 0.000 & 0.000 & 83.287 & 0.000 \\
		\midrule
		$\ln(n)/n$ & 3 & 23.333 & 0.238 & 0.000 & 0.000 & 91.153 & 0.000 & 0.000 & 91.028 & 0.000 & 0.000 & 91.069 & 0.000 \\
		& 5 & 0.000 & 1.773 & 0.000 & 0.000 & 86.883 & 0.000 & 0.000 & 86.826 & 0.000 & 0.000 & 87.245 & 0.000 \\
		& 8 & 0.000 & 8.071 & 0.000 & 0.000 & 84.632 & 0.000 & 0.000 & 84.425 & 0.000 & 0.000 & 84.345 & 0.000 \\
		& 10 & 0.000 & 16.256 & 0.000 & 0.000 & 83.000 & 0.000 & 0.000 & 82.425 & 0.000 & 0.000 & 82.315 & 0.000 \\
		& 12 & 0.000 & 21.976 & 0.000 & 0.000 & 82.280 & 0.000 & 0.000 & 82.247 & 0.000 & 0.000 & 82.976 & 0.000 \\
	\end{tabular*}
	\par\medskip
	\noindent\begin{tabular*}{\linewidth}{@{\extracolsep{\fill}}lrrrrrrrrrrrrr@{}}
		\toprule
		& & \multicolumn{4}{c}{$LB_{5}$} & \multicolumn{4}{c}{$LB_{6}$} & \multicolumn{4}{c}{$LB_{7}$} \\
		\cmidrule(lr){3-6} \cmidrule(lr){7-10} \cmidrule(lr){11-14}
		$p$ & $m$ & $a$ & $b$ & $c$ & $d$ & $a$ & $b$ & $c$ & $d$ & $a$ & $b$ & $c$ & $d$ \\ \midrule
		$0.5/n$ & 3 & 0 & \textbf{100.000} & \textbf{0.000} & 1.312 & 0 & 0.000 & 54.255 & 167.532 & 0 & 0.000 & 68.976 & 0.019 \\
		& 5 & 0 & 98.333 & 0.002 & 1.229 & 0 & 0.000 & 35.423 & 184.334 & 0 & 0.000 & 53.509 & 0.026 \\
		& 8 & 0 & 98.333 & 0.002 & 1.116 & 0 & 0.000 & 17.945 & 202.981 & 0 & 0.000 & 36.825 & 0.019 \\
		& 10 & 0 & 85.000 & 0.090 & 1.151 & 0 & 0.000 & 9.498 & 204.506 & 0 & 0.000 & 27.796 & 0.020 \\
		& 12 & 0 & 94.167 & 0.018 & 0.917 & 0 & 0.000 & 7.589 & 173.605 & 0 & 0.000 & 22.963 & 0.017 \\
		\midrule
		$2/n$ & 3 & 0 & 95.000 & 0.006 & 3.680 & 0 & 0.000 & 54.317 & 122.599 & 0 & 0.000 & 69.025 & 0.017 \\
		& 5 & 0 & 69.167 & 0.093 & 2.846 & 0 & 0.000 & 35.848 & 141.381 & 0 & 0.000 & 53.850 & 0.019 \\
		& 8 & 0 & 29.167 & 1.050 & 2.618 & 0 & 0.000 & 20.643 & 142.330 & 0 & 0.000 & 39.159 & 0.019 \\
		& 10 & 0 & 9.167 & 3.448 & 2.250 & 0 & 0.000 & 16.363 & 144.266 & 0 & 0.000 & 33.542 & 0.020 \\
		& 12 & 0 & 8.333 & 3.569 & 2.112 & 0 & 0.000 & 17.318 & 128.350 & 0 & 0.000 & 31.354 & 0.017 \\
		\midrule
		$\ln(n)/n$ & 3 & 0 & 85.833 & 0.026 & 6.241 & 0 & 0.000 & 54.354 & 123.651 & 0 & 0.000 & 69.055 & 0.017 \\
		& 5 & 0 & 28.333 & 0.498 & 5.766 & 0 & 0.000 & 36.390 & 140.443 & 0 & 0.000 & 54.299 & 0.019 \\
		& 8 & 0 & 2.500 & 4.156 & 5.062 & 0 & 0.000 & 23.989 & 157.172 & 0 & 0.000 & 42.042 & 0.019 \\
		& 10 & 0 & 0.000 & 9.172 & 3.607 & 0 & 0.000 & 22.747 & 133.222 & 0 & 0.000 & 38.852 & 0.026 \\
		& 12 & 0 & 0.833 & 10.231 & 3.694 & 0 & 0.000 & 26.229 & 124.060 & 0 & 0.000 & 38.842 & 0.018 \\
	\end{tabular*}
	\par\medskip
	\noindent\begin{tabular*}{\linewidth}{@{\extracolsep{\fill}}lrrrrrrrrrrrrr@{}}
		\toprule
		& & \multicolumn{4}{c}{$LB'_{5}$} & \multicolumn{4}{c}{$LB'_{6}$} & \multicolumn{4}{c}{$LB'_{7}$} \\
		\cmidrule(lr){3-6} \cmidrule(lr){7-10} \cmidrule(lr){11-14}
		$p$ & $m$ & $a$ & $b$ & $c$ & $d$ & $a$ & $b$ & $c$ & $d$ & $a$ & $b$ & $c$ & $d$ \\ \midrule
		$0.5/n$ & 3 & 0 & \textbf{100.000} & \textbf{0.000} & 6.870 & 40 & 0.000 & 47.761 & 910.196 & 0 & 0.000 & 43.215 & 900.542 \\
		& 5 & 0 & \textbf{100.000} & \textbf{0.000} & 6.948 & 40 & 3.333 & 27.897 & 904.428 & 0 & 0.000 & 22.092 & 901.151 \\
		& 8 & 0 & \textbf{100.000} & \textbf{0.000} & 6.638 & 40 & 16.667 & 10.695 & 772.073 & 0 & 9.167 & 5.706 & 855.784 \\
		& 10 & 0 & \textbf{100.000} & \textbf{0.000} & 7.331 & 40 & 39.167 & 4.354 & 604.512 & 0 & \textbf{100.000} & \textbf{0.000} & 3.271 \\
		& 12 & 0 & \textbf{100.000} & \textbf{0.000} & 6.539 & 40 & 52.500 & 1.034 & 556.258 & 0 & \textbf{100.000} & \textbf{0.000} & 0.745 \\
		\midrule
		$2/n$ & 3 & 0 & \textbf{100.000} & \textbf{0.000} & 20.111 & 40 & 0.000 & 48.885 & 914.095 & 0 & 0.000 & 39.365 & 900.293 \\
		& 5 & 0 & \textbf{100.000} & \textbf{0.000} & 18.827 & 40 & 4.167 & 26.330 & 900.518 & 0 & 0.833 & 16.493 & 895.663 \\
		& 8 & 0 & \textbf{100.000} & \textbf{0.000} & 21.864 & 40 & 32.500 & 8.799 & 697.012 & 0 & 60.000 & 1.855 & 387.529 \\
		& 10 & 0 & \textbf{100.000} & \textbf{0.000} & 20.266 & 40 & 47.500 & 4.791 & 595.435 & 0 & \textbf{100.000} & \textbf{0.000} & 0.899 \\
		& 12 & 0 & \textbf{100.000} & \textbf{0.000} & 16.498 & 40 & 52.500 & 3.262 & 581.116 & 0 & \textbf{100.000} & \textbf{0.000} & 0.198 \\
		\midrule
		$\ln(n)/n$ & 3 & 0 & \textbf{100.000} & \textbf{0.000} & 28.192 & 40 & 0.000 & 48.769 & 914.240 & 0 & 0.000 & 36.454 & 900.184 \\
		& 5 & 0 & \textbf{100.000} & \textbf{0.000} & 34.891 & 40 & 4.167 & 26.131 & 901.796 & 0 & 5.833 & 11.729 & 865.104 \\
		& 8 & 0 & \textbf{100.000} & \textbf{0.000} & 38.347 & 40 & 32.500 & 14.849 & 673.585 & 0 & 81.667 & 0.573 & 198.528 \\
		& 10 & 0 & \textbf{100.000} & \textbf{0.000} & 39.119 & 40 & 38.333 & 12.256 & 631.581 & 0 & \textbf{100.000} & \textbf{0.000} & 0.492 \\
		& 12 & 0 & \textbf{100.000} & \textbf{0.000} & 29.344 & 40 & 39.167 & 11.308 & 629.286 & 0 & \textbf{100.000} & \textbf{0.000} & 0.326 \\
		\bottomrule
	\end{tabular*}
	\par\smallskip
	\begin{minipage}{\linewidth}\fontsize{8}{9.5}\selectfont
		$a$: Number of instances where Gurobi could not find a lower bound within the time limit (for LPs, no optimal LP result). \par 
	$b$ (\%): Number of times in percentage a given lower bound yields the best bound. \par
	$c$ (\%): Average percentage deviation from the best lower bound over instances with a valid $LB$. \par
	$d$ (s): Average CPU time in seconds.
	\end{minipage}
\end{table}

\begin{table}[p]
	\centering
	\fontsize{8}{9.5}\selectfont
	\setlength{\tabcolsep}{0pt}
	\renewcommand{\arraystretch}{0.94}
	\caption{Comparison of the lower bounds on the CD and LD instances with $n=50$. Each $(p,m)$ contains 120 instances.}
	\label{results:LB50CDLOG}
	\noindent\begin{tabular*}{\linewidth}{@{\extracolsep{\fill}}lrrrrrrrrrrrrr@{}}
		\toprule
		& & \multicolumn{3}{c}{$LB_{1}$} & \multicolumn{3}{c}{$LB_{2}$} & \multicolumn{3}{c}{$LB_{3}$} & \multicolumn{3}{c}{$LB_{4}$} \\
		\cmidrule(lr){3-5} \cmidrule(lr){6-8} \cmidrule(lr){9-11} \cmidrule(lr){12-14}
		$p$ & $m$ & $b$ & $c$ & $d$ & $b$ & $c$ & $d$ & $b$ & $c$ & $d$ & $b$ & $c$ & $d$ \\ \midrule
		$0.5/n$ & 3 & 89.167 & 0.001 & 0.000 & 0.000 & 99.154 & 0.000 & 0.000 & 99.122 & 0.000 & 0.000 & 99.018 & 0.000 \\
		& 5 & 68.333 & 0.012 & 0.000 & 0.000 & 98.540 & 0.000 & 0.000 & 98.510 & 0.000 & 0.000 & 98.378 & 0.000 \\
		& 8 & 40.000 & 0.043 & 0.000 & 0.000 & 97.917 & 0.000 & 0.000 & 97.839 & 0.000 & 0.000 & 97.550 & 0.000 \\
		& 10 & 24.167 & 0.106 & 0.000 & 0.000 & 97.359 & 0.000 & 0.000 & 97.252 & 0.000 & 0.000 & 97.053 & 0.000 \\
		& 12 & 22.500 & 0.139 & 0.000 & 0.000 & 96.995 & 0.000 & 0.000 & 96.920 & 0.000 & 0.000 & 96.631 & 0.000 \\
		\midrule
		$2/n$ & 3 & 65.000 & 0.007 & 0.000 & 0.000 & 98.751 & 0.000 & 0.000 & 98.707 & 0.000 & 0.000 & 98.675 & 0.000 \\
		& 5 & 32.500 & 0.040 & 0.000 & 0.000 & 98.049 & 0.000 & 0.000 & 98.000 & 0.000 & 0.000 & 97.875 & 0.000 \\
		& 8 & 5.000 & 0.183 & 0.000 & 0.000 & 97.110 & 0.000 & 0.000 & 97.006 & 0.000 & 0.000 & 96.968 & 0.000 \\
		& 10 & 0.000 & 0.422 & 0.000 & 0.000 & 96.607 & 0.000 & 0.000 & 96.407 & 0.000 & 0.000 & 96.283 & 0.000 \\
		& 12 & 0.000 & 0.736 & 0.000 & 0.000 & 95.852 & 0.000 & 0.000 & 95.828 & 0.000 & 0.000 & 95.651 & 0.000 \\
		\midrule
		$\ln(n)/n$ & 3 & 55.833 & 0.015 & 0.000 & 0.000 & 98.465 & 0.000 & 0.000 & 98.421 & 0.000 & 0.000 & 98.390 & 0.000 \\
		& 5 & 13.333 & 0.086 & 0.000 & 0.000 & 97.394 & 0.000 & 0.000 & 97.388 & 0.000 & 0.000 & 97.383 & 0.000 \\
		& 8 & 0.000 & 0.405 & 0.000 & 0.000 & 96.224 & 0.000 & 0.000 & 96.158 & 0.000 & 0.000 & 96.054 & 0.000 \\
		& 10 & 0.000 & 1.104 & 0.000 & 0.000 & 95.339 & 0.000 & 0.000 & 95.305 & 0.000 & 0.000 & 95.264 & 0.000 \\
		& 12 & 0.000 & 1.806 & 0.000 & 0.000 & 94.660 & 0.000 & 0.000 & 94.485 & 0.000 & 0.000 & 94.648 & 0.000 \\
	\end{tabular*}
	\par\medskip
	\noindent\begin{tabular*}{\linewidth}{@{\extracolsep{\fill}}lrrrrrrrrr@{}}
		\toprule
		& & \multicolumn{4}{c}{$LB_{5}$} & \multicolumn{4}{c}{$LB_{7}$} \\
		\cmidrule(lr){3-6} \cmidrule(lr){7-10}
		$p$ & $m$ & $a$ & $b$ & $c$ & $d$ & $a$ & $b$ & $c$ & $d$ \\ \midrule
		$0.5/n$ & 3 & 0 & \textbf{100.000} & \textbf{0.000} & 25.934 & 0 & 0.000 & 86.144 & 0.034 \\
		& 5 & 0 & \textbf{100.000} & \textbf{0.000} & 20.730 & 0 & 0.000 & 77.941 & 0.047 \\
		& 8 & 0 & 99.167 & 0.000 & 10.860 & 0 & 0.000 & 67.333 & 0.035 \\
		& 10 & 0 & 99.167 & 0.001 & 14.121 & 0 & 0.000 & 60.715 & 0.049 \\
		& 12 & 0 & 99.167 & 0.000 & 12.411 & 0 & 0.000 & 55.019 & 0.052 \\
		\midrule
		$2/n$ & 3 & 0 & \textbf{100.000} & \textbf{0.000} & 97.758 & 0 & 0.000 & 86.145 & 0.034 \\
		& 5 & 0 & 96.667 & 0.000 & 60.096 & 0 & 0.000 & 77.947 & 0.034 \\
		& 8 & 0 & 91.667 & 0.004 & 42.594 & 0 & 0.000 & 67.384 & 0.036 \\
		& 10 & 0 & 79.167 & 0.014 & 41.969 & 0 & 0.000 & 60.851 & 0.046 \\
		& 12 & 0 & 60.833 & 0.051 & 48.274 & 0 & 0.000 & 55.312 & 0.046 \\
		\midrule
		$\ln(n)/n$ & 3 & 0 & \textbf{97.500} & \textbf{0.001} & 301.750 & 0 & 0.000 & 86.146 & 0.039 \\
		& 5 & 0 & 87.500 & 0.004 & 186.480 & 0 & 0.000 & 77.959 & 0.033 \\
		& 8 & 0 & 64.167 & 0.033 & 189.971 & 0 & 0.000 & 67.463 & 0.038 \\
		& 10 & 0 & 15.000 & 0.184 & 169.419 & 0 & 0.000 & 61.145 & 0.042 \\
		& 12 & 0 & 7.500 & 0.489 & 150.471 & 0 & 0.000 & 55.835 & 0.045 \\
	\end{tabular*}
	\par\medskip
	\noindent\begin{tabular*}{\linewidth}{@{\extracolsep{\fill}}lrrrrrrrrr@{}}
		\toprule
		& & \multicolumn{4}{c}{$LB'_{5}$} & \multicolumn{4}{c}{$LB'_{7}$} \\
		\cmidrule(lr){3-6} \cmidrule(lr){7-10}
		$p$ & $m$ & $a$ & $b$ & $c$ & $d$ & $a$ & $b$ & $c$ & $d$ \\ \midrule
		$0.5/n$ & 3 & 0 & 98.333 & 0.000 & 91.157 & 0 & 0.000 & 73.158 & 900.078 \\
		& 5 & 0 & \textbf{100.000} & \textbf{0.000} & 59.265 & 0 & 0.000 & 58.838 & 900.084 \\
		& 8 & 0 & \textbf{100.000} & \textbf{0.000} & 59.542 & 0 & 0.000 & 42.185 & 900.138 \\
		& 10 & 0 & \textbf{100.000} & \textbf{0.000} & 53.318 & 0 & 0.000 & 32.754 & 900.212 \\
		& 12 & 0 & \textbf{100.000} & \textbf{0.000} & 52.223 & 0 & 0.000 & 25.363 & 900.289 \\
		\midrule
		$2/n$ & 3 & 0 & 99.167 & 0.833 & 265.388 & 0 & 0.000 & 73.062 & 900.093 \\
		& 5 & 0 & \textbf{100.000} & \textbf{0.000} & 157.099 & 0 & 0.000 & 58.206 & 900.084 \\
		& 8 & 0 & \textbf{100.000} & \textbf{0.000} & 158.318 & 0 & 0.000 & 41.211 & 900.092 \\
		& 10 & 0 & \textbf{100.000} & \textbf{0.000} & 143.003 & 0 & 0.000 & 31.565 & 900.099 \\
		& 12 & 0 & \textbf{100.000} & \textbf{0.000} & 156.322 & 0 & 0.000 & 23.078 & 900.088 \\
		\midrule
		$\ln(n)/n$ & 3 & 0 & 90.000 & 10.000 & 484.040 & 0 & 0.000 & 74.294 & 900.072 \\
		& 5 & 0 & \textbf{100.000} & \textbf{0.000} & 384.039 & 0 & 0.000 & 59.611 & 900.095 \\
		& 8 & 0 & \textbf{100.000} & \textbf{0.000} & 437.445 & 0 & 0.000 & 41.613 & 900.085 \\
		& 10 & 0 & \textbf{100.000} & \textbf{0.000} & 449.295 & 0 & 0.000 & 30.661 & 900.088 \\
		& 12 & 0 & \textbf{100.000} & \textbf{0.000} & 476.904 & 0 & 0.000 & 21.113 & 900.087 \\
		\bottomrule
	\end{tabular*}
	\par\smallskip
	\begin{minipage}{\linewidth}\fontsize{8}{9.5}\selectfont
			$a$: Number of instances where Gurobi could not find a lower bound within the time limit (for LPs, no optimal LP result). \par 
		$b$ (\%): Number of times in percentage a given lower bound yields the best bound. \par
		$c$ (\%): Average percentage deviation from the best lower bound over instances with a valid $LB$. \par
		$d$ (s): Average CPU time in seconds.
	\end{minipage}
\end{table}

\begin{table}[p]
	\centering
	\fontsize{8}{9.5}\selectfont
	\setlength{\tabcolsep}{0pt}
	\renewcommand{\arraystretch}{0.94}
	\caption{Comparison of eight lower bounds on the CD and LD instances with $n=100$. Each $(p,m)$ contains 120 instances.}
	\label{results:LB100CDLOG}
	\noindent\begin{tabular*}{\linewidth}{@{\extracolsep{\fill}}lrrrrrrrrrrrrr@{}}
		\toprule
		& & \multicolumn{3}{c}{$LB_{1}$} & \multicolumn{3}{c}{$LB_{2}$} & \multicolumn{3}{c}{$LB_{3}$} & \multicolumn{3}{c}{$LB_{4}$} \\
		\cmidrule(lr){3-5} \cmidrule(lr){6-8} \cmidrule(lr){9-11} \cmidrule(lr){12-14}
		$p$ & $m$ & $b$ & $c$ & $d$ & $b$ & $c$ & $d$ & $b$ & $c$ & $d$ & $b$ & $c$ & $d$ \\ \midrule
		$0.5/n$ & 3 & \textbf{91.667} & 0.000 & 0.000 & 0.000 & 99.755 & 0.000 & 0.000 & 99.743 & 0.000 & 0.000 & 99.713 & 0.000 \\
		& 5 & 84.167 & 0.001 & 0.000 & 0.000 & 99.589 & 0.000 & 0.000 & 99.574 & 0.000 & 0.000 & 99.540 & 0.000 \\
		& 8 & 66.667 & 0.005 & 0.000 & 0.000 & 99.384 & 0.000 & 0.000 & 99.331 & 0.000 & 0.000 & 99.305 & 0.000 \\
		& 10 & 45.000 & 0.011 & 0.000 & 0.000 & 99.217 & 0.000 & 0.000 & 99.201 & 0.000 & 0.000 & 99.199 & 0.000 \\
		& 12 & 40.833 & 0.016 & 0.000 & 0.000 & 99.123 & 0.000 & 0.000 & 99.084 & 0.000 & 0.000 & 99.009 & 0.000 \\
		\midrule
		$2/n$ & 3 & \textbf{85.000} & 0.000 & 0.000 & 0.000 & 99.694 & 0.000 & 0.000 & 99.683 & 0.000 & 0.000 & 99.658 & 0.000 \\
		& 5 & 64.167 & 0.004 & 0.000 & 0.000 & 99.487 & 0.000 & 0.000 & 99.478 & 0.000 & 0.000 & 99.446 & 0.000 \\
		& 8 & 50.000 & 0.016 & 0.000 & 0.000 & 99.212 & 0.000 & 0.000 & 99.211 & 0.000 & 0.000 & 99.143 & 0.000 \\
		& 10 & 21.667 & 0.040 & 0.000 & 0.000 & 99.045 & 0.000 & 0.000 & 99.023 & 0.000 & 0.000 & 98.955 & 0.000 \\
		& 12 & 17.500 & 0.073 & 0.000 & 0.000 & 98.888 & 0.000 & 0.000 & 98.864 & 0.000 & 0.000 & 98.782 & 0.000 \\
		\midrule
		$\ln(n)/n$ & 3 & \textbf{100.000} & \textbf{0.000} & 0.000 & 0.000 & 99.567 & 0.000 & 0.000 & 99.576 & 0.000 & 0.000 & 99.557 & 0.000 \\
		& 5 & \textbf{92.500} & 0.001 & 0.000 & 0.000 & 99.323 & 0.000 & 0.000 & 99.284 & 0.000 & 0.000 & 99.320 & 0.000 \\
		& 8 & \textbf{70.833} & 0.013 & 0.000 & 0.000 & 98.922 & 0.000 & 0.000 & 98.898 & 0.000 & 0.000 & 98.935 & 0.000 \\
		& 10 & \textbf{65.000} & 0.037 & 0.000 & 0.000 & 98.731 & 0.000 & 0.000 & 98.666 & 0.000 & 0.000 & 98.677 & 0.000 \\
		& 12 & \textbf{62.500} & 0.084 & 0.000 & 0.000 & 98.540 & 0.000 & 0.000 & 98.534 & 0.000 & 0.000 & 98.580 & 0.000 \\
	\end{tabular*}
	\par\medskip
	\noindent\begin{tabular*}{\linewidth}{@{\extracolsep{\fill}}lrrrrrrrrr@{}}
		\toprule
		& & \multicolumn{4}{c}{$LB_{5}$} & \multicolumn{4}{c}{$LB_{7}$} \\
		\cmidrule(lr){3-6} \cmidrule(lr){7-10}
		$p$ & $m$ & $a$ & $b$ & $c$ & $d$ & $a$ & $b$ & $c$ & $d$ \\ \midrule
		$0.5/n$ & 3 & 18 & 85.000 & \textbf{0.000} & 318.291 & 0 & 0.000 & 92.800 & 0.127 \\
		& 5 & 0 & \textbf{100.000} & \textbf{0.000} & 170.616 & 0 & 0.000 & 88.293 & 0.124 \\
		& 8 & 3 & \textbf{97.500} & \textbf{0.000} & 231.661 & 0 & 0.000 & 81.964 & 0.124 \\
		& 10 & 0 & \textbf{100.000} & \textbf{0.000} & 165.612 & 0 & 0.000 & 77.990 & 0.128 \\
		& 12 & 1 & \textbf{99.167} & \textbf{0.000} & 129.374 & 0 & 0.000 & 74.227 & 0.127 \\
		\midrule
		$2/n$ & 3 & 52 & 56.667 & \textbf{0.000} & 558.164 & 0 & 0.000 & 92.800 & 0.115 \\
		& 5 & 30 & \textbf{75.000} & \textbf{0.000} & 468.664 & 0 & 0.000 & 88.294 & 0.116 \\
		& 8 & 64 & 46.667 & \textbf{0.000} & 581.177 & 0 & 0.000 & 81.966 & 0.118 \\
		& 10 & 43 & 63.333 & \textbf{0.000} & 508.898 & 0 & 0.000 & 77.997 & 0.126 \\
		& 12 & 32 & \textbf{72.500} & \textbf{0.000} & 460.530 & 0 & 0.000 & 74.244 & 0.126 \\
		\midrule
		$\ln(n)/n$ & 3 & 80 & 33.333 & \textbf{0.000} & 623.609 & 0 & 0.000 & 92.800 & 0.129 \\
		& 5 & 80 & 32.500 & \textbf{0.000} & 625.207 & 0 & 0.000 & 88.293 & 0.120 \\
		& 8 & 80 & 30.833 & \textbf{0.001} & 618.416 & 0 & 0.000 & 81.965 & 0.131 \\
		& 10 & 80 & 29.167 & \textbf{0.006} & 615.739 & 0 & 0.000 & 77.996 & 0.149 \\
		& 12 & 77 & 30.833 & \textbf{0.006} & 602.771 & 0 & 0.000 & 74.247 & 0.141 \\
	\end{tabular*}
	\par\medskip
	\noindent\begin{tabular*}{\linewidth}{@{\extracolsep{\fill}}lrrrrrrrrr@{}}
		\toprule
		& & \multicolumn{4}{c}{$LB'_{5}$} & \multicolumn{4}{c}{$LB'_{7}$} \\
		\cmidrule(lr){3-6} \cmidrule(lr){7-10}
		$p$ & $m$ & $a$ & $b$ & $c$ & $d$ & $a$ & $b$ & $c$ & $d$ \\ \midrule
		$0.5/n$ & 3 & 0 & 66.667 & 32.500 & 441.510 & 0 & 0.000 & 87.357 & 900.226 \\
		& 5 & 0 & 79.167 & 20.833 & 384.118 & 0 & 0.000 & 80.161 & 900.171 \\
		& 8 & 0 & \textbf{97.500} & 2.500 & 350.293 & 0 & 0.000 & 72.966 & 900.262 \\
		& 10 & 0 & \textbf{100.000} & \textbf{0.000} & 301.306 & 0 & 0.000 & 66.538 & 900.312 \\
		& 12 & 0 & 97.500 & 2.500 & 283.113 & 0 & 0.000 & 61.510 & 900.283 \\
		\midrule
		$2/n$ & 3 & 33 & 43.333 & 40.230 & 607.371 & 0 & 0.000 & 86.488 & 900.243 \\
		& 5 & 32 & 49.167 & 32.955 & 602.887 & 0 & 0.000 & 78.730 & 900.219 \\
		& 8 & 27 & \textbf{65.833} & 15.054 & 593.271 & 0 & 0.000 & 69.065 & 900.301 \\
		& 10 & 31 & \textbf{65.000} & 12.360 & 542.184 & 0 & 0.000 & 63.670 & 900.319 \\
		& 12 & 31 & 66.667 & 10.112 & 523.125 & 0 & 0.000 & 57.108 & 900.345 \\
		\midrule
		$\ln(n)/n$ & 3 & 40 & 33.333 & 50.000 & 629.753 & 0 & 0.000 & 87.668 & 900.259 \\
		& 5 & 40 & 33.333 & 50.000 & 629.769 & 0 & 0.000 & 79.912 & 900.266 \\
		& 8 & 40 & 34.167 & 48.750 & 630.058 & 0 & 0.000 & 69.132 & 900.249 \\
		& 10 & 40 & 35.000 & 47.500 & 631.444 & 0 & 0.000 & 62.193 & 900.272 \\
		& 12 & 40 & 36.667 & 45.000 & 626.889 & 0 & 0.000 & 55.777 & 900.231 \\
		\bottomrule
	\end{tabular*}
	\par\smallskip
	\begin{minipage}{\linewidth}\fontsize{8}{9.5}\selectfont
		$a$: Number of instances where Gurobi could not find a lower bound within the time limit (for LPs, no optimal LP result). \par 
	$b$ (\%): Number of times in percentage a given lower bound yields the best bound. \par
	$c$ (\%): Average percentage deviation from the best lower bound over instances with a valid $LB$. \par
	$d$ (s): Average CPU time in seconds.
	\end{minipage}
\end{table}

\begin{table}[p]
	\centering
	\fontsize{8}{9.5}\selectfont
	\setlength{\tabcolsep}{0pt}
	\renewcommand{\arraystretch}{0.94}
	\caption{Comparison of eight lower bounds on the CD and LD instances with $n=150$. Each $(p,m)$ contains 120 instances.}
	\label{results:LB150CDLOG}
	\noindent\begin{tabular*}{\linewidth}{@{\extracolsep{\fill}}lrrrrrrrrrrrrr@{}}
		\toprule
		& & \multicolumn{3}{c}{$LB_{1}$} & \multicolumn{3}{c}{$LB_{2}$} & \multicolumn{3}{c}{$LB_{3}$} & \multicolumn{3}{c}{$LB_{4}$} \\
		\cmidrule(lr){3-5} \cmidrule(lr){6-8} \cmidrule(lr){9-11} \cmidrule(lr){12-14}
		$p$ & $m$ & $b$ & $c$ & $d$ & $b$ & $c$ & $d$ & $b$ & $c$ & $d$ & $b$ & $c$ & $d$ \\ \midrule
		$0.5/n$ & 3 & \textbf{100.000} & \textbf{0.000} & 0.000 & 0.000 & 99.889 & 0.000 & 0.000 & 99.882 & 0.000 & 0.000 & 99.871 & 0.000 \\
		& 5 & \textbf{100.000} & \textbf{0.000} & 0.000 & 0.000 & 99.817 & 0.000 & 0.000 & 99.806 & 0.000 & 0.000 & 99.790 & 0.000 \\
		& 8 & \textbf{93.333} & 0.000 & 0.000 & 0.000 & 99.715 & 0.000 & 0.000 & 99.696 & 0.000 & 0.000 & 99.675 & 0.000 \\
		& 10 & 68.333 & 0.002 & 0.000 & 0.000 & 99.634 & 0.000 & 0.000 & 99.615 & 0.000 & 0.000 & 99.596 & 0.000 \\
		& 12 & 63.333 & 0.004 & 0.000 & 0.000 & 99.574 & 0.000 & 0.000 & 99.556 & 0.000 & 0.000 & 99.531 & 0.000 \\
		\midrule
		$2/n$ & 3 & \textbf{100.000} & \textbf{0.000} & 0.000 & 0.000 & 99.855 & 0.000 & 0.000 & 99.850 & 0.000 & 0.000 & 99.845 & 0.000 \\
		& 5 & \textbf{99.167} & 0.000 & 0.000 & 0.000 & 99.766 & 0.000 & 0.000 & 99.758 & 0.000 & 0.000 & 99.748 & 0.000 \\
		& 8 & \textbf{97.500} & 0.000 & 0.000 & 0.000 & 99.639 & 0.000 & 0.000 & 99.625 & 0.000 & 0.000 & 99.605 & 0.000 \\
		& 10 & \textbf{92.500} & 0.001 & 0.000 & 0.000 & 99.551 & 0.000 & 0.000 & 99.528 & 0.000 & 0.000 & 99.518 & 0.000 \\
		& 12 & \textbf{80.833} & 0.004 & 0.000 & 0.000 & 99.464 & 0.000 & 0.000 & 99.443 & 0.000 & 0.000 & 99.428 & 0.000 \\
		\midrule
		$\ln(n)/n$ & 3 & \textbf{100.000} & \textbf{0.000} & 0.000 & 0.000 & 99.812 & 0.000 & 0.000 & 99.809 & 0.000 & 0.000 & 99.817 & 0.000 \\
		& 5 & \textbf{100.000} & \textbf{0.000} & 0.000 & 0.000 & 99.689 & 0.000 & 0.000 & 99.678 & 0.000 & 0.000 & 99.682 & 0.000 \\
		& 8 & \textbf{94.167} & 0.001 & 0.000 & 0.000 & 99.508 & 0.000 & 0.000 & 99.519 & 0.000 & 0.000 & 99.511 & 0.000 \\
		& 10 & \textbf{79.167} & 0.004 & 0.000 & 0.000 & 99.409 & 0.000 & 0.000 & 99.400 & 0.000 & 0.000 & 99.443 & 0.000 \\
		& 12 & \textbf{70.833} & 0.011 & 0.000 & 0.000 & 99.308 & 0.000 & 0.000 & 99.302 & 0.000 & 0.000 & 99.289 & 0.000 \\
	\end{tabular*}
	\par\medskip
	\noindent\begin{tabular*}{\linewidth}{@{\extracolsep{\fill}}lrrrrrrrrr@{}}
		\toprule
		& & \multicolumn{4}{c}{$LB_{5}$} & \multicolumn{4}{c}{$LB_{7}$} \\
		\cmidrule(lr){3-6} \cmidrule(lr){7-10}
		$p$ & $m$ & $a$ & $b$ & $c$ & $d$ & $a$ & $b$ & $c$ & $d$ \\ \midrule
		$0.5/n$ & 3 & 80 & 33.333 & \textbf{0.000} & 611.054 & 0 & 0.000 & 95.138 & 0.174 \\
		& 5 & 77 & 35.833 & \textbf{0.000} & 603.688 & 0 & 0.000 & 92.028 & 0.300 \\
		& 8 & 73 & 39.167 & \textbf{0.000} & 598.258 & 0 & 0.000 & 87.565 & 0.204 \\
		& 10 & 33 & \textbf{72.500} & \textbf{0.000} & 518.333 & 0 & 0.000 & 84.733 & 0.177 \\
		& 12 & 14 & \textbf{88.333} & \textbf{0.000} & 446.410 & 0 & 0.000 & 81.941 & 0.157 \\
		\midrule
		$2/n$ & 3 & 80 & 33.333 & \textbf{0.000} & 631.359 & 0 & 0.000 & 95.138 & 0.292 \\
		& 5 & 80 & 33.333 & \textbf{0.000} & 625.362 & 0 & 0.000 & 92.028 & 0.366 \\
		& 8 & 80 & 33.333 & \textbf{0.000} & 616.265 & 0 & 0.000 & 87.565 & 0.286 \\
		& 10 & 80 & 33.333 & \textbf{0.000} & 609.470 & 0 & 0.000 & 84.733 & 0.248 \\
		& 12 & 80 & 33.333 & \textbf{0.000} & 612.374 & 0 & 0.000 & 81.941 & 0.246 \\
		\midrule
		$\ln(n)/n$ & 3 & 80 & 33.333 & \textbf{0.000} & 661.188 & 0 & 0.000 & 95.138 & 0.358 \\
		& 5 & 81 & 32.500 & \textbf{0.000} & 669.024 & 0 & 0.000 & 92.028 & 0.459 \\
		& 8 & 80 & 32.500 & \textbf{0.000} & 640.092 & 0 & 0.000 & 87.565 & 0.313 \\
		& 10 & 81 & 29.167 & \textbf{0.001} & 660.044 & 0 & 0.000 & 84.733 & 0.300 \\
		& 12 & 80 & 30.000 & \textbf{0.002} & 641.150 & 0 & 0.000 & 81.942 & 0.273 \\
	\end{tabular*}
	\par\medskip
	\noindent\begin{tabular*}{\linewidth}{@{\extracolsep{\fill}}lrrrrrrrrr@{}}
		\toprule
		& & \multicolumn{4}{c}{$LB'_{5}$} & \multicolumn{4}{c}{$LB'_{7}$} \\
		\cmidrule(lr){3-6} \cmidrule(lr){7-10}
		$p$ & $m$ & $a$ & $b$ & $c$ & $d$ & $a$ & $b$ & $c$ & $d$ \\ \midrule
		$0.5/n$ & 3 & 38 & 33.333 & 51.220 & 625.786 & 0 & 0.000 & 93.916 & 902.009 \\
		& 5 & 33 & 33.333 & 54.023 & 619.293 & 0 & 0.000 & 90.208 & 901.704 \\
		& 8 & 31 & 45.000 & 39.326 & 604.174 & 0 & 0.000 & 84.918 & 901.147 \\
		& 10 & 28 & 63.333 & 17.391 & 567.908 & 0 & 0.000 & 81.598 & 901.056 \\
		& 12 & 6 & 66.667 & 29.825 & 533.993 & 0 & 0.000 & 78.508 & 900.642 \\
		\midrule
		$2/n$ & 3 & 40 & 33.333 & 50.000 & 657.788 & 0 & 0.000 & 92.808 & 900.492 \\
		& 5 & 40 & 33.333 & 50.000 & 641.130 & 0 & 0.000 & 88.236 & 900.626 \\
		& 8 & 40 & 33.333 & 50.000 & 627.309 & 0 & 0.000 & 81.780 & 900.593 \\
		& 10 & 40 & 33.333 & 50.000 & 626.247 & 0 & 0.000 & 77.571 & 900.501 \\
		& 12 & 40 & 35.000 & 47.500 & 626.495 & 0 & 0.000 & 73.440 & 900.484 \\
		\midrule
		$\ln(n)/n$ & 3 & 73 & 33.333 & 14.894 & 689.624 & 0 & 0.000 & 91.796 & 900.253 \\
		& 5 & 77 & 33.333 & 6.977 & 672.929 & 0 & 0.000 & 86.486 & 900.217 \\
		& 8 & 62 & 33.333 & 31.034 & 663.883 & 0 & 0.000 & 78.828 & 900.276 \\
		& 10 & 58 & 33.333 & 35.484 & 657.697 & 0 & 0.000 & 74.021 & 900.247 \\
		& 12 & 52 & 33.333 & 41.176 & 665.203 & 0 & 0.000 & 69.194 & 900.263 \\
		\bottomrule
	\end{tabular*}
	\par\smallskip
	\begin{minipage}{\linewidth} \fontsize{8}{9.5}\selectfont
		$a$: Number of instances where Gurobi could not find a lower bound within the time limit (for LPs, no optimal LP result). \par 
		$b$ (\%): Number of times in percentage a given lower bound yields the best bound. \par
		$c$ (\%): Average percentage deviation from the best lower bound over instances with a valid $LB$. \par
		$d$ (s): Average CPU time in seconds.
	\end{minipage}
\end{table}

\subsection{Performance of the lower bounds}

Tables~\ref{results:LB20CDLOG}, \ref{results:LB50CDLOG},
\ref{results:LB100CDLOG}, and~\ref{results:LB150CDLOG}
show that, as for the FD instances with $p=0.2$, $LB_1$
substantially outperforms the weighted independent set-based
bounds. However, their relative ordering differs from the
main trend observed for the FD instances: $LB_4$ often
yields smaller average deviations than $LB_2$ and $LB_3$,
particularly for $p=0.5/n$, although all three remain weak.
Among the LP-relaxation bounds, $LB_5$ yields the smallest
average deviations when available, followed by $LB_6$
for $n=20$ and then $LB_7$, which remains the fastest
to compute.

The bounds derived from \mFA\ perform particularly well
for $n\in\{20,50\}$. The final solver bound $LB'_5$ attains
$BestLB$ on all instances with $n=20$ and on at least
90\% of the instances in every group with $n=50$.
Moreover, unlike the FD experiments with $n=50$, both
$LB_5$ and $LB'_5$ are available for all instances.
These results suggest that the lower conflict densities
postpone the computational limitations of \mFA\ to
larger instance sizes.

Nevertheless, these limitations appear at $n=100$ and
become more pronounced at $n=150$, as the time-indexed
formulation grows with the number of jobs and the time
horizon $T$. For $n=150$, $LB_5$ and $LB'_5$ are
unavailable for 1,079 and 658 of the 1,800 instances,
respectively. Missing bounds are particularly common
for the LD instances, whose higher expected density
implies more conflict constraints on average than for
the CD instances of the same size. In contrast, \mFC\
continues to provide both $LB_7$ and $LB'_7$ for every
instance, although their average deviations remain large.
The small average deviations of the available values
of $LB_5$ should therefore be interpreted alongside
the availability counts.

As in the FD experiments, the inexpensive bound $LB_1$
becomes increasingly competitive on larger instances.
Its average deviation from $BestLB$ does not exceed
0.084\% in any group with $n=100$, and  0.011\% for $n=150$. At the latter size,
it attains $BestLB$ on at least 63.333\% of the instances
in every group. 

\begin{table}[p]
	\centering
	\fontsize{8}{8.5}\selectfont
	\setlength{\tabcolsep}{0pt}
	\renewcommand{\arraystretch}{0.84}
	\caption{Comparison of the mathematical formulations on the CD and LD
		instances with $n\in\{20,50,100,150\}$.}
	\label{results:UBCDLOG}

	\noindent\begin{tabular*}{\linewidth}
		{@{\extracolsep{\fill}}lrrrrrrrrrrrrr@{}}
		\toprule
		& & \multicolumn{4}{c}{$\mathrm{F}_1$}
		& \multicolumn{4}{c}{$\mathrm{F}_2$}
		& \multicolumn{4}{c}{$\mathrm{F}_3$} \\
		\cmidrule(lr){3-6}
		\cmidrule(lr){7-10}
		\cmidrule(lr){11-14}
		$p$ & $m$
		& $a$ & $b$ & $c$ & $d$
		& $a$ & $b$ & $c$ & $d$
		& $a$ & $b$ & $c$ & $d$ \\
		\midrule
		
		$0.5/n$ & 3
		& \textbf{100.000} & \textbf{0.000} & \textbf{100.000} & \textbf{0.000}
		& 90.000 & 0.051 & 0.000 & 47.798
		& 99.167 & 0.004 & 0.000 & 43.258 \\
		& 5
		& \textbf{100.000} & \textbf{0.000} & \textbf{100.000} & \textbf{0.000}
		& 57.500 & 0.362 & 3.333 & 28.146
		& 91.667 & 0.007 & 0.000 & 22.149 \\
		& 8
		& \textbf{100.000} & \textbf{0.000} & \textbf{100.000} & \textbf{0.000}
		& 59.167 & 0.386 & 16.667 & 10.907
		& 98.333 & 0.001 & 9.167 & 5.775 \\
		& 10
		& \textbf{100.000} & \textbf{0.000} & \textbf{100.000} & \textbf{0.000}
		& 57.500 & 1.301 & 39.167 & 4.579
		& \textbf{100.000} & \textbf{0.000} & \textbf{100.000} & 0.000 \\
		& 12
		& \textbf{100.000} & \textbf{0.000} & \textbf{100.000} & \textbf{0.000}
		& 69.167 & 0.794 & 52.500 & 1.053
		& \textbf{100.000} & \textbf{0.000} & \textbf{100.000} & 0.000 \\
		
		\midrule
		
		$2/n$ & 3
		& \textbf{100.000} & \textbf{0.000} & \textbf{100.000} & \textbf{0.000}
		& 30.000 & 0.560 & 0.000 & 49.244
		& 62.500 & 0.069 & 0.000 & 39.450 \\
		& 5
		& \textbf{100.000} & \textbf{0.000} & \textbf{100.000} & \textbf{0.000}
		& 12.500 & 1.763 & 4.167 & 27.432
		& 58.333 & 0.097 & 0.833 & 16.637 \\
		& 8
		& \textbf{100.000} & \textbf{0.000} & \textbf{100.000} & \textbf{0.000}
		& 41.667 & 2.181 & 32.500 & 9.735
		& 94.167 & 0.005 & 59.167 & 1.877 \\
		& 10
		& \textbf{100.000} & \textbf{0.000} & \textbf{100.000} & \textbf{0.000}
		& 48.333 & 2.945 & 46.667 & 5.510
		& \textbf{100.000} & \textbf{0.000} & \textbf{100.000} & 0.000 \\
		& 12
		& \textbf{100.000} & \textbf{0.000} & \textbf{100.000} & \textbf{0.000}
		& 57.500 & 3.147 & 51.667 & 3.609
		& \textbf{100.000} & \textbf{0.000} & \textbf{100.000} & 0.000 \\
		
		\midrule
		
		$\ln(n)/n$ & 3
		& \textbf{100.000} & \textbf{0.000} & \textbf{100.000} & \textbf{0.000}
		& 23.333 & 0.756 & 0.000 & 49.247
		& 55.000 & 0.080 & 0.000 & 36.546 \\
		& 5
		& \textbf{100.000} & \textbf{0.000} & \textbf{100.000} & \textbf{0.000}
		& 10.000 & 2.869 & 4.167 & 27.783
		& 40.833 & 0.190 & 5.833 & 11.963 \\
		& 8
		& \textbf{100.000} & \textbf{0.000} & \textbf{100.000} & \textbf{0.000}
		& 36.667 & 3.592 & 32.500 & 16.310
		& 96.667 & 0.004 & 81.667 & 0.579 \\
		& 10
		& \textbf{100.000} & \textbf{0.000} & \textbf{100.000} & \textbf{0.000}
		& 39.167 & 4.346 & 38.333 & 13.455
		& \textbf{100.000} & \textbf{0.000} & \textbf{100.000} & 0.000 \\
		& 12
		& \textbf{100.000} & \textbf{0.000} & \textbf{100.000} & \textbf{0.000}
		& 45.833 & 4.861 & 39.167 & 12.557
		& \textbf{100.000} & \textbf{0.000} & \textbf{100.000} & 0.000 \\
		
	\end{tabular*}
	
	\par\medskip

	\noindent\begin{tabular*}{\linewidth}
		{@{\extracolsep{\fill}}rllrrrrrrrr@{}}
		\toprule
		& & & \multicolumn{4}{c}{$\mathrm{F}_1$}
		& \multicolumn{4}{c}{$\mathrm{F}_3$} \\
		\cmidrule(lr){4-7}
		\cmidrule(lr){8-11}
		$n$ & $p$ & $m$
		& $a$ & $b$ & $c$ & $d$
		& $a$ & $b$ & $c$ & $d$ \\
		\midrule

		50 & $0.5/n$ & 3
		& \textbf{98.333} & \textbf{0.000} & \textbf{100.000} & \textbf{0.000}
		& 92.500 & 0.002 & 0.000 & 73.160 \\
		& & 5
		& \textbf{98.333} & \textbf{0.000} & \textbf{100.000} & \textbf{0.000}
		& 85.000 & 0.009 & 0.000 & 58.848 \\
		& & 8
		& \textbf{98.333} & \textbf{0.000} & \textbf{100.000} & \textbf{0.000}
		& 61.667 & 0.023 & 0.000 & 42.202 \\
		& & 10
		& \textbf{98.333} & \textbf{0.000} & \textbf{100.000} & \textbf{0.000}
		& 47.500 & 0.071 & 0.000 & 32.806 \\
		& & 12
		& \textbf{97.500} & \textbf{0.000} & \textbf{100.000} & \textbf{0.000}
		& 62.500 & 0.047 & 0.000 & 25.403 \\
		
		\cmidrule(lr){2-11}
		
		& $2/n$ & 3
		& \textbf{98.333} & \textbf{0.000} & \textbf{98.333} & \textbf{0.834}
		& 71.667 & 0.013 & 0.000 & 73.072 \\
		& & 5
		& \textbf{96.667} & \textbf{0.000} & \textbf{100.000} & \textbf{0.001}
		& 39.167 & 0.044 & 0.000 & 58.238 \\
		& & 8
		& \textbf{99.167} & \textbf{0.000} & \textbf{100.000} & \textbf{0.000}
		& 13.333 & 0.151 & 0.000 & 41.317 \\
		& & 10
		& \textbf{100.000} & \textbf{0.000} & \textbf{100.000} & \textbf{0.001}
		& 1.667 & 0.317 & 0.000 & 31.802 \\
		& & 12
		& \textbf{100.000} & \textbf{0.000} & \textbf{100.000} & \textbf{0.001}
		& 5.000 & 0.371 & 0.000 & 23.395 \\
		
		\cmidrule(lr){2-11}
		
		& $\ln(n)/n$ & 3
		& \textbf{89.167} & \textbf{0.012} & \textbf{75.833} & \textbf{10.015}
		& 62.500 & 0.020 & 0.000 & 74.308 \\
		& & 5
		& \textbf{93.333} & \textbf{0.028} & \textbf{93.333} & \textbf{0.043}
		& 15.833 & 0.136 & 0.000 & 59.686 \\
		& & 8
		& \textbf{92.500} & \textbf{0.068} & \textbf{87.500} & \textbf{0.130}
		& 9.167 & 0.465 & 0.000 & 41.951 \\
		& & 10
		& \textbf{93.333} & \textbf{0.058} & \textbf{85.833} & \textbf{0.225}
		& 6.667 & 1.002 & 0.000 & 31.489 \\
		& & 12
		& \textbf{89.167} & \textbf{0.258} & \textbf{73.333} & \textbf{0.624}
		& 10.833 & 1.206 & 0.000 & 22.410 \\
		
		\midrule

		100 & $0.5/n$ & 3
		& \textbf{99.167} & \textbf{0.000} & \textbf{66.667} & \textbf{32.500}
		& 98.333 & 0.000 & 0.000 & 87.359 \\
		& & 5
		& \textbf{98.333} & \textbf{0.002} & \textbf{79.167} & \textbf{20.834}
		& 87.500 & 0.002 & 0.000 & 80.164 \\
		& & 8
		& \textbf{94.167} & \textbf{0.000} & \textbf{97.500} & \textbf{2.501}
		& 70.000 & 0.008 & 0.000 & 72.972 \\
		& & 10
		& \textbf{94.167} & \textbf{0.000} & \textbf{100.000} & \textbf{0.001}
		& 53.333 & 0.021 & 0.000 & 66.550 \\
		& & 12
		& \textbf{92.500} & \textbf{0.001} & \textbf{97.500} & \textbf{2.501}
		& 53.333 & 0.017 & 0.000 & 61.521 \\
		
		\cmidrule(lr){2-11}
		
		& $2/n$ & 3
		& 92.500 & 0.004 & \textbf{39.167} & \textbf{40.231}
		& \textbf{97.500} & \textbf{0.000} & 0.000 & 86.489 \\
		& & 5
		& 76.667 & 0.024 & \textbf{40.000} & \textbf{32.965}
		& \textbf{81.667} & \textbf{0.004} & 0.000 & 78.734 \\
		& & 8
		& \textbf{69.167} & 0.044 & \textbf{50.833} & \textbf{15.104}
		& 56.667 & \textbf{0.028} & 0.000 & 69.082 \\
		& & 10
		& \textbf{70.833} & \textbf{0.062} & \textbf{61.667} & \textbf{12.374}
		& 37.500 & 0.085 & 0.000 & 63.719 \\
		& & 12
		& \textbf{69.167} & \textbf{0.015} & \textbf{64.167} & \textbf{10.115}
		& 35.000 & 0.122 & 0.000 & 57.173 \\
		
		\cmidrule(lr){2-11}
		
		& $\ln(n)/n$ & 3
		& 86.667 & 0.007 & \textbf{33.333} & \textbf{50.001}
		& \textbf{92.500} & \textbf{0.001} & 0.000 & 87.671 \\
		& & 5
		& 54.167 & 0.068 & \textbf{33.333} & \textbf{50.000}
		& \textbf{71.667} & \textbf{0.013} & 0.000 & 79.924 \\
		& & 8
		& 40.000 & 0.219 & \textbf{33.333} & \textbf{48.755}
		& \textbf{66.667} & \textbf{0.075} & 0.000 & 69.214 \\
		& & 10
		& 34.167 & 0.386 & \textbf{33.333} & \textbf{47.526}
		& \textbf{65.833} & \textbf{0.149} & 0.000 & 62.417 \\
		& & 12
		& 34.167 & 0.398 & \textbf{34.167} & \textbf{45.076}
		& \textbf{65.833} & \textbf{0.238} & 0.000 & 56.152 \\
		
		\midrule

		150 & $0.5/n$ & 3
		& 98.333 & 0.000 & \textbf{33.333} & \textbf{51.220}
		& \textbf{100.000} & \textbf{0.000} & 0.000 & 93.917 \\
		& & 5
		& 96.667 & 0.001 & \textbf{33.333} & \textbf{54.023}
		& \textbf{97.500} & \textbf{0.000} & 0.000 & 90.210 \\
		& & 8
		& 83.333 & 0.006 & \textbf{41.667} & \textbf{39.328}
		& \textbf{95.833} & \textbf{0.001} & 0.000 & 84.921 \\
		& & 10
		& \textbf{84.167} & 0.010 & \textbf{54.167} & \textbf{17.400}
		& 78.333 & \textbf{0.006} & 0.000 & 81.603 \\
		& & 12
		& \textbf{90.833} & \textbf{0.003} & \textbf{64.167} & \textbf{29.826}
		& 67.500 & 0.008 & 0.000 & 78.514 \\
		
		\cmidrule(lr){2-11}
		
		& $2/n$ & 3
		& 97.500 & 0.001 & \textbf{33.333} & \textbf{50.000}
		& \textbf{100.000} & \textbf{0.000} & 0.000 & 92.809 \\
		& & 5
		& 81.667 & 0.006 & \textbf{33.333} & \textbf{50.001}
		& \textbf{91.667} & \textbf{0.001} & 0.000 & 88.238 \\
		& & 8
		& 54.167 & 0.051 & \textbf{33.333} & \textbf{50.001}
		& \textbf{80.833} & \textbf{0.006} & 0.000 & 81.786 \\
		& & 10
		& 50.000 & 0.081 & \textbf{33.333} & \textbf{50.000}
		& \textbf{69.167} & \textbf{0.017} & 0.000 & 77.588 \\
		& & 12
		& 42.500 & 0.109 & \textbf{33.333} & \textbf{47.507}
		& \textbf{70.000} & \textbf{0.018} & 0.000 & 73.463 \\
		
		\cmidrule(lr){2-11}
		
		& $\ln(n)/n$ & 3
		& 90.000 & 0.003 & \textbf{33.333} & \textbf{14.895}
		& \textbf{98.333} & \textbf{0.000} & 0.000 & 91.797 \\
		& & 5
		& 63.333 & 0.028 & \textbf{33.333} & \textbf{6.979}
		& \textbf{82.500} & \textbf{0.003} & 0.000 & 86.489 \\
		& & 8
		& 45.833 & 0.122 & \textbf{33.333} & \textbf{31.035}
		& \textbf{66.667} & \textbf{0.021} & 0.000 & 78.846 \\
		& & 10
		& 36.667 & 0.215 & \textbf{33.333} & \textbf{35.484}
		& \textbf{66.667} & \textbf{0.046} & 0.000 & 74.068 \\
		& & 12
		& 33.333 & 0.331 & \textbf{33.333} & \textbf{41.177}
		& \textbf{66.667} & \textbf{0.066} & 0.000 & 69.269 \\
		
		\bottomrule
	\end{tabular*}
	
	\par\smallskip
	\begin{minipage}{\linewidth} \fontsize{8}{9.5}\selectfont
		$a$ (\%): Percentage of instances for which a given upper bound
		attains the best bound.\par
		$b$ (\%): Average percentage deviation from the best upper bound.\par
		$c$ (\%): Percentage of instances for which Gurobi terminated
		with status ``optimal solution found''.\par
		$d$ (\%): Average optimality gap reported by Gurobi,
		computed over instances with available gaps.
	\end{minipage}
\end{table}

\begin{table}[!htbp]
	\centering
	\small
	\setlength{\tabcolsep}{0pt}
	\renewcommand{\arraystretch}{1.10}
	\caption{Comparison of the best available lower bounds and the best upper bounds obtained from the mathematical formulations on the CD and LD instances.}
	\label{results:BestLBBestUBCDLOG}
	\smallskip
	\begin{tabular*}{\linewidth}{@{\extracolsep{\fill}}lrrrrrrr@{}}
		\toprule
		& & \multicolumn{2}{c}{$p=0.5/n$} & \multicolumn{2}{c}{$p=2/n$} & \multicolumn{2}{c}{$p=\ln(n)/n$} \\
		\cmidrule(lr){3-4}\cmidrule(lr){5-6}\cmidrule(lr){7-8}
		$n$ & $m$ & $a$ (\%) & $b$ (\%) & $a$ (\%) & $b$ (\%) & $a$ (\%) & $b$ (\%) \\
		\midrule
		20 & 3 & 100.000 & 0.000 & 100.000 & 0.000 & 100.000 & 0.000 \\
		& 5 & 100.000 & 0.000 & 100.000 & 0.000 & 100.000 & 0.000 \\
		& 8 & 100.000 & 0.000 & 100.000 & 0.000 & 100.000 & 0.000 \\
		& 10 & 100.000 & 0.000 & 100.000 & 0.000 & 100.000 & 0.000 \\
		& 12 & 100.000 & 0.000 & 100.000 & 0.000 & 100.000 & 0.000 \\
		\toprule
		50 & 3 & 96.667 & 0.000 & 89.167 & 0.000 & 68.333 & 0.004 \\
		& 5 & 94.167 & 0.000 & 90.000 & 0.001 & 83.333 & 0.015 \\
		& 8 & 98.333 & 0.000 & 98.333 & 0.000 & 75.000 & 0.063 \\
		& 10 & 99.167 & 0.000 & 91.667 & 0.001 & 65.000 & 0.169 \\
		& 12 & 96.667 & 0.000 & 92.500 & 0.001 & 55.833 & 0.380 \\
		\toprule
		100 & 3 & 91.667 & 0.000 & 66.667 & 0.001 & 44.167 & 0.004 \\
		& 5 & 84.167 & 0.001 & 47.500 & 0.006 & 34.167 & 0.029 \\
		& 8 & 85.000 & 0.000 & 55.833 & 0.012 & 33.333 & 0.145 \\
		& 10 & 89.167 & 0.000 & 61.667 & 0.027 & 33.333 & 0.385 \\
		& 12 & 82.500 & 0.001 & 61.667 & 0.007 & 34.167 & 0.517 \\
		\toprule
		150 & 3 & 85.833 & 0.000 & 64.167 & 0.001 & 43.333 & 0.002 \\
		& 5 & 70.000 & 0.001 & 47.500 & 0.003 & 26.667 & 0.010 \\
		& 8 & 65.000 & 0.002 & 35.833 & 0.017 & 34.167 & 0.046 \\
		& 10 & 64.167 & 0.005 & 33.333 & 0.042 & 30.833 & 0.111 \\
		& 12 & 70.000 & 0.002 & 29.167 & 0.047 & 31.667 & 0.140 \\
		\midrule
		\multicolumn{2}{l}{Overall} & 88.625 & 0.001 & 73.250 & 0.008 & 59.667 & 0.101 \\
		\bottomrule
	\end{tabular*}
	\par\smallskip
		\begin{minipage}{\linewidth}\fontsize{8}{9.5}\selectfont
		$a(\%)$: percentage of instances for which $\mathrm{BestLB}=\mathrm{BestUB}$. \par
		$b(\%)$: mean percentage gap, $100(\mathrm{BestUB}-\mathrm{BestLB})/\mathrm{BestUB}$, over all instances in the group.
	\end{minipage}
\end{table}

\subsection{Performance of the mathematical formulations}
Table~\ref{results:UBCDLOG} shows that \mFA\ performs
particularly well for $n\in\{20,50\}$. For $n=20$,
Gurobi reports optimality for all instances using \mFA,
while \mFC\ also achieves this for all instances with
$m\in\{10,12\}$. In contrast, \mFB\ is outperformed by both formulations
and was therefore excluded from the subsequent experiments. For $n=50$, \mFA\ outperforms \mFC\
in both upper-bound quality and optimality certification
in every group. This contrasts with the FD results
at the same size for $m\in\{3,5,8\}$, suggesting that
lower conflict densities postpone the computational
limitations of \mFA.

As the instance size increases, \mFC\ becomes more
competitive in terms of upper-bound quality. For
$n=100$, \mFA\ generally provides better upper bounds
for $p=0.5/n$, the results are mixed for $p=2/n$,
and \mFC\ provides better upper bounds in every LD
group. At $n=150$, the advantage of \mFC\ extends
to every group with $p=2/n$, while the comparison
remains mixed for $p=0.5/n$. This shift is consistent
with the non-time-indexed structure of \mFC, whereas
the size of \mFA\ grows with both the number of jobs
and the time horizon $T$. The higher conflict density
of the LD instances makes them particularly demanding
for \mFA. Nevertheless, among the formulations evaluated
for $n\geq50$, only \mFA\ obtains an optimal solver
status for some instances.

Table~19 compares the best upper bounds obtained from
the mathematical formulations ($\mathrm{BestUB}$) with the best available
lower bounds ($\mathrm{BestLB}$). According to the reported results, 1{,}800, 1{,}553, 1{,}086, and 878 of
the 1{,}800 instances at each size are solved to optimality for $n=20,50,100,150$, respectively. The
corresponding average percentage gaps are approximately
0.000\%, 0.042\%, 0.076\%, and 0.029\%.
Overall, 5{,}317 of the 7{,}200 instances (73.847\%)
are solved to optimality, with an average percentage
gap of approximately 0.037\%. 

\begin{table}[p]
	\centering
	\fontsize{7}{8}\selectfont
	\setlength{\tabcolsep}{0pt}
	\renewcommand{\arraystretch}{0.85}
	\caption{Performance of GA-ls on the CD and LD instances with
		$n\in\{20,50,100,150\}$. Each $(n,p,m)$ contains 120 instances.}
	\label{results:GALSCDLOG}
	\begin{tabular*}{\linewidth}{@{\extracolsep{\fill}}rllrrrrrrrrrrrrr@{}}
		\toprule
		& & & \multicolumn{5}{c}{Best} & \multicolumn{8}{c}{Average} \\
		\cmidrule(lr){4-8}\cmidrule(lr){9-16}
		$n$ & $p$ & $m$ & $a$ & $b$ & $c$ & $d$ & $e$
		& $a$ & $b$ & $c$ & $d$ & $e$ & $f$ & $g$ & $h$ \\
		\midrule
		20 & $0.5/n$ & 3 & 100.000 & 0.000 & - & - & - & 100.000 & 0.000 & - & - & - & 223.0 & 1287.7 & 0.3 \\
		& & 5 & 100.000 & 0.000 & - & - & - & 100.000 & 0.000 & - & - & - & 213.7 & 3144.0 & 0.5 \\
		& & 8 & 99.167 & 0.000 & - & - & - & 98.333 & 0.000 & - & - & - & 166.4 & 1557.0 & 0.3 \\
		& & 10 & 100.000 & 0.000 & - & - & - & 100.000 & 0.000 & - & - & - & 179.9 & 738.2 & 0.2 \\
		& & 12 & 100.000 & 0.000 & - & - & - & 100.000 & 0.000 & - & - & - & 123.5 & 1.5 & 0.1 \\
		\cmidrule(lr){2-16}
		& $2/n$ & 3 & 98.333 & 0.000 & - & - & - & 95.000 & 0.000 & - & - & - & 237.8 & 11373.0 & 1.6 \\
		& & 5 & 97.500 & 0.000 & - & - & - & 95.000 & 0.000 & - & - & - & 233.5 & 16382.7 & 2.2 \\
		& & 8 & 87.500 & 0.001 & - & - & - & 87.500 & 0.001 & - & - & - & 222.0 & 11242.4 & 1.9 \\
		& & 10 & 87.500 & 0.001 & - & - & - & 87.500 & 0.001 & - & - & - & 188.6 & 7217.6 & 1.5 \\
		& & 12 & 89.167 & 0.001 & - & - & - & 89.167 & 0.001 & - & - & - & 84.3 & 8333.5 & 1.2 \\
		\cmidrule(lr){2-16}
		& $\ln(n)/n$ & 3 & 100.000 & 0.000 & - & - & - & 92.500 & 0.000 & - & - & - & 248.1 & 13610.7 & 1.9 \\
		& & 5 & 96.667 & 0.000 & - & - & - & 86.667 & 0.000 & - & - & - & 235.7 & 23923.3 & 3.3 \\
		& & 8 & 71.667 & 0.002 & - & - & - & 71.667 & 0.002 & - & - & - & 226.4 & 16860.4 & 3.0 \\
		& & 10 & 75.000 & 0.002 & - & - & - & 75.000 & 0.002 & - & - & - & 152.7 & 13758.9 & 2.4 \\
		& & 12 & 65.833 & 0.005 & - & - & - & 65.833 & 0.005 & - & - & - & 130.4 & 20002.3 & 3.2 \\
		\midrule
		50 & $0.5/n$ & 3 & 100.000 & 0.000 & 100.000 & 0.000 & 0.000 & 100.000 & 0.000 & 100.000 & 0.000 & 0.000 & 300.0 & 4511.8 & 0.5 \\
		& & 5 & 100.000 & 0.000 & 100.000 & 0.000 & 0.000 & 99.115 & 0.000 & 100.000 & 0.000 & 0.000 & 300.0 & 10331.4 & 1.3 \\
		& & 8 & 100.000 & 0.000 & 100.000 & 0.000 & 0.000 & 99.153 & 0.000 & 100.000 & 0.000 & 0.000 & 272.4 & 16799.0 & 2.3 \\
		& & 10 & 100.000 & 0.000 & 100.000 & 0.000 & 0.000 & 93.277 & 0.000 & 100.000 & 0.000 & 0.000 & 263.6 & 18750.5 & 2.5 \\
		& & 12 & 100.000 & 0.000 & 100.000 & 0.000 & 0.000 & 94.828 & 0.000 & 100.000 & 0.000 & 0.000 & 243.1 & 18177.8 & 2.6 \\
		\cmidrule(lr){2-16}
		& $2/n$ & 3 & 100.000 & 0.000 & 92.308 & 0.000 & 0.000 & 95.327 & 0.000 & 92.308 & 0.000 & 0.000 & 300.0 & 20893.6 & 2.6 \\
		& & 5 & 99.074 & 0.000 & 100.000 & 0.000 & 0.000 & 94.444 & 0.000 & 100.000 & 0.000 & 0.000 & 300.0 & 39810.4 & 5.1 \\
		& & 8 & 97.458 & 0.000 & 50.000 & 0.000 & 0.000 & 80.508 & 0.000 & 50.000 & 0.000 & 0.000 & 280.2 & 47971.5 & 7.4 \\
		& & 10 & 90.000 & 0.000 & 80.000 & 0.000 & 0.000 & 54.545 & 0.000 & 10.000 & 0.000 & 0.000 & 265.9 & 71031.9 & 13.3 \\
		& & 12 & 90.991 & 0.000 & 66.667 & 0.000 & 0.000 & 53.153 & 0.000 & 22.222 & 0.000 & 0.000 & 256.4 & 70493.4 & 14.0 \\
		\cmidrule(lr){2-16}
		& $\ln(n)/n$ & 3 & 98.780 & 0.000 & 92.105 & 0.000 & 0.000 & 96.341 & 0.000 & 89.474 & 0.000 & 0.000 & 300.0 & 34931.9 & 4.5 \\
		& & 5 & 93.000 & 0.000 & 55.000 & 0.000 & 0.001 & 70.000 & 0.000 & 45.000 & 0.000 & 0.001 & 300.0 & 59936.0 & 9.3 \\
		& & 8 & 66.667 & 0.000 & 13.333 & 0.000 & 0.002 & 20.000 & 0.000 & 6.667 & 0.001 & 0.002 & 281.3 & 83839.6 & 20.6 \\
		& & 10 & 35.897 & 0.001 & 7.143 & 0.001 & 0.003 & 5.128 & 0.001 & 0.000 & 0.002 & 0.003 & 268.5 & 103811.2 & 29.4 \\
		& & 12 & 34.328 & 0.001 & 0.000 & 0.004 & 0.004 & 8.955 & 0.002 & 0.000 & 0.005 & 0.004 & 262.1 & 105995.4 & 33.4 \\
		\midrule
		100 & $0.5/n$ & 3 & 100.000 & 0.000 & 100.000 & 0.000 & 0.000 & 99.091 & 0.000 & 70.000 & 0.000 & 0.000 & 300.0 & 4886.1 & 0.9 \\
		& & 5 & 100.000 & 0.000 & 100.000 & 0.000 & 0.000 & 97.030 & 0.000 & 94.737 & 0.000 & 0.000 & 300.0 & 12904.8 & 2.3 \\
		& & 8 & 100.000 & 0.000 & 100.000 & 0.000 & 0.000 & 96.078 & 0.000 & 88.889 & 0.000 & 0.000 & 300.0 & 16889.7 & 3.5 \\
		& & 10 & 100.000 & 0.000 & 92.308 & 0.000 & 0.000 & 92.523 & 0.000 & 84.615 & 0.000 & 0.000 & 300.0 & 21642.3 & 6.2 \\
		& & 12 & 100.000 & 0.000 & 100.000 & 0.000 & 0.000 & 97.980 & 0.000 & 100.000 & 0.000 & 0.000 & 300.0 & 19081.1 & 4.0 \\
		\cmidrule(lr){2-16}
		& $2/n$ & 3 & 100.000 & 0.000 & 77.500 & 0.000 & 0.000 & 100.000 & 0.000 & 67.500 & 0.000 & 0.000 & 300.0 & 16922.9 & 4.1 \\
		& & 5 & 100.000 & 0.000 & 88.889 & 0.000 & 0.000 & 94.737 & 0.000 & 52.381 & 0.000 & 0.000 & 300.0 & 38284.3 & 10.8 \\
		& & 8 & 95.522 & 0.000 & 66.038 & 0.000 & 0.000 & 80.597 & 0.000 & 43.396 & 0.000 & 0.000 & 300.0 & 54151.6 & 23.1 \\
		& & 10 & 89.189 & 0.000 & 65.217 & 0.000 & 0.001 & 56.757 & 0.000 & 52.174 & 0.000 & 0.001 & 300.0 & 68472.0 & 35.9 \\
		& & 12 & 87.838 & 0.000 & 71.739 & 0.000 & 0.000 & 52.703 & 0.000 & 63.043 & 0.000 & 0.000 & 300.0 & 67095.5 & 40.0 \\
		\cmidrule(lr){2-16}
		& $\ln(n)/n$ & 3 & 100.000 & 0.000 & 46.269 & 0.000 & 0.000 & 100.000 & 0.000 & 43.284 & 0.000 & 0.000 & 300.0 & 34971.6 & 12.3 \\
		& & 5 & 100.000 & 0.000 & 39.241 & 0.000 & 0.000 & 100.000 & 0.000 & 39.241 & 0.000 & 0.000 & 300.0 & 63715.0 & 24.8 \\
		& & 8 & 87.500 & 0.000 & 13.750 & 0.001 & 0.002 & 82.500 & 0.000 & 10.000 & 0.001 & 0.002 & 300.0 & 92489.5 & 56.9 \\
		& & 10 & 75.000 & 0.000 & 0.000 & 0.001 & 0.005 & 37.500 & 0.000 & 0.000 & 0.001 & 0.004 & 300.0 & 106075.9 & 88.6 \\
		& & 12 & 26.829 & 0.000 & 0.000 & 0.002 & 0.006 & 12.195 & 0.001 & 0.000 & 0.002 & 0.006 & 300.0 & 108008.1 & 115.2 \\
		\midrule
		150 & $0.5/n$ & 3 & 100.000 & 0.000 & 58.824 & 0.000 & 0.000 & 100.000 & 0.000 & 58.824 & 0.000 & 0.000 & 300.0 & 6721.9 & 3.7 \\
		& & 5 & 100.000 & 0.000 & 44.444 & 0.000 & 0.000 & 100.000 & 0.000 & 44.444 & 0.000 & 0.000 & 300.0 & 13267.0 & 15.4 \\
		& & 8 & 100.000 & 0.000 & 52.381 & 0.000 & 0.000 & 100.000 & 0.000 & 50.000 & 0.000 & 0.000 & 300.0 & 14002.4 & 24.4 \\
		& & 10 & 100.000 & 0.000 & 81.395 & 0.000 & 0.000 & 94.805 & 0.000 & 65.116 & 0.000 & 0.000 & 300.0 & 19155.9 & 22.2 \\
		& & 12 & 98.810 & 0.000 & 91.667 & 0.000 & 0.000 & 85.714 & 0.000 & 77.778 & 0.000 & 0.000 & 300.0 & 18179.0 & 22.9 \\
		\cmidrule(lr){2-16}
		& $2/n$ & 3 & 100.000 & 0.000 & 48.837 & 0.000 & 0.000 & 100.000 & 0.000 & 48.837 & 0.000 & 0.000 & 300.0 & 16702.6 & 11.6 \\
		& & 5 & 100.000 & 0.000 & 42.857 & 0.000 & 0.000 & 100.000 & 0.000 & 42.857 & 0.000 & 0.000 & 300.0 & 25263.6 & 29.3 \\
		& & 8 & 100.000 & 0.000 & 45.455 & 0.000 & 0.000 & 100.000 & 0.000 & 45.455 & 0.000 & 0.000 & 300.0 & 36914.7 & 53.7 \\
		& & 10 & 100.000 & 0.000 & 33.750 & 0.000 & 0.000 & 97.500 & 0.000 & 33.750 & 0.000 & 0.000 & 300.0 & 55022.3 & 86.2 \\
		& & 12 & 100.000 & 0.000 & 32.941 & 0.000 & 0.000 & 97.143 & 0.000 & 29.412 & 0.000 & 0.000 & 300.0 & 64566.0 & 115.0 \\
		\cmidrule(lr){2-16}
		& $\ln(n)/n$ & 3 & 100.000 & 0.000 & 50.000 & 0.000 & 0.000 & 100.000 & 0.000 & 50.000 & 0.000 & 0.000 & 300.0 & 25952.0 & 19.8 \\
		& & 5 & 100.000 & 0.000 & 54.545 & 0.000 & 0.000 & 100.000 & 0.000 & 54.545 & 0.000 & 0.000 & 300.0 & 49702.4 & 38.8 \\
		& & 8 & 100.000 & 0.000 & 32.911 & 0.000 & 0.000 & 100.000 & 0.000 & 32.911 & 0.000 & 0.000 & 300.0 & 75533.1 & 80.6 \\
		& & 10 & 97.297 & 0.000 & 13.253 & 0.000 & 0.001 & 89.189 & 0.000 & 13.253 & 0.000 & 0.001 & 300.0 & 97968.8 & 138.3 \\
		& & 12 & 76.316 & 0.000 & 6.098 & 0.001 & 0.001 & 55.263 & 0.000 & 2.439 & 0.001 & 0.001 & 300.0 & 104994.3 & 203.8 \\
		\bottomrule
	\end{tabular*}
	
	\par\smallskip
	\begin{minipage}{\linewidth}
\begin{tablenotes} \fontsize{8}{9.5}\selectfont
	\item $a$ (\%): Percentage of instances with proven optimality
	for which the GA TCT equals the optimal TCT.
	
	\item $b$: Average deviation of the GA TCT from the optimal
	TCT over instances with proven optimality.
	
	\item $c$ (\%): Percentage of instances without proven
	optimality for which the GA TCT equals the best lower bound.
	
	\item $d$: Average deviation of the GA TCT from the best
	lower bound over instances without proven optimality.
	
	\item $e$: Average deviation of the GA TCT from the best
	upper bound from the MILPs ($\mathrm{BestUB}$) over instances
	without proven optimality.
	
	\item $f$: Average population size.
	$g$: Average number of generations.
	$h$ (s): CPU time in seconds.
\end{tablenotes}
	\end{minipage}
\end{table}

\subsection{Performance of GA with local search}

Table~\ref{results:GALSCDLOG} presents the performance of GA-ls on the CD and LD instances. Following the same distinction as in Section~6.3, instances with
a proven optimal TCT are those for which the best lower bound ($BestLB$) coincides with the best upper bound obtained from the mathematical formulations ($BestUB$). The remaining instances are referred to as instances without a proven optimal TCT. This classification is based on the comparison in Table~\ref{results:BestLBBestUBCDLOG}. 

Table~\ref{results:GALSCDLOG} shows that GA-ls performs
particularly well on the CD instances, while the LD
instances are more challenging. Among the instances
with a proven optimal TCT, the best solutions over five
runs attain the optimal TCT for 2{,}125, 1{,}665, and
1{,}136 instances with $p=0.5/n$, $p=2/n$, and
$p=\ln(n)/n$, respectively. Overall, GA-ls recovers
the optimal TCT for 4{,}926 of the 5{,}317 instances
with a proven optimal TCT (92.646\%). As observed for the low-density
FD instances, the percentage of instances with a proven
optimal TCT for which GA-ls recovers the optimal value
generally decreases as the number of machines increases,
particularly for the LD instances. Nevertheless, the
average deviations of the best solutions
from the proven optimal TCT on instances with a proven optimal TCT do not exceed
0.005. Among the 1{,}883
instances without a proven optimal TCT, GA-ls reaches
$BestLB$ on 826 instances. On this subset of 1{,}883 instances,
its average relative deviation from $BestLB$ is
approximately 0.041\%, while the average
relative improvement over the best MILP upper bound
is approximately 0.099\%.
Retaining the best solutions from both approaches
therefore increases the number of instances with
certified optimal solutions from 5{,}317 (73.847\%)
to 6{,}143 (85.319\%). Computation times generally
increase with the number of jobs and the conflict
density, with group averages ranging from 0.1 to
203.8 seconds.

\end{document}